\documentclass[journal]{IEEEtran}

\usepackage{algorithm}
\usepackage{algorithmic}
\usepackage{amsmath,empheq}
\usepackage{amsthm}
\usepackage{amssymb}
\usepackage{multicol}
\usepackage{multirow}
\usepackage{graphicx}
\usepackage[mathscr]{euscript}
\usepackage{multirow}
\usepackage{centernot}
\usepackage{listliketab}
\usepackage{cutwin}
\usepackage{verbatim}
\usepackage{slashbox}
\usepackage{color}
 \usepackage{multirow}
 \usepackage{wrapfig}
 \usepackage{bm}
 
 \usepackage{framed}
 \usepackage{newfloat}
 \DeclareFloatingEnvironment[
     fileext=los,
     listname=List of Derivations,
     name=Derivation Box,
     placement=tbhp,
     within=section,
 ]{derivation}
 
 \usepackage{epstopdf}
 
 \usepackage{tikz}
 \usetikzlibrary{positioning}
 \usepackage{empheq}

\usepackage{array}
\newcolumntype{L}[1]{>{\raggedright\let\newline\\\arraybackslash\hspace{0pt}}m{#1}}
\newcolumntype{C}[1]{>{\centering\let\newline\\\arraybackslash\hspace{0pt}}m{#1}}
\newcolumntype{R}[1]{>{\raggedleft\let\newline\\\arraybackslash\hspace{0pt}}m{#1}}

\newtheorem{claim}{Claim}
\newtheorem{proposition}{Proposition}
\newtheorem{remark}{Remark}

\newtheorem{intuition}{Insight}

\DeclareMathOperator*{\argmax}{arg\,max}

\DeclareMathOperator*{\sgn}{sgn\,}

\newcommand\T{\rule{0pt}{2.7ex}}
\newcommand\B{\rule[-1.4ex]{5pt}{0pt}}

\makeatletter
\renewcommand{\maketag@@@}[1]{\hbox{\m@th\normalsize\normalfont#1}}%
\makeatother

\begin{document}


\title{Optimal Investment Strategies for Competing Camps in a Social Network: A Broad Framework
}
\author{Swapnil Dhamal,
        Walid Ben-Ameur, Tijani Chahed, and Eitan Altman
\thanks{S. Dhamal is with 
INRIA Sophia
Antipolis, France, working at Laboratoire Informatique d'Avignon; he was with Samovar, T\'el\'ecom SudParis, CNRS, Universit\'e Paris-Saclay, when most of this work was done.}
\thanks{
 W. Ben-Ameur and T. Chahed are with 
 Samovar, 
 T\'el\'ecom SudParis, CNRS, Universit\'e Paris-Saclay.}
 \thanks{
E. Altman is with INRIA Sophia Antipolis and Laboratoire Informatique d'Avignon.}
}

\maketitle


\begin{abstract}
We study the problem of optimally investing in nodes of a social network in a competitive setting, wherein two camps aim to drive the average opinion of the population in their own favor. Using a well-established model of opinion dynamics, we formulate the problem as a zero-sum game with its players being the two camps. We derive optimal investment strategies for both camps, and show that a random investment strategy is optimal when the underlying network follows a popular class of weight distributions. We study a broad framework, where we consider various well-motivated settings of the problem, namely, when the influence of a camp on a node is a concave function of its investment on that node, when a camp aims at maximizing competitor's investment or deviation from its desired investment, and when one of the camps has uncertain information about the values of the model parameters. We also study a Stackelberg variant of this game under common coupled constraints on the combined investments by the camps and derive their equilibrium strategies, and hence quantify the first-mover advantage. For a quantitative and illustrative study, we conduct simulations on real-world datasets and provide results and insights.
%
\end{abstract}



\vspace{-2mm}
\section{Introduction}
\label{sec:ODSN_intro}

Opinion dynamics is a natural phenomenon in a system of cognitive agents, and is a well-studied topic across several disciplines. 
It is highly relevant to applications such as elections, viral marketing, propagation of ideas and behaviors, etc. 
%
%
%
In this paper, we consider two competing camps who aim to maximize the adoption of their respective opinions in a social network. In particular, we consider a strict competition setting where the opinion value of one camp is denoted by $+1$ and that of the other camp by $-1$; we refer to these camps as good and bad camps respectively.
Opinion adoption by a population can be quantified in a variety of ways;
here we consider a well-accepted way, namely, 
the average or equivalently, the sum of  opinion values of the nodes in the network \cite{gionis2013opinion,grabisch2017strategic}.
Hence the good camp's objective would be to maximize this sum, while the bad camp would aim to minimize it.

The average or sum of opinion values of the nodes or individuals is of relevance in several applications.
In a fund collection scenario, for instance, the magnitude of the opinion value of an individual can be viewed as the amount of funds and its sign as the camp towards which he or she is willing to contribute. 
%
Another example is that of
 a group of sensors or reporting agents, who are assigned the job of reporting their individual measurements of a particular parameter or event; the resulting measurement would be obtained by averaging the individual values. In this case, two competitors may aim to manipulate the resulting average (one perhaps for a good cause of avoiding panic, and another for elevating it).

While the opinion values can be unbounded in the above examples,
there are scenarios which can be modeled aptly by bounded opinion values. In elections, for instance,  an individual can vote at most once.
Here one could view bounded opinion value of an individual as a proxy for the probability with which the individual would vote for a camp. 
For instance, an opinion value
of $v \in [-1,+1]$ could imply that the probability of voting for the good
camp is $(1+v)/2$
and that of voting for the bad camp is $(1-v)/2$.
Hence the good (or respectively bad) camp would want to maximize (or respectively minimize) the sum of opinion values, since this sum would indicate the expected number of votes in favor of the good camp.
%
%
Product adoption 
is another example where bounded opinion values are well justified; the opinion value of an individual would indicate its probability of purchasing the product from the company that corresponds to  good camp.

Social networks play a prime role in determining the opinions, preferences, behaviors, etc. of the constituent individuals \cite{easley2010networks}. There have been efforts to develop models which could determine how the individuals update their opinions based on the opinions of their connections, and hence study the dynamics of opinions in the network \cite{acemoglu2011opinion}.
With such an underlying model of opinion dynamics, a camp would aim to 
 maximize the adoption of its opinion in a social network, in presence of a competitor.
A camp could act on achieving this objective by strategically investing on selected individuals in a social network who could adopt its opinion; these individuals would in turn influence the opinions of their connections, who would then influence the opinions of their respective connections, and so on. Based on the underlying application, this investment could be in the form of money, free products or discounts, attention, convincing discussions, etc. 
Given that both camps have certain budget constraints, the strategy of the good camp hence comprises of how much to invest on each node in the network, so as to maximize the sum of opinion values of the nodes, while that of the bad camp comprises of how much to invest on each node, so as to minimize this sum.

This setup results in a game, and since we consider a strict competition setting with constraints such as budget (and other constraints as we shall encounter), the setup fits into the framework of constrained zero-sum games \cite{charnes1953constrained}.

\vspace{-2mm}
\subsection{Motivation}
\label{sec:motiv}

There have been studies   
 to identify influential nodes and the amounts to be invested on them,
 specific to analytically tractable models of opinion dynamics (such as DeGroot)
\cite{dubey2006competing,bimpikis2016competitive,grabisch2017strategic}.
Such studies are important to complement the empirical and experimental studies, since they provide more concrete results and rigorous reasonings behind them. However, most of the studies are based in a very preliminary setting and a limited framework. This paper aims to consider a broader framework by motivating and analyzing a variety of settings, which could open interesting future directions for a broader analytical study of opinion dynamics.

Throughout the paper,
we study settings wherein the investment per node by a camp could be unbounded or bounded. Bounded investments  could be viewed as discounts
which cannot exceed 100\%, attention capacity or time constraint
of a voter to receive convincing arguments, company
policy to limit the number of free samples that can be given
to a customer, government policy of limiting the monetary
investment by a camp on a voter, etc. 
As we will see, bounded investments in our model would result in bounded opinion values, which as explained earlier, could be transformed into probability of voting for a
party or adopting a product,
and hence the expected number of votes or sales in the favor of each camp.
We first study in \textbf{Section~\ref{sec:ODSN_probform}}, the cases of unbounded and bounded investment in a fundamental setting where a camp's influence on a node is a linear function of its  investment.




While the linear influence function is consistent with the well-established Friedkin-Johnsen model, the influence of a camp on a node might not increase linearly with the corresponding investment.
In fact, several social and economic settings follow law of {\em diminishing marginal returns}, which says that for higher investments, the marginal returns (influence in this context) are lower for a marginal increase in investment.
An example of this law is when we watch a particular product advertisement on television; as we watch the advertisement more number of times, its marginal influence on us tends to get lower.
A concave influence function naturally captures this law.
We study such an influence function in the settings of both unbounded and bounded investment per node, and relate it to the skewness of investment in optimal strategies as well as user perception of fairness.
We study this  in \textbf{Section \ref{sec:concave}}.









%
%
%

There  are scenarios where a camp may want to maximize the total investment of the competing camp, so as to upset the latter's broad budget allocation, which might lead to reduction in its available budget for future investments  or for other channels such as mass media advertisement. The latter may also be forced to implement unappealing actions such as increasing the product cost or seek further monetary sources in order to compensate for its investments.
Alternatively, the camps may have been instructed a desired investment strategy by a mediator such as government or a central authority, and deviating from this strategy would incur a penalty. 
For instance,
the mediator itself would have its own broader optimization problem (which could be for the benefit for it or the society), whose optimal solution would require the camps to devise their investment strategies in a particular desired way.
The mediator would then instruct the camps to follow the corresponding desired investment strategies, and
%
in case of violation, the mediator could impose a penalty so as to compensate for the suboptimal outcome of its own optimization problem.
For similar reasons as mentioned before, a camp may want to maximize the penalty incurred by the competing camp.
We study these settings which capture the adversarial behavior of a camp towards another camp, in \textbf{Section~\ref{sec:minmax_investment}}.

For all of the aforementioned settings, we show in this paper that
it does not matter whether the camps strategize simultaneously or sequentially. 
%
%
We use Nash equilibrium as
the equilibrium notion to analyze the game in these settings. 
However, there could be settings where a sequential play would be more natural than a simultaneous one, which would result in a Stackelberg game. 
The sequence may be determined by a mediator or central authority which, for example, may be responsible for giving permissions for campaigning or scheduling product advertisements to be presented to a node.
We use subgame perfect Nash equilibrium as
the equilibrium notion for the game in such sequential play settings. Moreover, since we are concerned with a zero-sum game, we express the equilibrium in terms of maxmin or minmax. 
Assuming the good camp plays first (without loss of analytical generality), the bad camp would choose a strategy that minimizes the sum of opinion values as a best response to the good camp's strategy. Knowing this, the good camp would want to maximize this minimum value. 
We motivate two such settings.

It would often be the case that the total attention capacity of a node or the time it could allot for receiving campaigning from both camps combined, is bounded. 
%
This leads us to study the game under common coupled constraints (CCC) that the sum of investments by the  camps on any node is bounded.
These are called common coupled constraints since the constraints of one camp are satisfied if and only if the constraints of the other camp are satisfied, for every strategy profile. 
We study this setting in \textbf{Section \ref{sec:CCC}}.

Another sequential setting is one that results in uncertainty of information, where
the good camp (which plays first) may not have exact information regarding the network parameters. However, the bad camp (which plays second) would have  perfect information regarding these parameters, which are either revealed over time or deduced based on the effect of the good camp's investment. Forecasting the optimal strategy of  bad camp, we derive a robust strategy for the good camp 
which would give it a good payoff even in the worst case.
We study this setting in \textbf{Section \ref{sec:robustness}}.

It can be noted that the common coupled constraints setting captures the first mover advantage, while the uncertainty setting captures the first mover disadvantage.




%
\vspace{-3mm}
\subsection{Related Work}


A principal part of opinion dynamics in a population is how nodes update their opinions over time. 
One of the most well-accepted and well-studied approaches of updating a node's opinion is based on imitation, where each node adopts the opinion of some of its neighbors with a certain probability. One such well-established variant is  DeGroot model \cite{degroot1974reaching} where each node updates its opinion using a weighted convex combination of its neighbors' opinions. 
The model developed by Friedkin and Johnsen \cite{friedkin1990social,friedkin1997social} considers that, in addition to its neighbors' opinions, a node also gives certain weightage to its initial biased opinion.
%
%
%

Acemoglu and Ozdaglar \cite{acemoglu2011opinion} review several other models of opinion dynamics.
Lorenz \cite{lorenz2007continuous} surveys  modeling frameworks concerning continuous opinion dynamics under bounded confidence, wherein nodes pay more attention to beliefs that do not differ too much from their own.
Xia, Wang, and Xuan \cite{xia2013opinion} give a multidisciplinary review of the field of opinion dynamics as a combination of the social processes which are conventionally studied in social sciences, and the analytical and computational tools developed in mathematics, physics and complex system studies.
Das, Gollapudi, and Munagala \cite{das2014modeling} show that the widely studied theoretical models of opinion dynamics
do not explain their experimental observations, and hence propose a new model
as a combination of the DeGroot model and the Voter model \cite{clifford1973model,holley1975ergodic}. 
Parsegov et al. \cite{parsegov2017novel} develop a multidimensional extension of  Friedkin-Johnsen model, describing the evolution of the nodes' opinions on several interdependent topics, and analyze its convergence.

Ghaderi and Srikant \cite{ghaderi2013opinion} consider a setting where a node iteratively updates its opinion as a myopic best response to the opinions of its own and its neighbors, and hence study how the equilibrium and convergence to it depend on the network structure, initial opinions of the nodes, the location of stubborn agents (forceful nodes with unchanging opinions) and the extent of their stubbornness. 
Ben-Ameur, Bianchi, and Jakubowicz
\cite{ben2016robust}
analyze the convergence of some widespread gossip algorithms in the presence of stubborn agents and show that the network is driven to a state which exclusively depends on the stubborn agents. 
%
Jia et al. \cite{jia2015opinion} propose an empirical model 
combining the DeGroot and Friedkin models, and
hence study the evolution of self-appraisal, social power, and interpersonal influences for a group of nodes who discuss and form opinions. 
%
%
Halu et al. \cite{halu2013connect} consider the case of two interacting social networks, and hence study the case of political elections using simulations.

%

%
%

Yildiz, Ozdaglar, and Acemoglu \cite{yildiz2013binary} study the problem of optimal placement of stubborn agents in the discrete binary opinions setting with the objective of maximizing influence, given the location of  competing stubborn agents.
%
Gionis, Terzi, and Tsaparas \cite{gionis2013opinion} study from an algorithmic and experimental perspective, the problem of identifying a set of target nodes whose positive opinions about an information item would maximize the overall positive opinion for the item in the  network. 
Ballester, Calv{\'o}-Armengol, and Zenou
\cite{ballester2006s}
study optimal targeting by analyzing a noncooperative network game with local payoff complementarities.
%
Sobehy et al. \cite{sobehy2016elections} propose strategies to win an election using a Mixed Integer Linear Programming approach.


%

The basic model we study is similar to that considered by Grabisch et al. \cite{grabisch2017strategic}, that is, a zero-sum game with two camps holding distinct binary opinion values, aiming to select a set of nodes to invest on, so as to influence the average opinion that eventually emerges in the network. Their study, however, considers non-negative matrices and focuses on the existence and the characterization of equilibria in a preliminary setting, where the influence and cost functions are linear, camps have network information with certainty, and there is no bound on combined investment by the camps per node.
Dubey, Garg, and De Meyer \cite{dubey2006competing} study existence and uniqueness of Nash equilibrium, while also considering convex cost functions. The study, however, does not consider the possibility of bounded investment on a node, and the implications on the extent of skewness of investment and user perception of fairness owing to the convexity of cost functions.
Bimpikis, Ozdaglar, and Yildiz \cite{bimpikis2016competitive} provide a sharp characterization of the optimal targeted advertizing strategies and highlight their dependence on the underlying social network structure, in a preliminary setting. Their study emphasizes the effect of absoption centrality, which is encountered in our study as well.

The problem of maximizing information diffusion in social networks under popular models such as Independent Cascade and Linear Threshold, has been extensively studied \cite{easley2010networks, kempe2003maximizing, guille2013information}.
%
%
The competitive setting has resulted in several game theoretic studies of this problem  \cite{bharathi2007competitive,goyal2014competitive,etesami2016complexity}.
%
There have been preliminary studies addressing interaction among different informations, where the spread of one information influences the spread of the others \cite{myers2012clash,su2016understanding}.

There have been studies on games with constraints.
%
%
A notable study by Rosen \cite{rosen1965existence} shows existence of equilibrium in a constrained game, and its uniqueness in a strictly concave game.
Altman and Solan \cite{altman2009constrained} study constrained games, where the strategy set available to a player depends on the choice of strategies made by other players.  The authors show that, in constrained zero-sum games, the value of the game need not exist (that is, maxmin and minmax values need not be the same) and contrary to general functions,  maxmin value could be larger than  minmax.


The topic of decision under uncertainty has been of interest to the game theory and optimization communities. An established way of analyzing decision under uncertainty is using robust optimization tools. 
Ben-Tal, El Ghaoui, and Nemirovski
\cite{ben2009robust} present a thorough review of such tools.

\vspace{-3mm}
\subsection{Contributions of the Paper}

A primary goal of this work is to provide a broad framework for optimal investment strategies for competing camps in a social network, and propose and explore several aspects of the problem. 
In particular,
we study several well-motivated variants of a constrained zero-sum game where two competing camps aim to maximize the adoption of their respective opinions, under the well-established Friedkin-Johnsen model of opinion dynamics.
%
Following are
our specific contributions:

\begin{itemize}


\item 
We show that a random investment strategy is optimal when the underlying network follows a particular popular class of weight distributions. (Section \ref{sec:WC})

\item
We investigate when a camp's influence  on a node is a concave function of its investment on that node, for the cases of unbounded and bounded investment per node. We hence  provide implications for the skewness of optimal investment strategies and user perception of fairness. (Section \ref{sec:concave})

\item
We look at the complementary problem where a camp acts as an adversary to the competing camp by aiming to maximize the latter's investment. 
We also look at the  problem where a camp aims to maximize the deviation from the desired investment of the competing camp. (Section \ref{sec:minmax_investment})


\item
We study the Stackelberg variant under common coupled constraints, that the combined investment by the good and bad camps on any given node cannot exceed  a certain limit.
We study the maxmin and minmax values and present some interesting implications. (Section \ref{sec:CCC})

\item 
We analyze a setting where one of the camps would need to make decision under uncertainty. (Section \ref{sec:robustness})


\item
Using simulations, we illustrate our analytically derived results on real-world social networks, and present further insights based on our observations.
(Section \ref{sec:ODSN_sim})

%


\end{itemize}

\section{Model}
\label{sec:ODSN_prob}

Consider a social network with $N$ as its set of nodes and $E$ as its set of weighted, directed edges. 
Two competing camps (good and bad) aim to maximize the adoption of their respective opinions in the social network. We consider a strict competition setting where the opinion value of the good camp is denoted by $+1$ and that of the bad camp by $-1$. 
%
In this section,
we  present the parameters of the considered model of opinion dynamics, and the update rule along with its convergence result.
We first provide an introduction to the well-established Friedkin-Johnsen model, followed by our proposed extension.

\vspace{-2mm}
\subsection{Friedkin-Johnsen Model}
\label{sec:ODSN_model}



As per Friedkin-Johnsen model \cite{friedkin1990social,friedkin1997social},
prior to the process of opinion dynamics, every node holds a bias in opinion which could have been formed owing to various factors such as 
the node's fundamental views, its experiences,
past information from news and other sources, opinion dynamics in the past, etc. 
We denote this opinion bias of a node $i$ by $v_i^0$ and the weightage that the node attributes to it by $w_{ii}^0$.

The network effect is captured by how much a node is influenced by each of its friends or connections, that is, how much weightage is attributed by a node to the opinion of each of its connections. Let $v_j$ be the opinion held by node $j$ and $w_{ij}$ be the weightage attributed by node $i$ to the opinion of node $j$. The influence on node $i$ owing to node $j$ is given by $w_{ij}v_j$, thus the net influence on $i$ owing to all of its connections is $\sum_{j \in N} w_{ij}v_j$ (where $w_{ij} \neq 0$ only if $j$ is a connection of $i$). 
%
It is to be noted that we do not make any assumptions regarding the sign of the edge weights, that is, they could be negative as well (as justified in \cite{altafini2013consensus,proskurnikov2016opinion}). A negative edge weight $w_{ij}$ can be interpreted as some form of distrust that node $i$ holds on node $j$, that is, $i$ would be driven towards adopting an opinion that is opposite to that held or suggested by $j$. 

Since in Friedkin-Johnsen model,
each node updates its opinion using a weighted convex combination of its neighbors' opinions,
 the update rule   is given by

\vspace{-3mm}
\begin{small}
\begin{align*}
\forall i \in N :\;
v_i
\leftarrow
w_{ii}^0 v_i^0
+ \sum_{j \in N} w_{ij}v_j
\end{align*}
\end{small}
\vspace{-3mm}

\noindent
where

\vspace{-3mm}
\begin{small}
\begin{align*}
\forall i \in N :\;
|w_{ii}^0|+\sum_{j \in N} |w_{ij}|\leq 1
\end{align*}
\end{small}
\vspace{-3mm}


\vspace{-2mm}
\subsection{Our Extended Model}
\label{sec:extended_model}

We extend Friedkin-Johnsen model to incorporate the camps' investments and the weightage attributed by nodes to the camps' opinions.
The good and bad camps attempt to directly influence the nodes so that their opinions are driven towards being positive and negative, respectively. This direct influence depends on the investment or effort made by the camps, and also on how much a node weighs the camps' opinions. A given amount of investment may have different influence on different nodes based on how much these nodes weigh the camps' recommendations. We denote the investment made by the good and bad camps on node $i$ by $x_i$ and $y_i$ respectively, and the weightage that node $i$ attributes to them by $w_{ig}$ and $w_{ib}$ respectively. Since the influence of good camp on node $i$ would be an increasing function of both $x_i$ and $w_{ig}$, we assume the influence to be $w_{ig} x_i$ so as to maintain the multilinearity of Friedkin-Johnsen model. Similarly, $w_{ib} y_i$ is the influence of bad camp on node $i$. Also note that since the good and bad camps hold the opinions $+1$ and $-1$ respectively, the net influence owing to the direct recommendations from these camps is $(w_{ig} x_i - w_{ib} y_i)$.

The camps have budget constraints stating that the good camp can invest a total amount of $k_g$ across all the nodes, while the bad camp can invest a total amount of $k_b$.

\begin{table}[t]
\caption{Notation table}
\label{tab:notation}
\vspace{-5mm}
\begin{center}
\begin{tabular}{|c|p{7.5cm}|}
\hline
\T \B
$v_i^0$ & the initial biased opinion of node $i$
\\ \hline
\T \B
$w_{ii}^0$ & weightage given to the initial opinion by node $i$
\\ \hline
\T \B
$w_{ig}$ & weightage given by node $i$ to the good camp's opinion
\\ \hline
\T \B
$w_{ib}$ & weightage given by node $i$ to the bad camp's opinion
\\ \hline
\T \B
$w_{ij}$ & weightage given by node $i$ to the opinion of node $j$
\\ \hline
\T \B
$x_i$ & investment made by  good camp to directly influence node $i$
\\ \hline
\T \B
$y_i$ & investment made by  bad camp to directly influence node $i$
\\ \hline 
\T \B
$k_g$ & budget of the good camp
\\ \hline 
\T \B
$k_b$ & budget of the bad camp
\\ \hline 
\T \B
$v_i$ & the resulting opinion of node $i$
\\ \hline
\end{tabular}
\end{center}
\vspace{-5mm}
\end{table}


Table \ref{tab:notation} presents the required notation.
Consistent with the standard opinion dynamics models, we have the condition on the influence weights on any node $i$ that they sum to at maximum 1 (since a node updates its opinions using a weighted `convex' combination of the influencing factors).

\vspace{-3mm}
\begin{small}
\begin{align*}
\forall i \in N :\;
|w_{ii}^0|+\sum_{j\in N} |w_{ij}|+|w_{ig}|+|w_{ib}| \leq 1
\end{align*}
\end{small}
\vspace{-3mm}

\noindent
A standard assumption for guaranteeing convergence of the dynamics is 

\vspace{-5mm}
\begin{small}
\begin{align*}
\sum_{j\in N} |w_{ij}| < 1
\end{align*}
\end{small}
\vspace{-3mm}

\noindent
This  assumption is actually well suited for our model where we  would generally have non-zero weights attributed to the influence outside of the network, namely, the influence due to bias ($w_{ii}^0$) and campaigning ($w_{ig},w_{ib}$).

Nodes update their opinions in discrete time steps starting with time step 0.
With the aforementioned factors into consideration, each node $i$ updates its opinion at each step, using the following update rule
(an extension of the Friedkin-Johnsen update rule):

\vspace{-3mm}
\begin{small}
\begin{align*}
\forall i \in N :\;
v_i
\leftarrow
w_{ii}^0 v_i^0
+ \sum_{j \in N} w_{ij}v_j
+ w_{ig} x_i
- w_{ib} y_i
\end{align*}
\end{small}
\vspace{-3mm}

\noindent
 Let $v_i^{\langle\tau\rangle}$ be the opinion of node $i$ at time step $\tau$, and $v_i^{\langle 0 \rangle} = v_i^{0}$.
The update rule can hence be written as

\vspace{-3mm}
\begin{small}
\begin{align}
\forall i \in N :\;
v_i^{\langle\tau\rangle}
=
w_{ii}^0 v_i^0
+ \sum_{j \in N} w_{ij}v_j^{\langle\tau-1\rangle}
+ w_{ig} x_i
- w_{ib} y_i
\label{eqn:update_rule}
\end{align}
\end{small}
\vspace{-3mm}

\noindent
For any  node $i$, the static components are $x_i, y_i, v_i^0$ (weighed by $w_{ig}, w_{ib}, w_{ii}^0$), while the dynamic components are $v_j$'s (weighed by $w_{ij}$'s).
The static components remain unchanged while the dynamic ones get updated in every time step.

Let $\mathbf{w}$ be the matrix consisting of the elements $w_{ij}$ for each pair $(i,j)$ (note that $\mathbf{w}$ contains only the network weights and not $w_{ig}, w_{ib}, w_{ii}^0$).
Let $\mathbf{v}$ be the vector consisting of the opinions $v_i$, $\mathbf{v^0}$ and $\mathbf{w^0}$ be the vectors consisting of the elements $v_i^0$ and $w_{ii}^0$ respectively, $\mathbf{x}$ and $\mathbf{y}$ be the vectors consisting of the investments $x_i$ and $y_i$ respectively, $\mathbf{w_g}$ and $\mathbf{w_b}$ be the vectors consisting of the weights $w_{ig}$ and $w_{ib}$ respectively. 
Let the operation $\circ$ denote Hadamard product (elementwise product) of vectors, that is, $(\mathbf{a}\circ\mathbf{b})_{i} = {a}_i {b}_i$. Let Hadamard power be expressed as $(\mathbf{a}^{\circ p})_{i} = a_i^p$.

Assuming $\mathbf{v}^{\langle\tau\rangle}$ to be the vector consisting of the opinions $v_i^{\langle\tau\rangle}$,
the update rule (\ref{eqn:update_rule}) can be written in matrix form as

\vspace{-3mm}
\begin{small}
\begin{align}
\mathbf{v}^{\langle\tau\rangle} = \mathbf{w}\mathbf{v}^{\langle\tau-1\rangle} + \mathbf{w^0}\circ\mathbf{v^0} + \mathbf{w_g}\circ\mathbf{x} - \mathbf{w_b}\circ\mathbf{y}
\label{eqn:update_rule_matrix}
\end{align}
\end{small}
\vspace{-4mm}

%




\begin{proposition}
The dynamics defined by the update rule in (\ref{eqn:update_rule_matrix}) converges to
$
\mathbf{v} = (\mathbf{I}-\mathbf{w})^{-1} (\mathbf{w^0}\circ\mathbf{v^0} + \mathbf{w_g}\circ\mathbf{x} - \mathbf{w_b}\circ\mathbf{y})
$.
\label{prop:convergence}
\end{proposition}
\begin{proof}
The recursion in (\ref{eqn:update_rule_matrix}) can be simplified as
\begin{small}
\begin{align*}
\mathbf{v}^{\langle\tau\rangle} &= \mathbf{w}\mathbf{v}^{\langle\tau-1\rangle} + \mathbf{w^0}\circ\mathbf{v^0} + \mathbf{w_g}\circ\mathbf{x} - \mathbf{w_b}\circ\mathbf{y}
\\
&= \mathbf{w} \left( \mathbf{w}\mathbf{v}^{\langle\tau-2\rangle} + \mathbf{w^0}\circ\mathbf{v^0} + \mathbf{w_g}\circ\mathbf{x} - \mathbf{w_b}\circ\mathbf{y} \right)
\\ &\;\;\;\;\;+ \mathbf{w^0}\circ\mathbf{v^0} + \mathbf{w_g}\circ\mathbf{x} - \mathbf{w_b}\circ\mathbf{y}
\\
&= \mathbf{w}^2\mathbf{v}^{\langle\tau-2\rangle} + (\mathbf{I}+\mathbf{w})(\mathbf{w^0}\circ\mathbf{v^0} + \mathbf{w_g}\circ\mathbf{x} - \mathbf{w_b}\circ\mathbf{y})
\\
&= \mathbf{w}^3\mathbf{v}^{\langle\tau-3\rangle} + (\mathbf{I}+\mathbf{w}+\mathbf{w}^2)(\mathbf{w^0}\circ\mathbf{v^0} + \mathbf{w_g}\circ\mathbf{x} - \mathbf{w_b}\circ\mathbf{y})
\\
&= \mathbf{w}^\tau \mathbf{v}^{\langle 0 \rangle} + \left( \sum_{\eta=0}^{\tau-1}  \mathbf{w}^\eta \right) (\mathbf{w^0}\circ\mathbf{v^0} + \mathbf{w_g}\circ\mathbf{x} - \mathbf{w_b}\circ\mathbf{y})
\end{align*}
\end{small}
\vspace{-2mm}

Now, the initial opinion: $\mathbf{v}^{\langle 0 \rangle} = \mathbf{v}^{\mathbf{0}}$.
Also, $\mathbf{w}$ is a strictly substochastic matrix, 
since $\forall i \in N:\sum_{j\in N} |w_{ij}| < 1$;
 its spectral radius is hence less than 1.
So when $\tau \rightarrow \infty$, we have $\lim_{\tau \rightarrow \infty}\!\mathbf{w}^\tau\! = \mathbf{0}$.
Since $\mathbf{v}^0$ is a constant,
we have $\lim_{\tau \rightarrow \infty} \mathbf{w}^\tau \mathbf{v}^0 = \mathbf{0}$.
Furthermore,
$\lim_{\tau \rightarrow \infty} \sum_{\eta=0}^{\tau-1} \mathbf{w}^\eta = (\mathbf{I}-\mathbf{w})^{-1}$,
 an established matrix identity \cite{hubbard2015vector}. 
This  implicitly means that $(\mathbf{I}-\mathbf{w})$ is invertible.  
Hence,

\vspace{-3mm}
\begin{small}
\begin{align*}
\lim_{\tau \rightarrow \infty}
\mathbf{v}^{\langle\tau\rangle} = (\mathbf{I}-\mathbf{w})^{-1} (\mathbf{w^0}\circ\mathbf{v^0} + \mathbf{w_g}\circ\mathbf{x} - \mathbf{w_b}\circ\mathbf{y})
\end{align*}
\end{small}
\vspace{-3mm}

\noindent
which is a constant  vector, that is, the dynamics converges to this steady state of  opinion values. 
%
\end{proof}


\vspace{-2mm}
\section{{The Fundamental Problem}}
\label{sec:ODSN_probform}

We  now present the fundamental problem of competitive opinion dynamics under the Friedkin-Johnsen model.


\vspace{-2mm}
\subsection{Introduction of the Fundamental Problem}
\label{sec:maxmin_sumvalues}

The problem of maximizing opinion adoption can be modeled as an optimization problem. In particular, considering perfect competition, this problem can be modeled as a maxmin problem as we now present. Here our objective is to determine the strategies of the good and bad camps (the values of $x_i$ and $y_i$ such that they satisfy certain constraints), so that 
the good camp aims to maximize the sum of opinion values of the nodes while the bad camp aims to minimize it.
%
Considering linear constraints for setting the problem in the linear programming framework, we represent these constraints by $A\mathbf{x} \leq b$ and $C\mathbf{y} \leq d$, respectively, where $A,C$ are matrices and $b,d$ are vectors, in general. 

Owing to $x_i$ and $y_i$ being investments, we have the natural constraints: $x_i,y_i \geq 0, \forall i \in N$.
We can hence write the maxmin optimization problem in its general form as

\vspace{-3mm}
\begin{small}
\begin{align*}
\max_{\substack{A\mathbf{x} \leq b \\ \mathbf{x} \geq 0}} \;
 \min_{\substack{C\mathbf{y} \leq d \\ \mathbf{y} \geq 0}} \; \sum_{i\in N} v_i
 \end{align*}
 \vspace{-3mm}
 \begin{align*}
 \text{s.t. }
 \forall i\in N: \;
 v_i =
 w_{ii}^0 v_i^0 
+  \sum_{j \in N} w_{ij}v_j 
+ w_{ig} x_i
 - w_{ib} y_i
\end{align*}
\end{small}
\vspace{-2mm}

\noindent
From Proposition \ref{prop:convergence}, we have

 \vspace{-3mm}
 \begin{small}
 \begin{align*}
&
\mathbf{v} = (\mathbf{I}-\mathbf{w})^{-1} (\mathbf{w^0}\circ\mathbf{v^0} + \mathbf{w_g}\circ\mathbf{x} - \mathbf{w_b}\circ\mathbf{y})
\\
\implies &
\mathbf{1}^T \mathbf{v} = \mathbf{1}^T (\mathbf{I}-\mathbf{w})^{-1} (\mathbf{w^0}\circ\mathbf{v^0} + \mathbf{w_g}\circ\mathbf{x} - \mathbf{w_b}\circ\mathbf{y})
\end{align*}
\end{small}
\vspace{-3mm}

\noindent
Let the constant $\mathbf{1}^T (\mathbf{I}-\mathbf{w})^{-1} = \left( ((\mathbf{I}-\mathbf{w})^{-1})^T \mathbf{1}\right)^T$ $= \left( (\mathbf{I}-\mathbf{w}^T)^{-1} \mathbf{1}\right)^T = \mathbf{r}^T$.
So the above is equivalent to

 \vspace{-3mm}
 \begin{small}
 \begin{align}
\sum_{i \in N} v_i = \sum_{i \in N} r_i w_{ii}^0 v_i^0 +
\sum_{i \in N} r_i  w_{ig}x_i  - \sum_{i \in N} r_i w_{ib}y_i
\label{eqn:steadystate} 
\end{align}
\end{small}
\vspace{-3mm}

So our objective function becomes 

\vspace{-3mm}
\begin{small}
\begin{align*}
\max_{\substack{A\mathbf{x} \leq b \\ \mathbf{x} \geq 0}} \;
 \min_{\substack{C\mathbf{y} \leq d \\ \mathbf{y} \geq 0}} \; 
\sum_{i \in N} r_i w_{ii}^0 v_i^0 +
\sum_{i \in N} r_i  w_{ig}x_i  - \sum_{i \in N} r_i w_{ib}y_i
 \end{align*}
 \end{small}
 \vspace{-1mm}

\noindent
This can be achieved by solving two independent optimization problems, namely,

\vspace{-3mm}
\begin{small}
\begin{align*}
\max_{\substack{A\mathbf{x} \leq b \\ \mathbf{x} \geq 0}} \;
 \sum_{i \in N} r_i  w_{ig}x_i 
 \;\;\text{ and }\;\;
 \min_{\substack{C\mathbf{y} \leq d \\ \mathbf{y} \geq 0}} \;
  \sum_{i \in N} r_i w_{ib}y_i
 \end{align*}
 \end{small}
 \vspace{-1mm}
 
 \noindent
 which can be easily solved.

\subsubsection{The Specific Case: Overall Budget Constraints}
\label{sec:fundamental}

For studying the problem in a broader framework, we consider the case, specific to our model that we introduced in Section \ref{sec:ODSN_prob}.
%
This case that considers overall budget constraints $k_g$ and $k_b$ for the good and bad camps respectively, corresponds to
$\sum_{i \in N} x_i \leq k_g$ and $\sum_{i \in N} y_i \leq k_b$.
That is, we have $A=C=\mathbf{1}^T, b=k_g, d=k_b$.
%
%
It is clear that the solution to this specific optimization problem is

\vspace{-3mm}
\begin{small}
\begin{align}
x_i^* = 0 \;\; , \;\; \forall i \notin \argmax_{i \in N} r_i w_{ig}
\nonumber
\\
\sum_{\substack{i \in \argmax_i r_i w_{ig}\\ x_i^* \geq 0}} x_i^* = k_g  \;\; , \;\; \text{if } \max_{i \in N} r_i w_{ig} > 0
\label{eqn:xoptthis}
\end{align}
\vspace{-3mm}
\end{small}

%
%

~
\vspace{-5mm}
\begin{small}
\begin{align}
\text{ and }\;\;\;\;\;\;\;\;\;\;\;\;\;\;\;\;
y_i^* = 0 \;\; , \;\; \forall i \notin \argmax_{i \in N} r_i w_{ib}
\nonumber
\\
\sum_{\substack{i \in \argmax_i r_i w_{ib}\\ y_i^* \geq 0}} y_i^* = k_b
\;\; , \;\; \text{if } \max_{i \in N} r_i w_{ib} > 0
\label{eqn:yoptthis}
\end{align}
\vspace{-3mm}
\end{small}

\noindent
Note that if $\max_{i \in N} r_i w_{ig} \leq 0$, then
$
x_i^* = 0, \forall {i \in N}
$
and 
if $\max_{i \in N} r_i w_{ib} \leq 0$, then
$
y_i^* = 0, \forall {i \in N}
$.

Equations (\ref{eqn:xoptthis}) and (\ref{eqn:yoptthis}) lead to the following result.

\begin{proposition}
In Setting \ref{sec:fundamental}, it is optimal for the good and bad camps to  invest their entire budgets in node $i$ with maximum value of $r_i w_{ig}$ and $r_i w_{ib}$ respectively, subject to the value being positive.
\label{prop:fundamental_unbounded}
\end{proposition}

\begin{intuition}
Parameter $r_i$ could be interpreted as the influencing power of node $i$ on the network, while $w_{ig}$ and $w_{ib}$ are respectively the influencing powers of the good and bad camps on node $i$. So it is clear why these parameters factor into the result. Furthermore, the strategies of the camps are mutually independent, which arises from the sum of steady state values of nodes as derived in (\ref{eqn:steadystate}). 
The multilinearity of the model and unconstrained investment on nodes allow  the camps to exhaust their budgets by concentrating their entire investments on a node possessing the highest value of $r_i w_{ig}$ or $r_i w_{ib}$ respectively.
Also, the camps' strategies are independent of the initial opinions, since they aim to optimize the sum of  opinion values without considering their relative values.
\end{intuition}


Actually,  $r_i$ can be viewed as a variant of {\em Katz centrality} \cite{katz1953new} in that, Katz centrality of node $i$ measures its relative influence in a social network (say having adjacency matrix $\mathcal{A}$) with all edges having the same weight (say $\alpha$), while $r_i$ measures its influence in a general weighted social network.
Katz centrality of node $i$ is defined as the $i^{\text{th}}$ element of vector
$
\left( \left( \mathbf{I}-\alpha \mathcal{A}^T\right)^{-1} \! - \mathbf{I} \right) \mathbf{1}
=
 \left( \mathbf{I}-\alpha \mathcal{A}^T\right)^{-1} \mathbf{1}  - \mathbf{1}
$,
for $0\!<\!\alpha \!<\! \frac{1}{|\rho|}$ where $\rho$ is the largest eigenvalue of $\mathcal{A}$.
In our case where 
$
\mathbf{r}= \left( \mathbf{I}- \mathbf{w}^T\right)^{-1} \mathbf{1}
$,
 $\mathcal{A}$ is replaced by the weighted adjacency matrix $\mathbf{w}$, for which $|\rho|<1$ (since $\mathbf{w}$ is strictly substochastic), and we have $\alpha=1$. The subtraction of  vector $\mathbf{1}$ is common for all nodes, so its relative effect can be ignored.
$r_i$ can also be viewed as a variant of {\em absorption centrality} of node $i$ \cite{bimpikis2016competitive}, which captures the expected number of visits to node $i$ in a random walk starting at a node other than $i$ uniformly at random, with  transition probability matrix $\mathbf{w}$  (assuming all elements of $\mathbf{w}$ to be non-negative).


%




Furthermore, recall that 

\vspace{-3mm}
\begin{small}
\begin{align*}
\mathbf{r}^T = \mathbf{1}^T (\mathbf{I}-\mathbf{w})^{-1} 
= \mathbf{1}^T \Big( \mathbf{I} + \sum_{\eta=1}^\infty \mathbf{w}^\eta \Big)
\end{align*}
\end{small}
\vspace{-3mm}

\noindent
So if we have $w_{ij} \geq 0$ for all pairs of nodes $(i,j)$, we will have that all elements of vector $\mathbf{r}$ are at least 1. That is, $w_{ij} \geq 0, \forall (i,j) \implies r_i \geq 1, \forall {i \in N}$.

\subsubsection{The Case of Bounded Investment Per Node}
\label{sec:fundamental_bounded}


This setting, as motivated earlier,
includes an additional bound on the investment per node by a camp.
We assume this bound to be  1 unit without loss of generality, that is,
 $x_i,y_i \leq 1, \forall {i \in N}$.
With respect to the generic constraints $A\mathbf{x} \leq b$ and $C\mathbf{y} \leq d$, this case corresponds to
%
${A}=C=$ \begin{scriptsize}$\left( \begin{array}{c} \mathbf{1}^T \\ \mathbf{I} \\ \end{array} \right)$\end{scriptsize},
$b=$ \begin{scriptsize}$\left( \begin{array}{c} k_g \\ \mathbf{1} \\ \end{array} \right)$\end{scriptsize},
$d=$ \begin{scriptsize}$\left( \begin{array}{c} k_b \\ \mathbf{1} \\ \end{array} \right)$\end{scriptsize}.

From Equation (\ref{eqn:steadystate}), an optimal $\mathbf{x}$ can be obtained as follows.
Let $\mathbb{I}_{r_i w_{ig} > 0} = 1$ if $r_i w_{ig} > 0$, and 0 otherwise.
Let $\omega_1, \omega_2, \ldots, \omega_n$ be the ordering of nodes in decreasing values of $r_i w_{ig}$ with any tie-breaking rule. 
So (\ref{eqn:steadystate}) is maximized with respect to $\mathbf{x}$ when

\vspace{-3mm}
\begin{small}
\begin{align*}
x_i = 1 \cdot \mathbb{I}_{r_i w_{ig} > 0}, &\text{ for } i = \omega_1 , \ldots, \omega_{\lfloor k_g \rfloor}
\\
x_i = (k_g - \lfloor k_g \rfloor) \cdot \mathbb{I}_{r_i w_{ig} > 0}, &\text{ for } i = \omega_{\lfloor k_g \rfloor + 1}
\\
x_i = 0, &\text{ for } i = \omega_{\lfloor k_g \rfloor + 2}, \ldots, \omega_{n}
\end{align*}
\end{small}
%
An optimal $\mathbf{y}$ is analogous, hence the following result.

\begin{proposition}
In Setting \ref{sec:fundamental_bounded}, it is optimal for the good camp to invest in nodes one at a time, subject to a maximum investment of 1 unit per node,
in decreasing order of values of $r_i w_{ig}$
until either the budget $k_g$ is exhausted or we reach a node with a non-positive value of $r_i w_{ig}$.
The optimal strategy of the bad camp is analogous.
\label{prop:fundamental_bounded}
\end{proposition}

Also, from Proposition \ref{prop:convergence},
if a camp's investment per node is bounded by 1 unit, 
the opinion value of every node would be bounded between $-1$ and $+1$.
As stated earlier, such bounded opinion value is relevant to elections and product adoption scenarios, where the bounded opinion value of a node could be translated into the probability of the node voting for a camp or adopting a particular product.

%


\subsection{Maxmin versus Minmax Values}
\label{sec:maxmin_basic}

With no bounds on investment per node, it is clear that the maxmin and minmax values are the same, since the strategies of the camps are mutually independent, that is,
\begin{align}
\max_{\mathbf{x}\geq \mathbf{0} } \; \min_{\mathbf{y}\geq \mathbf{0} } \; \sum_{i \in N} v_i  = \min_{\mathbf{y}\geq \mathbf{0} } \; \max_{\mathbf{x}\geq \mathbf{0} } \; \sum_{i \in N} v_i
\label{eqn:maxmin_unC}
\end{align}
The equality would hold even with mutually independent bounds on the camps' investment on a node, that is, 
\begin{align}
\max_{\mathbf{0} \leq \mathbf{x}  \leq \mathbf{1}} \; \min_{\mathbf{0} \leq \mathbf{y} \leq \mathbf{1}} \; \sum_{i \in N} v_i  = \min_{\mathbf{0} \leq \mathbf{y} \leq \mathbf{1}} \; \max_{\mathbf{0} \leq \mathbf{x}  \leq \mathbf{1}} \; \sum_{i \in N} v_i  
\label{eqn:maxmin_C}
\end{align}
It is to be noted that we cannot compare the values  in (\ref{eqn:maxmin_unC}) and (\ref{eqn:maxmin_C}), in general. For instance,  if all $i$'s have equal values of $r_i w_{ib}$ and only one $i$ has good value of $r_i w_{ig}$, then for \mbox{$k_g > 1$}, the value in (\ref{eqn:maxmin_unC}) would be greater than that in (\ref{eqn:maxmin_C}).
This can be seen using Equation (\ref{eqn:steadystate}); 
the value of $\sum_{i \in N} r_i(w_{ii}^0 v_i^0 - w_{ib}y_i)$ would stay the same while the value of $\sum_{i \in N} r_i w_{ig} x_i$ would be higher in (\ref{eqn:maxmin_unC}) than in (\ref{eqn:maxmin_C}).
On the other hand, if all $i$'s have equal values of $r_i w_{ig}$ and only one $i$ has good value of $r_i w_{ib}$, then for $k_b > 1$, the value in (\ref{eqn:maxmin_C}) would be greater than that in (\ref{eqn:maxmin_unC}).

\vspace{-2mm}
\subsection{Result for a Popular Class of Weight Distributions}
\label{sec:WC}

We now present a result concerning a class of distribution of edge weights in a network, which includes the popular weighted cascade (WC) model.

\begin{proposition}
Let $N_i = \{j:w_{ij} \neq 0\}$, $d_i=|N_i|$, and $j \in N_i \Longleftrightarrow i \in N_j$.
If $\forall {i \in N}, w_{ig} = w_{ib} = w_{ii}^0 = \frac{1}{\alpha + d_i}= w_{ij} , \forall j\in N_i$, where $\alpha>0$, then $r_i w_{ig} = r_i w_{ib} = \frac{1}{\alpha}, \forall {i \in N}$.
\label{prop:WC}
\end{proposition}

%
%
%
%
%
 
 \begin{proof}
 We know that 
 
 \vspace{-3mm}
 \begin{small}
 \begin{align*}
 &\big( \mathbf{I}-\mathbf{w}^T \big) \mathbf{r} =  \mathbf{1}
 \\
 \Longleftrightarrow&\;
 \mathbf{r}  = \mathbf{1} + \mathbf{w}^T \mathbf{r}
 \\
 \Longleftrightarrow&\;
 \forall {i \in N}:\;
 r_i = 1+ \sum_{j\in N_i} w_{ji} r_j
 = 1 + \sum_{j\in N_i} \left( \frac{1}{\alpha + d_j} \right) r_j
 \end{align*}
 \end{small}
 \vspace{-2mm}
 
 Let us assume $r_i = \gamma (\alpha+d_i)$, where $\gamma$ is some constant. If this satisfies the above equation, the uniqueness of $r_i$ ensures that it is the only solution. Hence we have
 
 \vspace{-3mm}
 \begin{small}
 \begin{align*}
 &\forall {i \in N}:\;
 \gamma(\alpha+d_i)= 1 + \sum_{j\in N_i} \gamma
 = 1 + \gamma d_i
 \\[-.5em]
 \Longleftrightarrow&\;
 \gamma = \frac{1}{\alpha}
 \end{align*}
 \vspace{-3mm}
 \end{small}

 
 \vspace{-3mm}
 \begin{small}
 \begin{align*}
 \therefore
 \forall {i \in N}:\;
 r_i w_{ig} = r_i w_{ib} = \frac{\alpha+d_i}{\alpha}\cdot\frac{1}{\alpha+d_i}
 = \frac{1}{\alpha}
 \end{align*}
 \end{small}
 \vspace{-3mm}
 \end{proof}
 
 %
The above result implies that models which assign weights for all $i$ such that $w_{ig} = w_{ib} = w_{ii}^0 = \frac{1}{\alpha + d_i}= w_{ij} , \forall j\in N_i$, are suitable for the use of a random strategy,
since the decision parameter for either camp ($r_i w_{ig},r_i w_{ib}$) holds the same value for all nodes. That is, in these models, a random strategy that exhausts the entire budget is optimal. This class of models includes the popular weighted cascade  model, which would assign the weights  with $\alpha = 3$.

\vspace{-2mm}
\section{Effect of Concave Influence Function}
\label{sec:concave}

The linear influence function (\ref{eqn:update_rule})  without any  bound on investment per node, leads to an optimal strategy that concentrates the investment on a single node (Proposition~\ref{prop:fundamental_unbounded}). 
As motivated in Section \ref{sec:motiv},
several social and economic settings follow law of {diminishing marginal returns}, which says that for higher investments, the marginal returns (influence in our context) are lower for a marginal increase in investment.
%
%
%
A concave influence function would account for such diminishing marginal influence of a camp with increasing investment on a node which, as we shall see, would advise against concentrated investment on a single node. 
For the purpose of our analysis so as to arrive at precise closed-form expressions and specific insights, we consider a particular form of concave functions: $x_i^{1/t}$  when the investment is $x_i$. It is to be noted, however, that it can be extended to other concave functions since we use a common framework of convex optimization, however the analysis could turn out to be more complicated or intractable.


\vspace{-2mm}
\subsection{The Case of Unbounded Investment per Node}
\label{sec:concave_unbounded}

~
\vspace{-3mm}
\begin{small}
\begin{align*}
\max_{\substack{\sum_i x_i \leq k_g \\ x_i \geq 0}} \;
 \min_{\substack{\sum_i y_i \leq k_b \\ y_i \geq 0}} \; \sum_{{i \in N}} v_i
 \\ \text{s.t. }
 \forall {i \in N}: \;
 v_i 
 = 
 w_{ii}^0 v_i^0
 + \sum_{j \in N} w_{ij}v_j 
 + w_{ig} x_i^{1/t}
 - w_{ib} y_i^{1/t}
 & 
\end{align*}
\end{small}
\vspace{-3mm}

\begin{proposition}
In Setting \ref{sec:concave_unbounded}, for $t>1$, it is optimal for the good and bad camps to invest in node $i$ proportional to $(r_i w_{ig})^\frac{t}{t-1}$ and $(r_i w_{ib})^\frac{t}{t-1}$, subject to  positivity of $r_i w_{ig}$ and $r_i w_{ib}$ respectively.
\label{prop:concave_unbounded}
\end{proposition}

%
%
\noindent
A proof of Proposition \ref{prop:concave_unbounded} is provided in 
Appendix \ref{app:concave_unbounded}.
%

%
%


 

\begin{remark}
[Skewness of investment]


When we compare the results for lower and higher values of $t$,
the investment made by the good camp has an exaggerated correlation with the value of $r_i w_{ig}$ for lower values of $t$. In particular, the investment made is very skewed towards nodes with high values of $r_i w_{ig}$ when $t$ is very low, while it is proportional to $r_i w_{ig}$ when $t$ is very high.
Note that $t=1$ corresponds to the linear case in Setting \ref{sec:fundamental} where the investment is extremely skewed with each camp investing its entire budget on only one node.
%



\label{rem:skewness}
\end{remark}

\begin{remark}[User perception of fairness]
The skewness can be linked to user perception of fairness
\cite{jain1984quantitative}.
Suppose a node $p$ is such that $r_p w_{pg} = \max_i r_i w_{ig}$, and it is the unique node with this maximum value. Suppose a node $q$ is such that $r_q w_{qg} = r_p w_{pg} - \epsilon$, where $\epsilon$ is positive and infinitesimal. From the perspective of node $q$, the  strategy would be fair if the investment in $q$ is not much less than that in $p$, since they are almost equally valuable. However, $t=1$ leads to a highly skewed investment where $p$ receives $k_g$ and $q$ receives 0, which can be perceived as unfair by  $q$. As $t$ increases, the investment becomes less skewed; in particular, $t \rightarrow \infty$ leads to investment on a node $i$ to be proportional to $r_i w_{ig}$, which could be perceived as fair by the nodes.
\end{remark}

\vspace{-2mm}
\subsection{The Case of Bounded Investment Per Node}
\label{sec:concave_bounded}

With the additional constraints $x_i \leq 1$ and $y_i \leq 1, \forall {i \in N}$, the optimal investment strategies are given by Proposition~\ref{prop:concave_bounded}. We provide its proof in 
Appendix \ref{app:concave_bounded}.


\begin{proposition}
\label{prop:concave_bounded}

\noindent
In Setting \ref{sec:concave_bounded},
if the number of nodes with $r_i w_{ig}>0$ is less than $k_g$,
it is optimal for  good camp to invest 1 unit on each node $i$ with $r_i w_{ig}>0$ and 0 on all other nodes.

\noindent
If the number of nodes with $r_i w_{ig}\!>\!0$ is at least $k_g$,
let $\hat{\gamma}>0$ be the solution of 

\vspace{-3mm}
\begin{small}
\begin{align*}
\sum_{{i: {r_i w_{ig}} \in (0, t \gamma]}} \left( \frac{r_i w_{ig}}{t \gamma} \right) ^\frac{t}{t-1} + \sum_{i:r_i w_{ig} > t {\gamma}} 1 = k_g
\end{align*}
\end{small}
%
\noindent
It can be shown that $\hat{\gamma}$ exists and is unique;
 it is then optimal for the good camp to follow the investment strategy:

\vspace{-3mm}
\begin{small}
\begin{align*}
x_i^* &= 0, \text{ if } r_i w_{ig} \leq 0
\\
x_i^* &= 1, \text{ if } r_i w_{ig} > t \hat{\gamma}
\\
x_i^* &= \left( k_g - \sum_{i:r_i w_{ig} > t \hat{\gamma}} 1 \right) \left( \frac{(r_i w_{ig})^\frac{t}{t-1}}{\sum_{{i: {r_i w_{ig}} \in (0, t \hat{\gamma}]}} (r_i w_{ig})^\frac{t}{t-1}} \right)  
, 
\\ &\;\;\;\;\;\;\;\;\;\;\;\;\;\;\;\;\;\;\;\;\;\;\;\;\;\;\;\;\;\;\;\;\;\;\;\;\;\;\;\;\;\;\;\;\;\;\;\;\;\;\;\;\;\;\;\;\text{ if } r_i w_{ig} \in (0, t \hat{\gamma}]
\end{align*}
\end{small}
\vspace{-4mm}

\noindent
The optimal strategy of the bad camp is analogous.
\end{proposition}

Note that for $r_i w_{ig} \in (0,t \hat{\gamma}]$, we can alternatively write 
$x_i^* = \left( \frac{r_i w_{ig}}{t \hat{\gamma}} \right) ^\frac{t}{t-1}$, which would be between 0 and 1. So the nodes with positive values of $r_i w_{ig}$ should be classified into two sets, one containing nodes with $r_i w_{ig} \in (0,t \hat{\gamma}]$ (for which $x_i^* \in (0,1]$) and the other containing nodes with $r_i w_{ig} > t \hat{\gamma}$ ($x_i^*$ forcefully limited to $1$). So we can effectively start with all nodes in the former set (meaning $t\gamma \geq \max_{i \in N} r_i w_{ig}$) and then transfer nodes to the latter set as per descending values of $r_i w_{ig}$ (as we reduce $t\gamma$), until we have two sets, one with \mbox{$x_i^* = \left( \frac{r_i w_{ig}}{t {\gamma}} \right) ^\frac{t}{t-1} 
= \left( \frac{(r_i w_{ig})^\frac{t}{t-1}}{\sum_{{i: {r_i w_{ig}} \in (0, t {\gamma}]}} (r_i w_{ig})^\frac{t}{t-1}} \right)
\leq 1$} and the other with $x_i^*$ forcefully limited to $1$.

%
%

\begin{intuition}
The solution suggests that the optimal strategy can be obtained using a trial-and-error iterative process. A camp could use the optimal strategy for the unbounded case suggested in Proposition \ref{prop:concave_unbounded}. If we get $x_i^*>1$ for any node, we assign $x_i^*=1$ to node $i$ with the highest value of $r_i w_{ig}$, and use Proposition \ref{prop:concave_unbounded} again by excluding node $i$ and decrementing the available budget by 1. This process would be repeated until $x_i^* \leq 1, \forall {i \in N}$.
\label{int:trialnerror}
\end{intuition}

\vspace{-5mm}
\section{Acting as Competitor's Adversary}
\label{sec:minmax_investment}


In this setting, a camp explicitly acts to maximize the competitor's investment  or deviation from its desired investment, that is required to drive the sign of the average opinion value of the population in the latter's favor. Without loss of analytical generality, we consider that the good camp acts as the adversary.



\vspace{-3mm}
\subsection{The Case of Unbounded Investment per Node}
\label{sec:minmax_investment_unbounded}

~
\vspace{-4mm}
\begin{small}
\begin{gather*}
\max_{\substack{\sum_i x_i \leq k_g \\ x_i \geq 0}} \;
 \min_{y_i \geq 0} \; \sum_{{i \in N}} y_i
\\
\text{s.t. }
\sum_{i \in N} v_i \leq 0
\\
\hspace{-3mm}
\sum_{i \in N}
v_i 
=
 \sum_{i \in N} r_i ( w_{ig}x_i + w_{ii}^0 v_i^0) - \sum_{i \in N} r_i w_{ib}y_i
\end{gather*}
\end{small}
\vspace{-3mm}

\begin{proposition}
In Setting \ref{sec:minmax_investment_unbounded}, it is optimal for the good camp to invest its budget in node $i$ with the maximum value of $r_i w_{ig}$, subject to it being positive. For the bad camp, it is optimal to invest in node $i$ with maximum value of $r_i w_{ib}$, subject to its positivity.
(If there does not exist any node $i$ with positive value of $r_i w_{ib}$, it is optimal for the bad camp to not invest at all).
%
The optimal amount of investment made by the bad camp is 

\vspace{-3mm}
\begin{small}
\begin{align*}
\max \left\{ \frac{1}{\max_{j\in N} w_{jb} r_j}  \left( k_g \max \left\{ \max_{i \in N} r_i w_{ig} , 0 \right\} +  \sum_{i \in N} r_i w_{ii}^0 v_i^0\right) , 0 \right\}
\end{align*}
\end{small}
\vspace{-4mm}
\label{prop:investment_unbounded}
\end{proposition}

\noindent
A proof of Proposition \ref{prop:investment_unbounded} is provided in 
Appendix \ref{app:investment_unbounded}.


It is to be noted that, contrary to the previous settings, the amount of investment made by the bad camp in this setting is dependent on the good camp's parameters ($r_i w_{ig}$ and $k_g$) as well as the opinion bias parameters ($w_{ii}^0 v_i^0$). This is because in the previous settings, the bad camp's objective was to minimize the sum of opinion values without considering the actual value of this sum, while the current setting necessitates the bad camp to ensure that this sum is non-positive; this requires taking into account the effects of good camp's influence and the initial biases on this sum.

\begin{remark}
[Maximizing competitor's deviation]

Let the desired investments for the good and bad camps be $\bar{x}_i$ and $\bar{y}_i$, respectively.
Thus the optimization problem is

\vspace{-3mm}
\begin{small}
\begin{gather*}
\max_{\substack{\sum_i (x_i - \bar{x}_i)^2 \leq k_g \\ x_i \geq 0}} \;
 \min_{y_i \geq 0} \; \sum_{{i \in N}} (y_i - \bar{y}_i)^2
\\
\text{s.t. }
\sum_{i \in N} v_i \leq 0
\\
\hspace{-3mm}
\sum_{i \in N}
v_i 
=
 \sum_{i \in N} r_i ( w_{ig}x_i + w_{ii}^0 v_i^0) - \sum_{i \in N} r_i w_{ib}y_i
\end{gather*}
\end{small}
\vspace{-2mm}

\noindent
Let $\hat{\gamma}>0$ be the solution of 

\vspace{-3mm}
\begin{small}
\begin{align*}
\sum_{\substack{i: {r_i w_{ig}} \\ \geq -2 \gamma \bar{x}_i}} \left( \frac{r_i w_{ig}}{2 \gamma} \right) ^2 + \sum_{\substack{i:r_i w_{ig} \\ < -2 {\gamma} \bar{x}_i}} (\bar{x}_i)^2 = k_g
\end{align*}
\end{small}
\vspace{-1mm}

\noindent
Then the good camp's optimal strategy is the following:
\begin{small}
\begin{align*}
& x_i^* = 0, \text{ if } r_i w_{ig} < -2 {\gamma} \bar{x}_i
\\
 &= \bar{x}_i + \sgn(r_i w_{ig}) \Bigg( k_g - \!\! \sum_{\substack{i:r_i w_{ig} \\ < -2 {\gamma} \bar{x}_i}} \!\! (\bar{x}_i)^2 \Bigg)^{\frac{1}{2}} \left(  \frac{  (r_i w_{ig})^2}{\sum_{\substack{i: {r_i w_{ig}} \\ \geq -2 \gamma \bar{x}_i}} (r_i w_{ig})^2} \right)^{\frac{1}{2}}
\\ &\;\;\;\;\;\;\;\;\;\;\;\;\;\;\;\;\;\;\;\;\;\;\;\;\;\;\;\;\;\;\;\;\;\;\;\;\;\;\;\;\;\;\;\;\;\;\;\;\;\;\;\;\;\;\;\;\;\;\;\;\;\;\;\;\;\;\;\text{ if } r_i w_{ig} \geq -2 {\gamma} \bar{x}_i
\end{align*}
\end{small}
%
%
If there does not exist a $\hat{\gamma}>0$ (because $\sum_{i:r_i w_{ig}<0} (\bar{x}_i)^2 < k_g$ and no node with $r_i w_{ig}>0$), we invest 0 on any node with $r_i w_{ig}<0$ and $\bar{x}_i$ on any node with $r_i w_{ig}=0$.
\end{remark}

\noindent
This can be proved on similar lines as Proposition \ref{prop:concave_bounded}. Here,
\begin{small}
\begin{align*}
(x_i^* - \bar{x}_i)^2 =  \Bigg( k_g - \sum_{\substack{i:r_i w_{ig} \\ < -2 {\gamma} \bar{x}_i}} (\bar{x}_i)^2 \Bigg) \left(  \frac{  (r_i w_{ig})^2}{\sum_{\substack{i: {r_i w_{ig}} \\ \geq -2 \gamma \bar{x}_i}} (r_i w_{ig})^2} \right)
\end{align*}
\end{small}
and the optimal square root is determined by $\sgn(r_i w_{ig})$ (since a positive $r_i w_{ig}$ would mean a higher optimal investment as opposed to a negative $r_i w_{ig}$).
Here, it is possible that a node $i$ is invested on by the good camp even if it has negative $r_i w_{ig}$, so as to have the investment close to $\bar{x}_i$.

\vspace{-2mm}
\subsection{The Case of Bounded Investment per Node}
\label{sec:minmax_investment_bounded}

The optimal strategies of the camps can be easily obtained for this setting on similar lines as  Proposition \ref{prop:fundamental_bounded}.

\begin{proposition}
In Setting \ref{sec:minmax_investment_bounded}, it is optimal for the good camp to invest in nodes one at a time, subject to a maximum investment of 1 unit per node,
in decreasing order of values of $r_i w_{ig}$
until either the budget $k_g$ is exhausted or we reach a node with a non-positive value of $r_i w_{ig}$.
Say the so derived optimal investment on node $i$ is $x_i^*$.
The optimal strategy of the bad camp is to invest in nodes one at a time, subject to a maximum investment of 1 unit per node,
in decreasing order of values of $r_i w_{ib}$
until  
$\sum_{i \in N} r_i w_{ib} y_i \geq \sum_{j \in N} r_j(w_{jg}x_j^* + w_{jj}^0 v_j^0) $.
%
\label{prop:minmax_investment_bounded}
\end{proposition}

Note that the terminating condition $\sum_{i \in N} r_i w_{ib} y_i \geq \sum_{j \in N} r_j(w_{jg}x_j^* + w_{jj}^0 v_j^0) $ is same as the required condition $\sum_{i \in N} v_i \leq 0$, when $x_j^*$ is the optimal investment by the good camp on node $j$.

\vspace{-2mm}
\section{Common Coupled Constraints Relating Bounds on Combined Investment per Node}
\label{sec:CCC}

As motivated in Section \ref{sec:motiv},
 a sequential play would be more natural than a simultaneous one in certain scenarios, for instance, in presence of a  mediator or central authority which may be responsible for giving permissions for campaigning or scheduling product advertisements to be presented to an individual.
%
%
We hence consider two sequential play settings, which result in Stackelberg variants of the considered game.
We use {\em subgame perfect Nash equilibrium} as the equilibrium notion; also since it is a  zero-sum game, we refer to the equilibrium as either maxmin or minmax, based on which camp plays first. 
Without loss of analytical generality, we conduct our analysis while assuming the good camp plays first. The bad camp would hence choose a strategy that minimizes the sum of opinion values as a best response to the good camp's strategy. Foreseeing this, the good camp would want to maximize this minimum value; this is popularly known as the {\em backward induction approach}. We hence derive the subgame perfect Nash equilibrium strategy profile and the corresponding maxmin value. The minmax profile and value can be obtained symmetrically. 





In this section, we consider a setting
in which the combined  investment on a node by both  camps  is bounded by a certain limit. Without loss of generality, we assume  this limit to be 1 unit. This leads to the introduction of common coupled constraints (CCC): $\,x_i+y_i \leq 1, \forall i \in N$. 
%
%

\vspace{-3mm}
\begin{small}
\begin{align*}
\max_{\substack{\sum_i x_i \leq k_g \\ x_i \geq 0 }} 
 \min_{\substack{\sum_i y_i \leq k_b \\ 0 \leq y_i \leq (1-x_i) }}  \sum_{{i \in N}} v_i
\end{align*}
\end{small}
\vspace{-3mm}

\vspace{-1mm}
\begin{small}
\begin{align*}
\text{s.t. }
\forall {i \in N}: 
v_i 
= 
 w_{ii}^0 v_i^0
+ \sum_{j \in N} w_{ij}v_j 
+ w_{ig} x_i
 -  w_{ib} y_i
\end{align*}
\end{small}
\vspace{-3mm}

The inner term is

\vspace{-5mm}
\begin{small}
\begin{align*}
 \min \sum_{{i \in N}} v_i
\\ \text{s.t. }
\mathbf{y} \geq \mathbf{0}
\\
\sum_{i \in N} y_i \leq k_b
\text{ or }
 -\sum_{i \in N} y_i \geq -k_b
& \;\; \leftarrow \alpha
\\
\forall {i \in N}: 
v_i 
- \sum_{j \in N} w_{ij}v_j 
+ w_{ib} y_i
= w_{ig} x_i
+ w_{ii}^0 v_i^0
& \;\; \leftarrow z_i
\\
\forall {i \in N}: \;
x_i + y_i \leq 1
\text{ or }
 -y_i \geq -(1-x_i)
& \;\; \leftarrow \gamma_i
\end{align*}
\end{small}
\vspace{-3mm}

Its dual problem can be written as

\vspace{-3mm}
\begin{small}
\begin{align}
\max \;\; -\alpha k_b + \sum_{i\in N} \left( z_i ( w_{ig} x_i + w_{ii}^0 v_i^0 ) - \gamma_i(1-x_i) \right)
\label{eqn:dualthisnew}
&
\\ \text{s.t. }
\forall {i \in N}: \;
 z_i 
- \sum_{j \in N} w_{ji} z_j 
= 1
 \;\; \leftarrow v_i &
\label{eqn:thisnew}
\\
\forall {i \in N}: \;
w_{ib} z_i 
-  \gamma_i
- \alpha
\leq 0
 \;\; \leftarrow y_i &
\label{eqn:ynz1new}
\\
\nonumber
\alpha \geq 0
\\
\nonumber
\forall {i \in N}: \;
z_i \in \mathbb{R}, 
\gamma_i \geq 0
\end{align}
\end{small}
\vspace{-3mm}

\noindent
As earlier, from (\ref{eqn:thisnew}), we have $z_i = \left( (\mathbf{I}-\mathbf{w}^T)^{-1} \mathbf{1} \right)_i = r_i$.
%
%
For satisfying Constraint (\ref{eqn:ynz1new}), it is required that
\begin{align}
\forall {i \in N}: \;
r_i w_{ib}  - \gamma_i -\alpha\leq 0
\;\text{ or }\;
\gamma_i \geq {r_i w_{ib}}-\alpha
\label{eqn:gammanew}
\end{align}

To maximize objective function (\ref{eqn:dualthisnew}), it is required that $\gamma_i$ should be as low as possible (knowing that $1-x_i \geq 0$).
So the above condition $\gamma_i \geq {r_i w_{ib}}-\alpha$ along with $\gamma_i \geq 0$ gives
\begin{align*}
\forall {i \in N}: \;
\gamma_i = \max \{r_i w_{ib} -\alpha, 0\}
\end{align*}

So we need to maximize the objective function with respect to $\gamma_i,x_i, \forall {i \in N}$ and $\alpha$. For this purpose, let us define a set with respect to $\alpha$, namely, 
\begin{align*}
J_{\alpha} = \{j:r_j w_{jb}-\alpha \geq 0\}
\end{align*}
So the objective function to be maximized is
\begin{small}
\begin{align}
 -\alpha k_b -\sum_{j \in J_{\alpha}} (r_j w_{jb} -\alpha)(1-x_j) + \sum_{i\in N} r_i ( w_{ig} x_i + w_{ii}^0 v_i^0 ) 
 \label{eqn:objfunbudget}
\end{align}
\end{small}
which is equal to
\begin{small}
\begin{align}
 \alpha  \Big( \! \sum_{j\in J_{\alpha}} (1 \! - \! x_j)-k_b \Big) - \! \sum_{j\in J_{\alpha}} (1 \! - \! x_j) r_j w_{jb} + \sum_{i\in N} r_i ( w_{ig} x_i \! + \! w_{ii}^0 v_i^0 ) 
  \label{eqn:objfunrearranged}
\end{align}
\end{small}
\vspace{-2mm}

\begin{claim}
It is sufficient to search the values of
$\alpha \in \{r_j w_{jb}\}_{j:r_j w_{jb}>0} \cup \{0\}$ to find an optimal solution.
%
\label{clm:alpharange}
\end{claim}

\begin{proof}


%

Since $\alpha\geq 0$, we have $\alpha \neq r_j w_{jb}$ for any $r_j w_{jb}<0$.
Consider a range $[r_l w_{lb},r_u w_{ub}]$ for a consecutive pair of distinct values of $r_j w_{jb}$.
If a range has both these values negative, we do not search for $\alpha$ in that range, since $\alpha \geq 0$.
If a range has $r_l w_{lb} \leq 0$ and $r_u w_{ub} > 0$, we search for $\alpha$ in $[0,r_u w_{ub}]$.
%
We will now determine an optimal value of $\alpha$ in the valid searchable subset of $[r_l w_{lb},r_u w_{ub}]$,
for a given $\mathbf{x}$.


\vspace{1mm}
Case 1: If $\alpha =r_l w_{lb}$ (where $r_l w_{lb} \geq 0$):
\\
We have $J_{\alpha}=\{j:r_j w_{jb} \geq r_l w_{lb}\}$.
The value of the objective function (\ref{eqn:objfunbudget}) becomes

\vspace{-3mm}
\begin{small}
\begin{align*}
 - r_l w_{lb} k_b -\sum_{j \in J_{\alpha}} (r_j w_{jb} - r_l w_{lb})(1-x_j) + \sum_{i \in N} r_i ( w_{ig} x_i + w_{ii}^0 v_i^0 ) 
\end{align*}
\end{small}
\vspace{-1mm}

\vspace{1mm}
Case 2: If $\alpha \in (r_l w_{lb},r_u w_{ub}]$:
\\
We have $J_{\alpha}=\{j:r_j w_{jb} \geq r_u w_{ub}\}$.
\\
Case 2a: If $\sum_{j \in J_{\alpha}} (1-x_j)-k_b \geq 0$, we have an optimal $\alpha = r_u w_{ub}$ (from (\ref{eqn:objfunrearranged})).
\\
Case 2b: Instead, if $\sum_{j \in J_{\alpha}} (1-x_j)-k_b < 0$, we have optimal $\alpha \rightarrow r_l w_{lb}$ if $r_l w_{lb} \geq 0$, and the optimal value is the same as that for $\alpha = r_l w_{lb}$ (Case 1).
Note here that if $r_l w_{lb}<0$ and $r_u w_{ub} \geq 0$, we would have optimal $\alpha = 0$.

\vspace{1mm}
Case 3:
If $r_l w_{lb} = \max_{i \in N} r_i w_{ib}$, that is, when we are looking for $\alpha \geq \max_{i \in N} r_i w_{ib}$. For $\alpha = \max_{i \in N} r_i w_{ib}$, we have $J_{\alpha} = \{ \argmax_{i \in N} r_i w_{ib} \}$ and so the term $\sum_{j \in J_{\alpha}} (r_j w_{jb} -\alpha)(1-x_j)$ in (\ref{eqn:objfunbudget}) vanishes. 
For $\alpha > \max_{i \in N} r_i w_{ib}$, we have $J_{\alpha} = \{ \}$ and so the term $\sum_{j \in J_{\alpha}} (r_j w_{jb} -\alpha)(1-x_j)$ in (\ref{eqn:objfunbudget}) vanishes in this case too. 
So the objective function to be maximized becomes

\vspace{-3mm}
\begin{small}
\begin{align*}
 -\alpha k_b + \sum_{i \in N} r_i ( w_{ig} x_i + w_{ii}^0 v_i^0 ) 
\end{align*}
\end{small}
\vspace{-3mm}

\noindent
for which the optimal $\alpha = \max_{i \in N} r_i w_{ib}$ (the lowest value of $\alpha$ such that $\alpha \geq \max_{i \in N} r_i w_{ib}$).

\vspace{1mm}
The above cases show that it is sufficient to search the values of
$\alpha \in \{r_j w_{jb}\}_{j:r_j w_{jb}>0} \cup \{0\}$ to determine an optimal value of the objective function.
\end{proof}

Now that we have established that the only possible values of optimal $\alpha$ are $\{r_j w_{jb}\}_{j:r_j w_{jb}>0} \cup \{0\}$, we can assume optimal $\alpha = r_{\hat{j}} w_{\hat{j}b}$ for $\hat{j} \in \{j:r_j w_{jb}>0\} \cup \{d\}$, where the dummy node $d$ is such that $r_d w_{db} = 0$. 



Recalling the objective function in (\ref{eqn:objfunrearranged}),

\vspace{-3mm}
\begin{small}
\begin{align*}
\hspace{-3mm}
&\sum_{i \in N} r_i ( w_{ig} x_i \! + \! w_{ii}^0 v_i^0 )  +  \alpha  \Big( \! \sum_{j\in J_{\alpha}} (1 \! - \! x_j)-k_b \Big) - \! \sum_{j\in J_{\alpha}} (1 \! - \! x_j) r_j w_{jb} 
\\
\hspace{-5mm}
&=
\sum_{i \in N} r_i ( w_{ig} x_i \! + \! w_{ii}^0 v_i^0 )   - \! \left[\sum_{j\in J_{\alpha}} (1 \! - \! x_j) r_j w_{jb} +  \alpha  \Big( \! k_b - \! \sum_{j\in J_{\alpha}} (1 \! - \! x_j) \Big) \! \right]
\\
\hspace{-5mm}
&=
\sum_{i \in N} r_i ( w_{ig} x_i \! + \! w_{ii}^0 v_i^0 ) 
\\ 
\hspace{-5mm}
&\;\;\;\;\;
- \! \left[\sum_{j\in J_{\alpha}} (1 \! - \! x_j) r_j w_{jb} +  \Big( \! k_b - \sum_{j\in J_{\alpha}} (1 \! - \! x_j) \Big) r_{\hat{j}} w_{\hat{j}b} \right]
\end{align*}
\end{small}
\vspace{-2mm}

\noindent
Let $I_{\alpha}=\{j:r_j w_{jb} > \alpha\}$, $P_{\alpha}=\{j:r_j w_{jb} = \alpha = r_{\hat{j}} w_{\hat{j}b}\}$.
So the objective function is

\vspace{-3mm}
\begin{small}
\begin{align*}
&
\sum_{i \in N} r_i ( w_{ig} x_i \! + \! w_{ii}^0 v_i^0 ) 
- \! \Bigg[\sum_{j\in P_{\alpha}} (1 \! - \! x_j) r_j w_{jb} - \sum_{j\in P_{\alpha}} (1 \! - \! x_j) r_{\hat{j}} w_{\hat{j}b} 
\\
&\;\;\;\;\;
+ \sum_{j\in I_{\alpha}} (1 \! - \! x_j) r_j w_{jb} +  \Big( \! k_b - \sum_{j\in I_{\alpha}} (1 \! - \! x_j) \Big) r_{\hat{j}} w_{\hat{j}b}  \Bigg]
\\
&=
\sum_{i \in N} r_i ( w_{ig} x_i \! + \! w_{ii}^0 v_i^0 ) 
\\ 
&\;\;\;\;\;
- \! \left[\sum_{j\in I_{\alpha}} (1 \! - \! x_j) r_j w_{jb} +  \Big( \! k_b - \sum_{j\in I_{\alpha}} (1 \! - \! x_j) \Big) r_{\hat{j}} w_{\hat{j}b} \right]
\end{align*}
\end{small}
\vspace{-2mm}

\noindent
Comparing this with generic objective function (\ref{eqn:steadystate}) and since it should hold for any $r_i,w_{ig},w_{ib},w_{ii}^0,v_i^0$, it is necessary that the coefficients of {non-zero} values of $r_i w_{ib}$ are the same in both forms of the objective function. This along with the fact that $\forall j \in I_{\alpha}: r_j w_{jb} > 0$ (since $\alpha \geq 0$), gives
$\sum_{j \in I_{\alpha}} y_j = \sum_{j \in I_{\alpha}} (1-x_j)$.
Also if $r_{\hat{j}} w_{\hat{j}b} > 0$, then
$\sum_{j \in P_{\alpha}} y_j = k_b - \sum_{j\in I_{\alpha}} (1  -  x_j)$.
And for all other terms, we have
$\sum_{j \notin J_{\alpha}} y_j = 0$.
Since $\forall j\in N : 0 \leq y_j \leq 1-x_j$, these are equivalent to

\begin{small}
\vspace{-3mm}
\begin{align}
\forall j \in I_{\alpha} \! : y_j = 1 \! - \! x_j \; ,
\forall j \notin J_{\alpha} \! : y_j = 0 \; ,
\sum_{j \in P_{\alpha}} \! y_j = k_b - \! \sum_{j\in I_{\alpha}} (1 \! - \! x_j)
\label{eqn:alloc}
\end{align}
\vspace{-3mm}
\end{small}

To check for the consistency of budget of the bad camp, it is necessary that
$\sum_{j \in I_{\alpha}} y_j \leq k_b$. This gives the constraint
$ \sum_{j \in I_{\alpha}} (1-x_j) \leq k_b$ or equivalently,

\begin{small}
\vspace{-3mm}
\begin{align}
\sum_{j \in I_{\alpha}} x_j \geq |I_{\alpha}| - k_b
\label{eqn:constraint1}
 \end{align}
 \vspace{-2mm}
 \end{small}
 
Also if $r_{\hat{j}} w_{\hat{j}b} > 0$, for the consistency of investment on the nodes in $P_{\alpha}$ (that is, $\forall j \in P_{\alpha} : x_j+y_j \leq 1$), it is necessary that 
$\sum_{j \in P_{\alpha}}  y_j \leq \sum_{j \in P_{\alpha}} (1-x_j)$ or equivalently, $ k_b -  \sum_{j\in I_{\alpha}} (1 -  x_j) \leq \sum_{j \in P_{\alpha}} (1-x_j)$ or equivalently,

\begin{small}
\vspace{-3mm}
\begin{align}
\sum_{j \in I_{\alpha}} x_j + \sum_{j \in P_{\alpha}} x_j \leq |I_{\alpha}| + |P_{\alpha}| - k_b
\label{eqn:constraint2}
 \end{align}
  \vspace{-2mm}
 \end{small}
 
 To check for the consistency of budget of the good camp, it is necessary that $\sum_{j \in I_{\alpha}} x_j \leq k_g$ and $\sum_{j \in I_{\alpha}} x_j + \sum_{j \in P_{\alpha}} x_j \geq 0$.
 These along with Inequalities (\ref{eqn:constraint1}) and (\ref{eqn:constraint2}) give
 $|I_{\alpha}| - k_b \leq k_g$ and $|I_{\alpha}| + |P_{\alpha}| - k_b \geq 0$, or equivalently,
 
 \begin{small}
 \vspace{-3mm}
 \begin{align}
|I_{\alpha}|  \leq k_g + k_b \;\;\;\text{ and }\;\;\; |I_{\alpha}| + |P_{\alpha}| \geq k_b 
 \label{eqn:constraint3}
  \end{align}
   \vspace{-3mm}
  \end{small}
  
  The sets $I_{\alpha}$ and $P_{\alpha}$  depend only on $\hat{j}$.
  So let the set of  $\hat{j}$'s that satisfy the constraints in (\ref{eqn:constraint3}) be denoted by $\tilde{J}$, that is,
  
   \begin{small}
   \vspace{-3mm}
   \begin{align*}
  \tilde{J} = \{  \hat{j} : |I_{\alpha}|  \leq k_g + k_b \text{ and } |I_{\alpha}| + |P_{\alpha}| \geq k_b  \}
    \end{align*}
     \vspace{-2mm}
    \end{small}

The term $\sum_{i \in N} r_i w_{ii}^0 v_i^0 $ being a constant, and substituting $\alpha = r_{\hat{j}} w_{\hat{j}b}$, objective function (\ref{eqn:objfunbudget}) becomes

\vspace{-2mm}
\begin{small}
\begin{align}
\nonumber
\max_{\mathbf{x},\hat{j}}
\sum_{i \in N} r_i w_{ig} x_i  + \sum_{j : r_j w_{jb} \geq r_{\hat{j}} w_{\hat{j}b} } x_j (r_j w_{jb} -r_{\hat{j}} w_{\hat{j}b}) 
\\- \sum_{j : r_j w_{jb} \geq r_{\hat{j}} w_{\hat{j}b} } (r_j w_{jb} -r_{\hat{j}} w_{\hat{j}b})
-r_{\hat{j}} w_{\hat{j}b} k_b 
\nonumber
\\
\nonumber
\Longleftrightarrow
\max_{\hat{j}} \Big[  \max_{\mathbf{x}} 
\sum_{i \in N} x_i (r_i w_{ig} + \max\{r_i w_{ib}-  r_{\hat{j}} w_{\hat{j}b} , 0\})
\\- \sum_{i \in N} \max\{ r_i w_{ib} -r_{\hat{j}} w_{\hat{j}b} , 0\}
-r_{\hat{j}} w_{\hat{j}b} k_b \Big]
\label{eqn:maxj_maxx}
\end{align}
\end{small}
\vspace{-2mm}


Hence the good camp's optimal strategy
 can be obtained by maximizing  (\ref{eqn:maxj_maxx}) with respect to $\mathbf{x}$ and $\hat{j}\in \tilde{J}$, subject to Constraints  
(\ref{eqn:constraint1}) and (\ref{eqn:constraint2}), and $x_i \! \in \! [0,1], \forall {i \in N}$.



\vspace{-2mm}
\subsection*{A Greedy Approach for Determining  Optimal Strategy}

For a given $\hat{j}$, it can be seen from (\ref{eqn:maxj_maxx}) that the optimal strategy of the good camp is to determine $\mathbf{x}$ which maximizes $\sum_{i \in N} x_i (r_i w_{ig} + \max\{r_i w_{ib}-  r_{\hat{j}} w_{\hat{j}b} , 0\})$.
Since Constraint (\ref{eqn:constraint1}) should be satisfied, the minimum total investment by the good camp on nodes belonging to set $I_{\alpha}$ should be $|I_{\alpha}|-k_b$.
Hence it should invest in nodes belonging to $I_{\alpha}$ one at a time (subject to a maximum investment of 1 unit per node)
in decreasing order of values of $(r_i w_{ig} + \max \{r_i w_{ib}-  r_{\hat{j}} w_{\hat{j}b} , 0\})$, until a total investment of $|I_{\alpha}|-k_b$ is made.
Let $\tilde{x}_i$ be the good camp's investment on node $i$ after this step; its remaining budget is $k_g-(|I_{\alpha}|-k_b)$ and the maximum amount that it could henceforth invest on a node $i$ is $1-\tilde{x}_i$ (since each node has an investment capacity of 1 unit).

Now since Constraint (\ref{eqn:constraint2}) should also be satisfied, the maximum total investment by the good camp on nodes belonging to set $I_{\alpha} \cup P_{\alpha}$ should be $|I_{\alpha}|+|P_{\alpha}|-k_b$.
Hence it should now invest in nodes one at a time (maximum investment of $1-\tilde{x}_i$ in node $i$)
in decreasing order of values of $(r_i w_{ig} + \max \{r_i w_{ib}-  r_{\hat{j}} w_{\hat{j}b} , 0\})$ until one of the following occurs: (a) the remaining budget ($k_g-|I_{\alpha}|+k_b$) is exhausted or (b) a node with a negative value of $(r_i w_{ig} + \max \{r_i w_{ib}-  r_{\hat{j}} w_{\hat{j}b} , 0\})$ is reached or (c) the investment made on nodes belonging to $I_{\alpha} \cup P_{\alpha}$ reaches $|I_{\alpha}|+|P_{\alpha}|-k_b$. 
If condition (a) or (b) is met, the so obtained strategy $\mathbf{x}_{\hat{j}}^* = ( x_{\hat{j} i}^* )$ is the optimal $\mathbf{x}$ for the given $\hat{j}$.
However, if condition (c) is met,  good camp should continue investing the remaining available amount on nodes belonging to $N \setminus (I_{\alpha} \cup P_{\alpha})$ one at a time (subject to a maximum investment of 1 unit per node)
in decreasing order of values of $(r_i w_{ig} + \max \{r_i w_{ib}-  r_{\hat{j}} w_{\hat{j}b} , 0\})$.
The so obtained strategy $\mathbf{x}_{\hat{j}}^* = ( x_{\hat{j} i}^* )$ would hence be the optimal $\mathbf{x}$ for the given $\hat{j}$.

The absolute optimal strategy of the good camp can now be computed by iterating over all $\hat{j} \in \tilde{J}$ and taking the one that maximizes (from Expression (\ref{eqn:maxj_maxx}))

\vspace{-3mm}
\begin{small}
\begin{align}
\nonumber
\max_{\hat{j} \in \tilde{J}} 
\sum_{i \in N} x_{\hat{j} i}^* (r_i w_{ig} + \max\{r_i w_{ib}-  r_{\hat{j}} w_{\hat{j}b} , 0\})
\\- \sum_{i \in N} \max\{ r_i w_{ib} -r_{\hat{j}} w_{\hat{j}b} , 0\}
-r_{\hat{j}} w_{\hat{j}b} k_b 
\label{eqn:CCC_final}
\end{align}
\end{small}
\vspace{-2mm}

For the bad camp's optimal strategy, recall that

\vspace{-3mm}
\begin{small}
\begin{align*}
\sum_{i \in N} v_i = \sum_{i \in N} r_i w_{ii}^0 v_i^0 + \sum_{i \in N} r_i w_{ig}x_i  - \sum_{i \in N} r_i w_{ib}y_i
\end{align*}
\end{small}
\vspace{-3mm}

%

Since  $y_i \in [0,1-x_i], \forall {i \in N}$,
the optimal strategy of the bad camp is to invest in nodes one at a time (subject to a maximum investment of $1-x_i$ per node)
in decreasing order of values of $r_i w_{ib}$
until either its budget $k_b$ is exhausted or we reach a node with a negative value of $r_i w_{ib}$.

%

It can also be seen that if $k_g$ and $k_b$ are integers, it is an optimal investment strategy of the good and bad camps to invest one unit or not invest at all in a node. 


\begin{intuition}
Assuming a $\hat{j}$, the strategy of the good camp is to choose nodes with good values of $(r_i w_{ig} + \max\{r_i w_{ib}-  r_{\hat{j}} w_{\hat{j}b} , 0\})$. That is, it chooses nodes with not only good values of $r_i w_{ig}$, but also good values of $r_i w_{ib}$.
This is expected since 
the budget constraint per node allows the good camp (which plays first) to block those nodes on which the bad camp would have preferred to invest. 
Also, 
based on (\ref{eqn:alloc}) and the definitions of $J_{\alpha},I_{\alpha},P_{\alpha}$, 
node $\hat{j}$ can be viewed as a boundary for the bad camp's investment, that is, the bad camp would not invest in any node $i$ such that $r_i w_{ib} < r_{\hat{j}} w_{\hat{j}b}$.
%
\label{int:CCC}
\end{intuition}

\subsection*{Time Complexity of the Greedy Approach}
\label{app:complexity}

For a given $\hat{j}$,
the above greedy approach would require the good camp to select a total of $O(k_g)$ nodes to invest on. This could be done by either 
(a) iteratively choosing a node with the maximum value of $(r_i w_{ig} + \max \{r_i w_{ib}-  r_{\hat{j}} w_{\hat{j}b} , 0\})$ according to the greedy approach or
(b) presorting the nodes as per decreasing values of $(r_i w_{ig} + \max \{r_i w_{ib}-  r_{\hat{j}} w_{\hat{j}b} , 0\})$ and then choosing nodes according to the greedy approach.
The time complexity of (a) would be $O(n k_g)$  and that of (b) would be $O(n \log n + k_g) = O(n \log n)$.
Hence, following (a) would be more efficient if $k_g << \log n$, while (b) would be better if $k_g >> \log n$ (else, asymptotically indifferent between the two).
So its time complexity is $O(n \cdot \min\{k_g , \log n \})$.

Now, the absolute optimal strategy is computed by iterating over all $\hat{j} \in \tilde{J}$. 
Based on (\ref{eqn:alloc}) and the definitions of $J_{\alpha},I_{\alpha},P_{\alpha}$,
node $\hat{j}$ can be viewed as a boundary for the bad camp's investment since it would not invest in any node $i$ such that $r_i w_{ib} < r_{\hat{j}} w_{\hat{j}b}$.
So if the nodes are ordered in decreasing order of $r_i w_{ib}$ values, such a node $\hat{j}$ would be in a position no later than $\lceil k_g+k_b \rceil$ (this limiting case is met if the good camp invests $k_g$ in nodes with the highest values of $r_i w_{ib}$, and then the bad camp invests $k_b$ in nodes with the highest values of $r_i w_{ib}$ which are not exhaustively invested on by the good camp).
So the possible $\hat{j}$'s are at most the top $\lceil k_g+k_b \rceil$ nodes in the sorted order of $r_i w_{ib}$, hence determining the possible $\hat{j}$'s requires $O(\min\{n k_b,n \log n\})$ time and the number of possible $\hat{j}$'s is $O(k_g+k_b)$.

Since  the absolute optimal strategy of the good camp is computed by iterating over all $\hat{j} \in \tilde{J}$, the overall time complexity of the greedy approach is 
$O(\min\{n k_b,n \log n\} + (k_g+k_b)\cdot n \cdot \min\{k_g , \log n \} )$, which is same as 
$O(n (k_g+k_b) \cdot \min\{k_g,\log n\})$.

\vspace{-3mm}
\subsection*{Maxmin versus Minmax Values}

Here, we compare the maxmin and minmax values of the game in the fundamental setting (Section~\ref{sec:maxmin_basic}) with that in the common coupled constraints setting.
%
%
%
The introduction of the total budget constraints per node disturbs the equality between maxmin and minmax, as we show now.
Let $(\mathbf{x'},\mathbf{y'})$ be an optimal maxmin strategy profile in (\ref{eqn:maxmin_C}).
Adding the constraint $\mathbf{0} \leq \mathbf{y} \leq \mathbf{1}-\mathbf{x}$ restricts the set of feasible strategies for the bad camp, and this set of feasible strategies and hence its optimal strategy now depends on $\mathbf{x}$. 
So we have

\vspace{-3mm}
\begin{small}
\begin{align*}
\max_{\mathbf{0} \leq \mathbf{x}  \leq \mathbf{1}} \; \min_{\mathbf{0} \leq \mathbf{y} \leq \mathbf{1}} \; \sum_{i \in N} v_i  \leq  \max_{\mathbf{0} \leq \mathbf{x}  \leq \mathbf{1}} \; \min_{\mathbf{0} \leq \mathbf{y} \leq \mathbf{1}-\mathbf{x}} \;  \sum_{i \in N} v_i 
\end{align*}
\end{small}
\vspace{-3mm}

\noindent
Similarly,

\vspace{-3mm}
\begin{small}
\begin{align*}
\min_{\mathbf{0} \leq \mathbf{y} \leq \mathbf{1}} \; \max_{\mathbf{0} \leq \mathbf{x}  \leq \mathbf{1}} \; \sum_{i \in N} v_i  \geq  \min_{\mathbf{0} \leq \mathbf{y}  \leq \mathbf{1}} \; \max_{\mathbf{0} \leq \mathbf{x} \leq \mathbf{1}-\mathbf{y}} \;  \sum_{i \in N} v_i
\end{align*}
\end{small}
\vspace{-3mm}

%
%
\noindent
These two inequalities, along with  (\ref{eqn:maxmin_C}), 
result in the following inequality,

\vspace{-3mm}
\begin{small}
\begin{align}
\max_{\mathbf{0} \leq \mathbf{x}  \leq \mathbf{1}} \; \min_{\mathbf{0} \leq \mathbf{y} \leq \mathbf{1}-\mathbf{x}} \;  \sum_{i \in N} v_i  \geq \min_{\mathbf{0} \leq \mathbf{y}  \leq \mathbf{1}} \; \max_{\mathbf{0} \leq \mathbf{x} \leq \mathbf{1}-\mathbf{y}} \;  \sum_{i \in N} v_i
\label{eqn:maxmin_totalC}
\end{align}
\end{small}
\vspace{-3mm}

This  result, which is contrary to general functions (for which maxmin is less than or equal to minmax), has also been derived in \cite{altman2009constrained}. 
In our problem, this is a direct consequence of the first mover advantage, which restricts the strategy set of the second mover.
In the maxmin case as analyzed earlier, the good camp invests in nodes with good values of $(r_i w_{ig} + \max \{r_i w_{ib}-  r_{\hat{j}} w_{\hat{j}b} , 0\})$ (assuming a $\hat{j}$). That is, it is likely to invest in nodes with good values of $r_i w_{ib}$ which are the preferred investees of the bad camp. Owing to total investment limit per node, the bad camp may not be able to invest in its preferred nodes (those with high values of $r_i w_{ib}$).
It can be shown on similar lines that, in the minmax case where the bad camp plays first, it would play symmetrically opposite, thus limiting the ability of good camp to invest in nodes with good values of $r_i w_{ig}$.

\begin{remark}[CCC under simultaneous play]
If instead of sequential play, the two camps play simultaneously under  CCC setting, it can be seen that the uniqueness of Nash equilibrium is not guaranteed (an immediate example is that the maxmin and minmax values could be different).
We address this precise question in \cite{dhamal2018resource} for general resource allocation games, albeit assuming strict preference ordering of the camps over nodes.
Therein, however, we do not derive an equilibrium strategy profile since there could be infinite number of Nash equilibria.
In order to derive a precise strategy profile which would be of practical and conceptual interest, we considered a sequential play in this paper and computed the subgame perfect Nash equilibrium.
For the case where the Nash equilibrium is unique, the sequence of play would not matter (that is, the maxmin and minmax values would be the same), 
and our derived subgame perfect Nash equilibrium would  be the unique pure strategy Nash equilibrium.
\label{rem:RAPG}
\end{remark}

\vspace{-3mm}
\section{Decision under Uncertainty}
\label{sec:robustness}

In this section, 
we  look at another sequential play setting which considers the possibility that the good camp, which plays first, may not have complete or exact information regarding the extrinsic weights, namely, $w_{ig}, w_{ib}, w_{ii}^0$. The bad camp, however, which plays second, has perfect information regarding the values of these parameters, and hence it is known that it would act optimally. Forecasting the optimal strategy of the bad camp, the good camp aims at choosing a robust strategy which would give it a good payoff even in the worst case.

Let $\mathbf{u}=$ \begin{scriptsize}$\left( \begin{array}{c} \mathbf{u_1} \\ \vdots \\ \mathbf{u_n} \\ \end{array} \right)$\end{scriptsize},
where $\mathbf{u_i}=$ \begin{scriptsize}$\left( \begin{array}{c} w_{ii}^0 \\ w_{ig} \\ w_{ib} \\ \end{array} \right)$\end{scriptsize}.
That is, $\mathbf{u}= \left( w_{11}^0 \; w_{1g} \; w_{1b} \; \cdots  \; w_{nn}^0  \; w_{ng} \; w_{nb} \right)^T$.

Let $U$ be a polytope defined by $E\mathbf{u} \leq f$ (that is, $\mathbf{u}\in U$). It can be viewed as the uncertainty set, which in this case, is a convex set. The polytope would be based on the application at hand and could be deduced from observations, predictions, etc. We use the framework of robust optimization \cite{ben2009robust} for solving this problem. 

For the purpose of this section,
let us assume that all the elements of $\mathbf{u}$ are non-negative. This is to ensure bounded values of the parameters. For instance, if we have a constraint in the linear program such as $w_{ig}+w_{ib}+w_{ii}^0 \leq \theta_i$, the individual values $w_{ig},w_{ib},w_{ii}^0$ can be unbounded. So for this current setting (under uncertain parameters), we will assume $w_{ig},w_{ib},w_{ii}^0 \geq 0, \forall {i \in N}$.

Since the good camp aims to optimize in the worst case of parameter values, while the bad camp has knowledge of these values with certainty, the optimization problem is

\vspace{-3mm}
\begin{small}
\begin{gather*}
\max_{\substack{\sum_i x_i \leq k_g \\ x_i \geq 0}} \;
\min_{E \mathbf{u} \leq f} \;
\min_{\substack{\sum_i y_i \leq k_b \\ y_i \geq 0}} \;
  \sum_{{i \in N}} r_i w_{ig} x_i + \sum_{i \in N} r_i w_{ii}^0 v_i^0 - \sum_{i \in N} r_i w_{ib} y_i
\end{gather*}
\end{small}
\vspace{-1mm}

If $\max_{j \in N} r_j w_{jb}>0$, that is, the bad camp has at least one feasible node to invest on, then we have $\sum_{i \in N} r_i w_{ib} y_i = k_b \max_{j \in N} r_j w_{jb}$, else we have $\sum_{i \in N} r_i w_{ib} y_i =0$.
For arriving at a concise solution, let $d$ be a dummy node such that $r_d w_{db}=0$. So now we have $\sum_{i \in N} r_i w_{ib} y_i = k_b \max_{j \in N \cup \{d\}} r_j w_{jb}$.
The optimization problem thus is

\vspace{-3mm}
\begin{small}
\begin{gather*}
\max_{\substack{\sum x_i \leq k_g \\ x_i \geq 0}} 
 \min_{\mathbf{u}} \sum_{{i \in N}} r_i w_{ig} x_i + \sum_{i \in N} r_i w_{ii}^0 v_i^0 - k_b \max_{j \in N \cup \{d\}} r_j w_{jb}
 \\ \text{s.t. }
 E \mathbf{u} \leq f
\end{gather*}
\end{small}
\vspace{-3mm}


\begin{figure*} 
\hspace{-7mm}
\begin{tabular}{|c||c||c|}
\hline 
 & & \vspace{-3mm} \\
\includegraphics[width=0.33\textwidth]{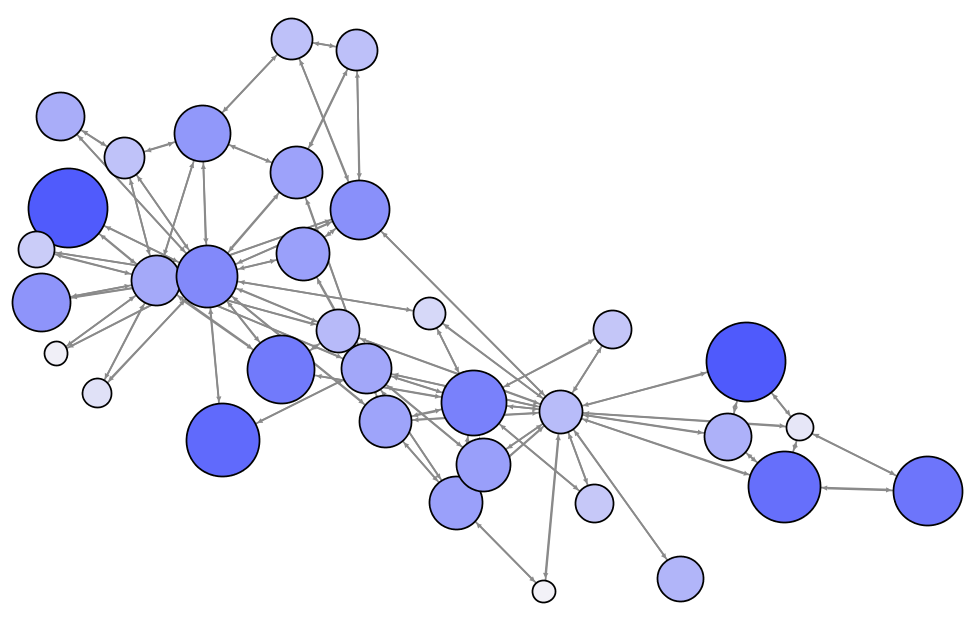} 
&
\includegraphics[width=0.33\textwidth]{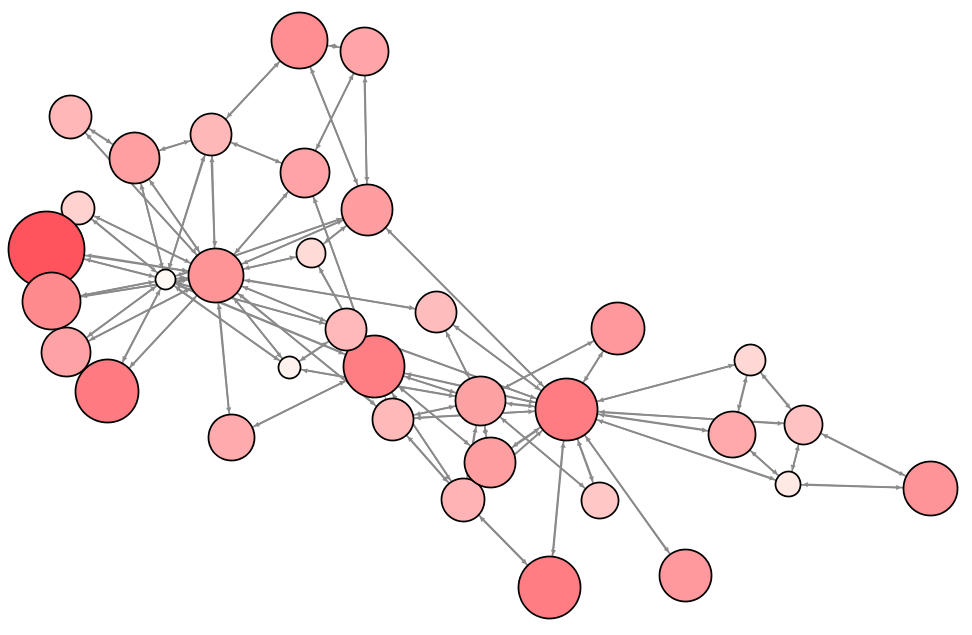}
&
\includegraphics[width=0.33\textwidth]{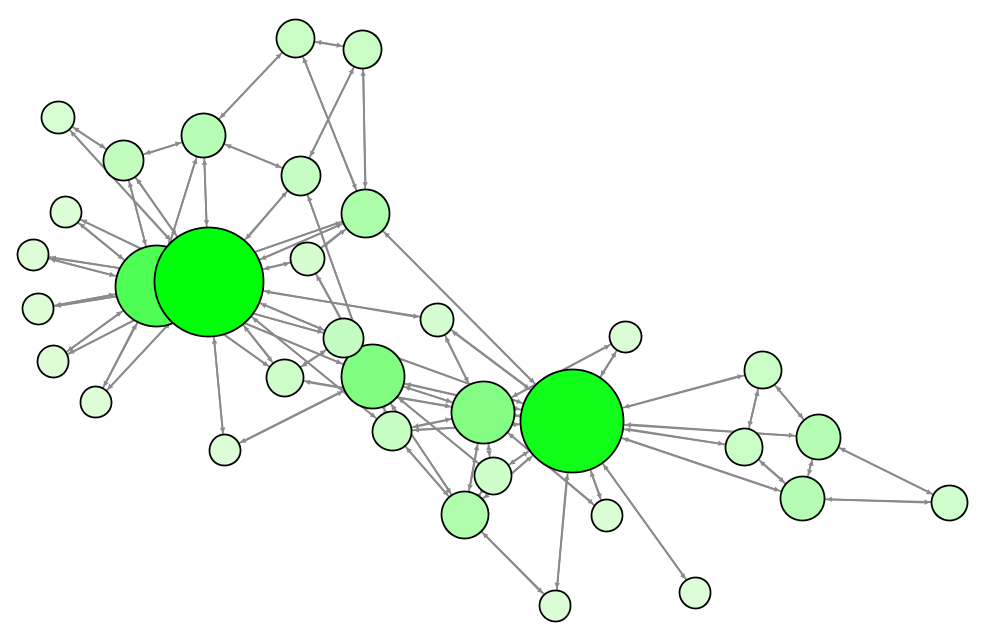} 
\\ & & \vspace{-3mm} \\ 
(a)
The value of $w_{ig}$
&
(b)
The value of $w_{ib}$
&
(d)
The value of $r_i$
\\
\hline
\end{tabular}
\caption{Details about the Karate club dataset used in our simulations. The size and color saturation of a node $i$ represent the value of the parameter.
}
\label{fig:karate_details}
\vspace{-2mm}
\end{figure*}

\begin{figure*} 
\hspace{-7mm}
\begin{tabular}{|c||c||c|}
\hline 
 & & \vspace{-3mm} \\
\includegraphics[width=0.33\textwidth]{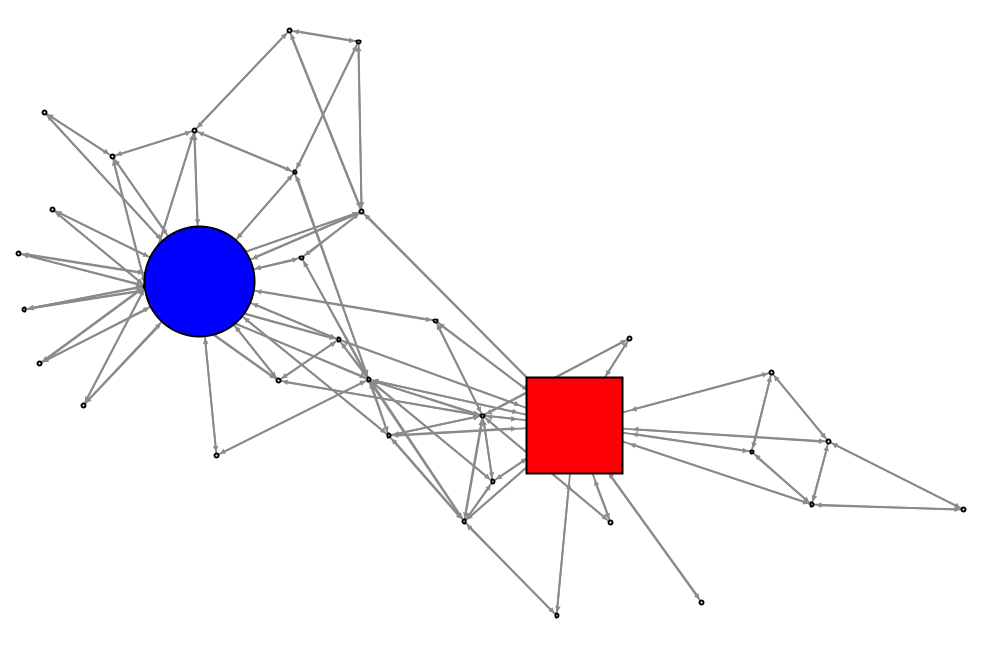} 
&
\includegraphics[width=0.33\textwidth]{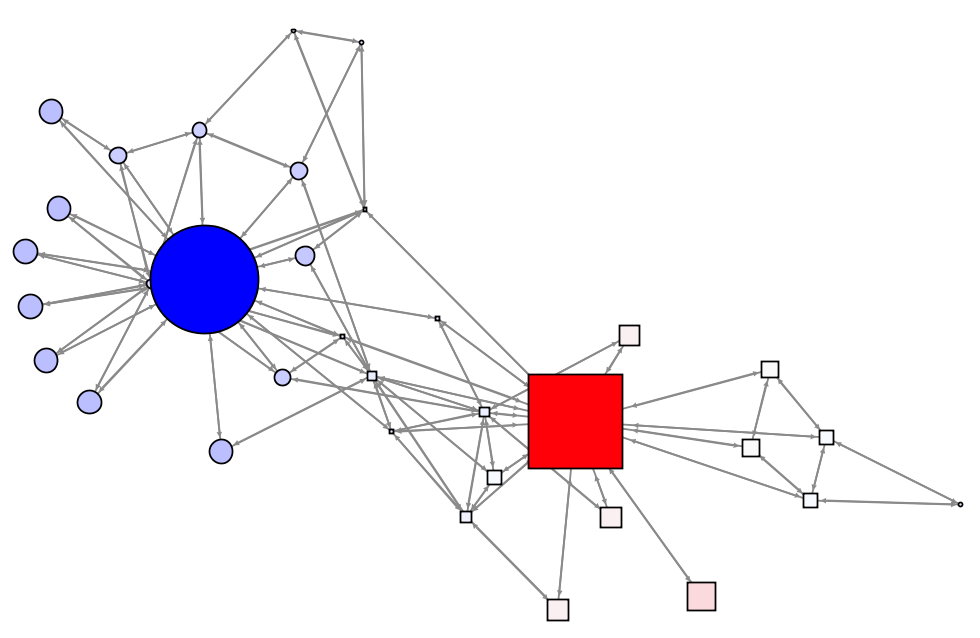}
&
\includegraphics[width=0.33\textwidth]{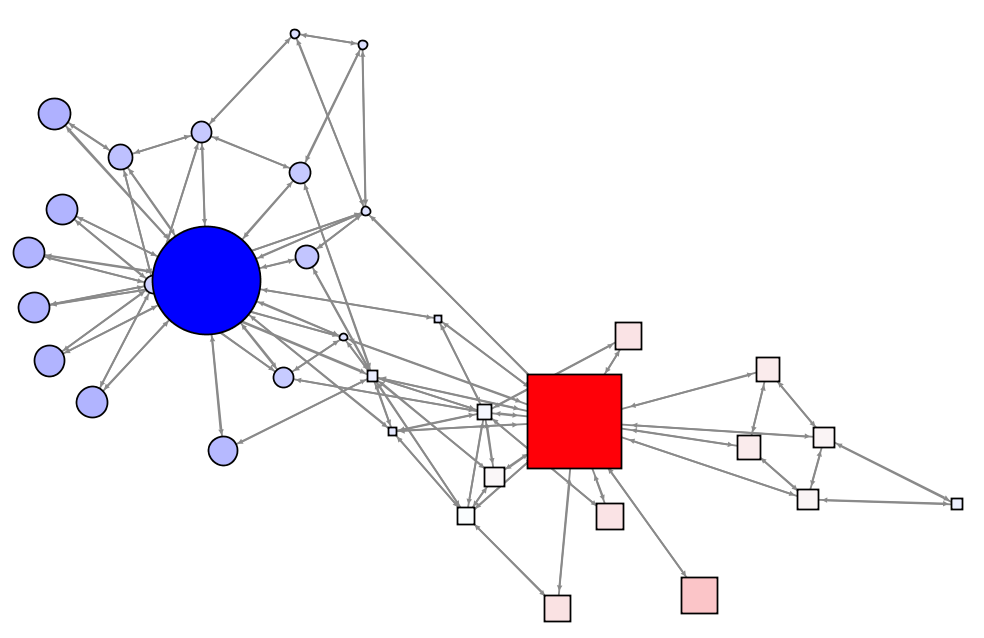} 
\\ & & \vspace{-3mm} \\ 
(a)
$\tau=1$ ($\max_{\mathbf{x}} \! \min_{\mathbf{y}} \! \sum_{i} v_i^{\langle \tau \rangle} \!=\! -0.0342$)
&
(b)
$\tau=2$ ($\max_{\mathbf{x}} \! \min_{\mathbf{y}} \! \sum_{i} v_i^{\langle \tau \rangle} \!=\! 0.2058$)
&
(d)
$\tau=4$ ($\max_{\mathbf{x}} \! \min_{\mathbf{y}} \! \sum_{i} v_i^{\langle \tau \rangle} \!=\! 0.1566$)
\\
\hline
\end{tabular}
\caption{Progression of opinion values for  Karate club dataset with $k_g\!=\!k_b\!=\!5$ under linear influence function with unbounded investment per node.
Opinions of nodes are signified by their shapes, sizes, color saturations (circular blue nodes: positive opinions, square red nodes: negative opinions)
}
\label{fig:karate_prog_linear}
\vspace{-3mm}
\end{figure*}

Note that there are in general $n+1$ possibilities for $\max_{j \in N \cup \{d\}} r_j w_{jb}$.
We could write a linear program for each possibility of $i_0 = \argmax_{j \in N \cup \{d\}} r_j w_{jb}$.
For a fixed $i_0\in N \cup \{d\}$, the inner term is

\vspace{-3mm}
\begin{small}
\begin{align*}
 \min_{\mathbf{u}} \sum_{{i \in N}} r_i w_{ig} x_i + \sum_{i \in N} r_i w_{ii}^0 v_i^0 - k_b  r_{i_0} w_{i_0 b}
 &
\\ \text{s.t. }
E \mathbf{u} \leq f
& \;\;\; \leftarrow \alpha_{i_0}
\\
\forall {i \in N}: \;
r_i w_{ib} \leq r_{i_0} w_{i_0 b}
& \;\;\; \leftarrow \beta_{i_0 i}
\end{align*}
\end{small}
\vspace{-3mm}

For this problem to be feasible, the constraint set should be non-empty. Let $N_f$ be the subset of $N \cup \{d\}$ consisting of nodes $i_0$ such that the constraint set satisfying $E \mathbf{u} \leq f$ and $\forall {i \in N}: 
r_i w_{ib} \leq r_{i_0} w_{i_0 b}$ is non-empty.
Its dual is:

\vspace{-3mm}
\begin{small}
\begin{align*}
\max -\alpha_{i_0}^T f
\\
\text{ s.t. }\;\;
\forall i\in N \setminus \{i_0\}: \;
-\alpha_{i_0}^T E_{ib} - \beta_{i_0 i} r_i \leq 0
& \;\;\; \leftarrow w_{ib}
\\
\text{if } i_0 \neq d:\;
-\alpha_{i_0}^T E_{i_0 b} + \sum_{i\neq i_0} \beta_{i_0 i} r_{i_0} \leq -k_b r_{i_0}
& \;\;\; \leftarrow w_{i_0 b}
\\
\forall {i \in N}: \;
-\alpha_{i_0}^T E_{ig} \leq r_i x_i
& \;\;\; \leftarrow w_{ig}
\nonumber
\\
\forall {i \in N}: \;
-\alpha_{i_0}^T E_{ii} \leq r_i v_i^0
& \;\;\; \leftarrow w_{ii}^0
\nonumber
\\
\alpha_{i_0} \geq \mathbf{0}
\nonumber
\\
\forall {i \in N}: \;\beta_{i_0 i} \geq 0
\end{align*}
\end{small}
\vspace{-3mm}

\noindent
We need to find a common $\mathbf{x}$ for all possibilities of \mbox{$i_0 \in N_f$}. So we have a constraint on the value of the dual, say $\rho$, namely, $\rho \leq -\alpha_{i_0}^T f, \forall i_0$.
We hence obtain a solution to the 
optimization problem by solving the following  LP.

~\vspace{-3mm}
\begin{small}
\begin{gather*}
\max \rho
\\
\text{s.t. }
\sum_{i \in N} x_i \leq k_g
\\
\forall {i \in N}:\; x_i \geq 0
\end{gather*}
\begin{empheq}[left= \forall i_0 \in N_f: \empheqlbrace]{align*}
\rho +\alpha_{i_0}^T f\leq 0
&
\\ 
\forall i\in N \setminus \{i_0\}: \;
-\alpha_{i_0}^T E_{ib} - \beta_{i_0 i} r_i \leq 0
& 
\\
\text{if } i_0 \neq d:\;
-\alpha_{i_0}^T E_{i_0 b} + \sum_{i\neq i_0} \beta_{i_0 i} r_{i_0} \leq -k_b r_{i_0}
& 
\\
\forall {i \in N}: \;
-\alpha_{i_0}^T E_{ig} - r_i x_i \leq 0
& 
\nonumber
\\
\forall {i \in N}: \;
-\alpha_{i_0}^T E_{ii} \leq r_i v_i^0
& 
\nonumber
\\
\alpha_{i_0} \geq \mathbf{0}
\nonumber
\\
\forall {i \in N}: \;\beta_{i_0 i} \geq 0
\nonumber
\end{empheq}
\end{small}

We solve the above LP for a specific example in our simulation study, so as to derive insights on the effect of uncertainty on the optimal strategy of the good camp.



\vspace{-3mm}
\section{Simulations and Results}
\label{sec:ODSN_sim}

Throughout this paper, we analytically derived the optimal investment strategies of competing camps in a social network for driving the opinion of the population in their favor. We hence presented either closed-form expressions or algorithms with polynomial running time. 
With the aim of determining implications of the analytically derived results on real-world networks and obtaining further insights, we conducted a simulation study on two popular network datasets. In this section, we present the setup and observed results, and provide insights behind them.

\vspace{-3mm}
\subsection{Simulation Setup}

%
We consider an academic collaboration network obtained from co-authorships in the ``High Energy Physics - Theory'' papers published on the e-print arXiv from 1991 to 2003. It contains 15,233 nodes and 31,376 links among them, and is popularly denoted as NetHEPT. This network exhibits many structural features of large-scale social networks  and is widely used for experimental justifications, for example, in \cite{kempe2003maximizing, chen2009efficient, chen2010scalable}.
For the purpose of graphical illustration, we use the popular Zachary's Karate club dataset  consisting of 34 nodes and 78 links among them \cite{zachary1977information}. 


Our analyses throughout the paper are valid for any distribution of edge weights satisfying the general constraints in Section \ref{sec:extended_model}.
However, in order to concretize our simulation study, we need to consider a particular  distribution of edge weights.
Proposition \ref{prop:WC} showed that some popular models of distributing edge weight in a graph would result in random strategy being optimal, and so are not suited for our simulations.
Hence in order to transform an undirected unweighted network dataset into a weighted directed one for our simulations,
we consider that
for any node $i$ (having  $d_i$ number of connections), the tuple $(w_{ii}^0 , w_{ig} , w_{ib})$ is randomly generated such that

\vspace{-3mm}
\begin{small}
\begin{gather*}
\forall {i \in N}:\; w_{ii}^0 + w_{ig} + w_{ib} = 0.5, \text{ and}
\\
w_{ij} = \frac{1\!-\!(w_{ii}^0 + w_{ig} + w_{ib})}{d_i}, \text{ if there is an edge between $i$ and $j$}.
\end{gather*}
\end{small}
\vspace{-3mm}

A primary reasoning for considering $w_{ii}^0+w_{ig}+w_{ib}=0.5$ is to have a natural first guess that nodes give 
equal weightage to intra-network influencing factors ($\{w_{ij}\}_{j\in N}$) and extra-network influencing factors ($w_{ii}^0,w_{ig},w_{ib}$).
We provide results for the extreme cases in Appendix~\ref{app:sim}, namely, when the value of $w_{ii}^0+w_{ig}+w_{ib}$ is $0.1$ or $0.9$ (with the values of the individual parameters scaled proportionally).
We also  highlight some key effects of this value on the obtained results, throughout this section.


Figure \ref{fig:karate_details} presents the values of the parameters used for our experimentation on the Karate club dataset, considering our weight distribution with $w_{ii}^0 + w_{ig} + w_{ib} = 0.5$. The size and color saturation of a node $i$ represent the value of the parameter mentioned in the corresponding caption (bigger size and higher saturation implies higher value).
Unless otherwise specified, we consider $k_g = k_b = 5$ for this dataset.
Also, unless otherwise specified, we start with an unbiased network, that is, $v_i^0=0, \forall {i \in N}$.

\begin{table*} 
\caption{Results for Karate club and NetHEPT datasets}
\vspace{-2mm}
\label{tab:karate_hep}
\centering
\begin{tabular}{|c|c|c|c|c|c|c|c|c|}
\hline 
\multicolumn{2}{|c|}{{Setting}} &  \multirow{3}{*}{Section} & \multicolumn{3}{c|}{\T \B Karate club} & \multicolumn{3}{c|}{\T \B NetHEPT}
\\ \cline{1-2}\cline{4-9} \T \B
Aspect & Case & & $k_g$ & $k_b$ & Results  & $k_g$ & $k_b$ & Results
\\ \hline \T \B
\multirow{4}{*}{Fundamental} & \multirow{3}{*}{Unbounded} & \multirow{3}{*}{\ref{sec:fundamental}} & 5 & 5 & $\max_{\mathbf{x}} \min_{\mathbf{y}} \sum_{i\in N} v_i = 0.1564$ 
& 100 & 100 & $\max_{\mathbf{x}} \min_{\mathbf{y}} \sum_{i\in N} v_i = 73.2539$ 
\\ \cline{4-9} \T \B
&  & & 10 & 5 & $\max_{\mathbf{x}} \min_{\mathbf{y}} \sum_{i\in N} v_i = 5.8811$ 
& 200 & 100 & $\max_{\mathbf{x}} \min_{\mathbf{y}} \sum_{i\in N} v_i = 347.0770$ 
\\ \cline{2-9} \T \B
& Bounded & \ref{sec:fundamental_bounded} & 5 & 5 & $\max_{\mathbf{x}} \min_{\mathbf{y}} \sum_{i\in N} v_i = -0.0538$ 
& 100 & 100 & $\max_{\mathbf{x}} \min_{\mathbf{y}} \sum_{i\in N} v_i = 2.8513$ 
\\ \hline \T \B 
\multirow{3}{*}{Adversary} &
Unbounded & \ref{sec:minmax_investment_unbounded} & 5 & - & $\max_{\mathbf{x}} \min_{\mathbf{y}} \sum_{i\in N} y_i = 5.1404$ 
& 100 & - & $\max_{\mathbf{x}} \min_{\mathbf{y}} \sum_{i\in N} y_i = 136.5231$ 
\\ \cline{2-9} \T \B
& Bounded & \ref{sec:minmax_investment_bounded} & 5 & - & $\max_{\mathbf{x}} \min_{\mathbf{y}} \sum_{i\in N} y_i = 4.8936$ 
& 100 & - & $\max_{\mathbf{x}} \min_{\mathbf{y}} \sum_{i\in N} y_i = 102.7266$ 
\\ \hline \T \B
\multirow{6}{*}{Concave ($t=2$)} &
\multirow{3}{*}{Unbounded} &
\multirow{3}{*}{\ref{sec:concave_unbounded}} & \multirow{1}{*}{5} & \multirow{1}{*}{5} &  $\max_{\mathbf{x}} \min_{\mathbf{y}} \sum_{i\in N} v_i = 0.4581$ 
& \multirow{1}{*}{100} & \multirow{1}{*}{100} &  $\max_{\mathbf{x}} \min_{\mathbf{y}} \sum_{i\in N} v_i = -0.8446$ 
\\ \cline{4-9} \T \B
&  &  & 20 & 20 &  $\max_{\mathbf{x}} \min_{\mathbf{y}} \sum_{i\in N} v_i = 0.9163$ 
& 400 &  400 &  $\max_{\mathbf{x}} \min_{\mathbf{y}} \sum_{i\in N} v_i = -1.6892$ 
\\ \cline{2-9} \T \B
& \multirow{3}{*}{Bounded} & \multirow{3}{*}{\ref{sec:concave_bounded}} & \multirow{1}{*}{5} & \multirow{1}{*}{5} & $\max_{\mathbf{x}} \min_{\mathbf{y}} \sum_{i\in N} v_i = 0.4612$ 
& \multirow{1}{*}{100} & \multirow{1}{*}{100} & $\max_{\mathbf{x}} \min_{\mathbf{y}} \sum_{i\in N} v_i = -0.8446$ 
\\ \cline{4-9} \T \B
&  &  & 20 & 20 &  $\max_{\mathbf{x}} \min_{\mathbf{y}} \sum_{i\in N} v_i = 1.2653$ 
& 400 &  400 &  $\max_{\mathbf{x}} \min_{\mathbf{y}} \sum_{i\in N} v_i = -1.7117$ 
\\ \hline \T \B
\multirow{6}{*}{Concave ($t=10$)} &
\multirow{3}{*}{Unbounded} &
\multirow{3}{*}{\ref{sec:concave_unbounded}} & \multirow{1}{*}{5} & \multirow{1}{*}{5} &  $\max_{\mathbf{x}} \min_{\mathbf{y}} \sum_{i\in N} v_i = 1.1180$ 
& \multirow{1}{*}{100} & \multirow{1}{*}{100} & $\max_{\mathbf{x}} \min_{\mathbf{y}} \sum_{i\in N} v_i = -12.4212$ 
\\ \cline{4-9} \T \B
&  &  & 20 & 20 &  $\max_{\mathbf{x}} \min_{\mathbf{y}} \sum_{i\in N} v_i = 1.2842$  
& 400 & 400 &  $\max_{\mathbf{x}} \min_{\mathbf{y}} \sum_{i\in N} v_i = -14.2682$ 
\\ \cline{2-9} \T \B
& \multirow{3}{*}{Bounded} & \multirow{3}{*}{\ref{sec:concave_bounded}} & \multirow{1}{*}{5} & \multirow{1}{*}{5} & $\max_{\mathbf{x}} \min_{\mathbf{y}} \sum_{i\in N} v_i =1.1180 $ 
& \multirow{1}{*}{100} & \multirow{1}{*}{100} &  $\max_{\mathbf{x}} \min_{\mathbf{y}} \sum_{i\in N} v_i =-12.4212$ 
\\ \cline{4-9} \T \B
&  &  & 20 & 20 &  $\max_{\mathbf{x}} \min_{\mathbf{y}} \sum_{i\in N} v_i = 1.3104$ 
& 400 & 400 &  $\max_{\mathbf{x}} \min_{\mathbf{y}} \sum_{i\in N} v_i = -14.2682$ 
 \\ \hline \T \B
 \multirow{3}{*}{CCC} & \multirow{3}{*}{Bounded} &
\multirow{3}{*}{\ref{sec:CCC}} & \multirow{3}{*}{5} & \multirow{3}{*}{5} & $\max_{\mathbf{x}} \min_{\mathbf{y}} \sum_{i\in N} v_i = 1.5399$ 
& \multirow{3}{*}{100} & \multirow{3}{*}{100} & $\max_{\mathbf{x}} \min_{\mathbf{y}} \sum_{i\in N} v_i = 7.4843$ 
\\  \T \B
 &  &  & & & $\min_{\mathbf{y}} \max_{\mathbf{x}} \sum_{i\in N} v_i = -0.8900$ 
 & & & $\min_{\mathbf{y}} \max_{\mathbf{x}} \sum_{i\in N} v_i = -3.6795$ 
\\ \hline
\end{tabular}
\vspace{-3mm}
\end{table*}

\subsubsection*{Progression with time}

Throughout this paper, our analyses were based on the steady state opinion values. However for the purpose of completion, we now provide a brief note on the progression of opinion values with time, which occurs according to our update rule given by (\ref{eqn:update_rule}).

Figure \ref{fig:karate_prog_linear}
illustrates the progression of opinion values with time under the linear influence function and unbounded investment per node,
where $\sum_{i\in N} v_i^{\langle \tau \rangle}$ is the sum of opinion values of nodes in time step $\tau$.
The network starts with $v_i^0=0, \forall i \in N$, and then at $\tau=1$, the good and bad camps invest their entire budgets on their respective target nodes having maximum values of $r_i w_{ig}$ and $r_i w_{ib}$, respectively. Hence the opinion values of these nodes change to being highly positive and negative respectively, while other nodes still hold an opinion value of $0$ (Figure~\ref{fig:karate_prog_linear}(a)). At $\tau=2$, nodes which are directly connected to these target nodes, update their opinions; as seen from Figure~\ref{fig:karate_prog_linear}(b), nodes in the left region hold positive opinions while the ones in the right region hold negative opinions. Few nodes like the ones on the top, center, and extreme right regions, still hold an opinion value of $0$. By $\tau=3$, all nodes hold a non-zero opinion value and at $\tau=4$, the individual opinion values and hence the sum of opinion values almost reach the convergent value. The sum of opinion values at $\tau=4$ is $0.1566$ (Figure~\ref{fig:karate_prog_linear}(c)), while the convergent sum is $0.1564$ (Table~\ref{tab:karate_hep}).
In general, assuming the threshold of convergence to be $10^{-4}$, the convergence is reached in 8-10 time steps
for the Karate club dataset, and in 12-15 time steps for the NetHEPT dataset.

An important insight is that any significant changes in the opinion value of a node occur in the earlier time steps.
For instance, a node which is geodesically closer to the target nodes receive influence from them in the earlier time steps, and also the influence is strong since the entries of the substochastic weight matrix $\mathbf{w}^{\tau}$ are significantly higher for lower values of $\tau$. Hence, owing to the substochastic nature of $\mathbf{w}^{\tau}$, the change in a node's opinion value is insignificant at a later time step.
%
%
We could also observe that the rate of convergence depends on the investment strategies, for instance, the convergence is almost immediate under the concave influence function setting where the investment is already distributed over nodes (see \mbox{Figures~\ref{fig:karate_prog_concave_2}-\ref{fig:karate_prog_concave_10}} in Appendix~\ref{app:sim}); since this investment is made in the earliest time step, it plays a significant role in determining the nodes' convergent opinion values.
Also, it is usually observed that the individual opinion values as well as sum of opinion values are, more often than not, monotone increasing or decreasing with time (see \mbox{Figures~\ref{fig:karate_prog_concave_2}-\ref{fig:karate_prog_minmax}} in Appendix~\ref{app:sim}).

\begin{figure*}
\hspace{-7mm}
\begin{tabular}{|c||c||c|}
\hline &~ &~ \vspace{-3mm} \\
\includegraphics[width=0.32\textwidth]{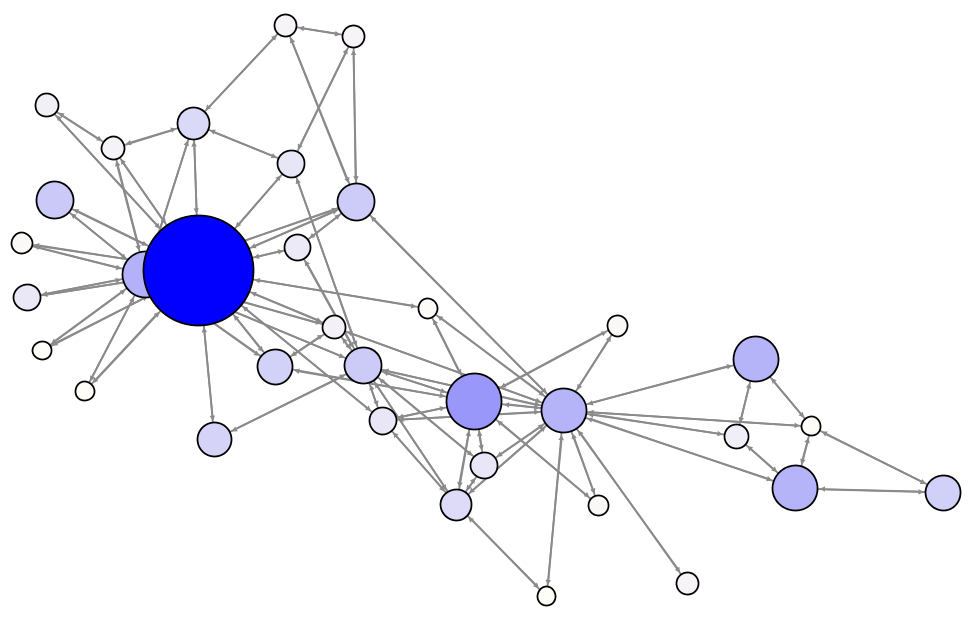} 
&
\includegraphics[width=0.32\textwidth]{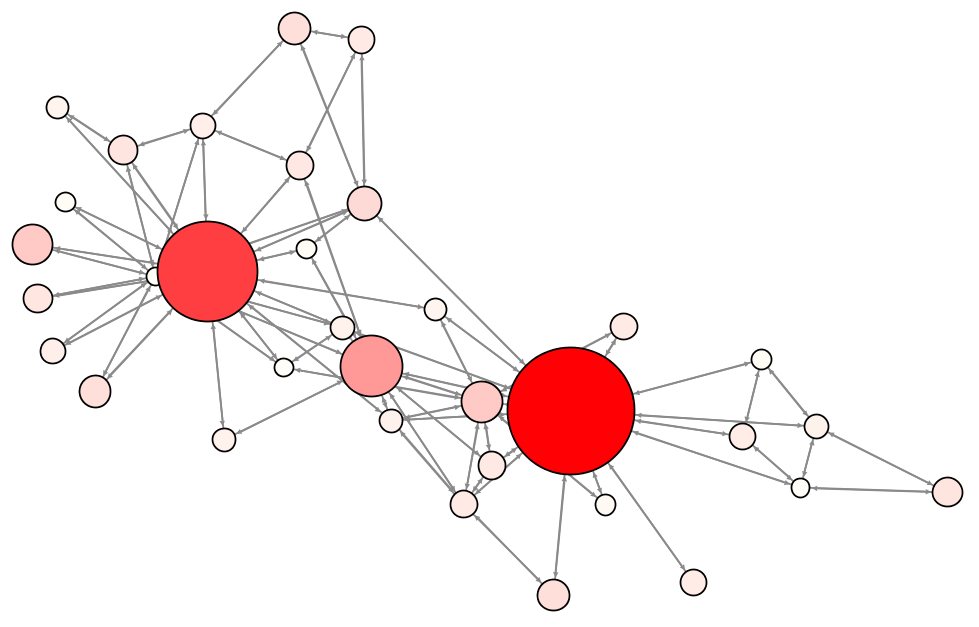}
&
\includegraphics[width=0.32\textwidth]{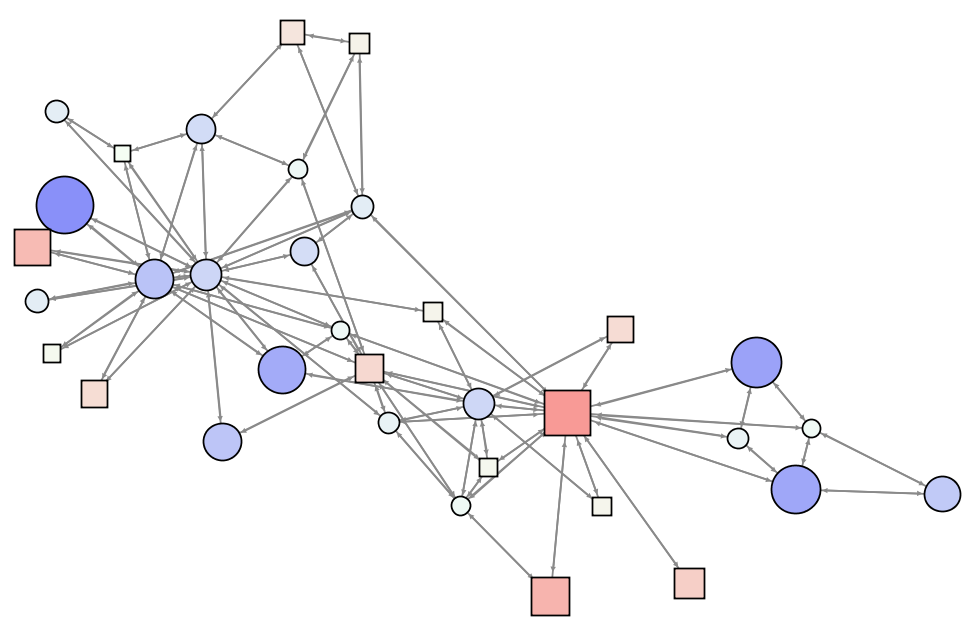} 
\\
&
&
(c) Final opinion values of nodes signified 
\\
(a) Investment made by  good camp 
&
(b) Investment made by  bad camp 
&
$ \! \! \! $ by  shapes, sizes,  color saturations for $t=2$ $ \! \! \! $
\\ 
on nodes signified by their sizes
&
on nodes signified by their sizes 
&
  (Circular blue nodes: positive opinions, 
\\
 and color saturations for $t=2$
 &
  and color saturations for $t=2$
  &
    square red nodes: negative opinions)
  \\
\hline
\hline &~ & \vspace{-3mm}\\
\includegraphics[width=0.32\textwidth]{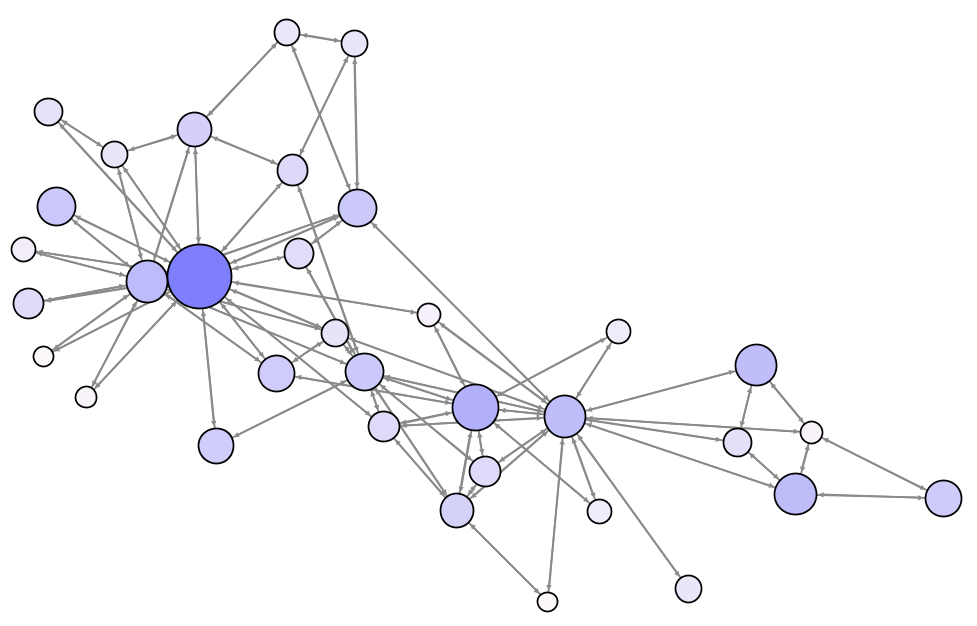} 
&
\includegraphics[width=0.32\textwidth]{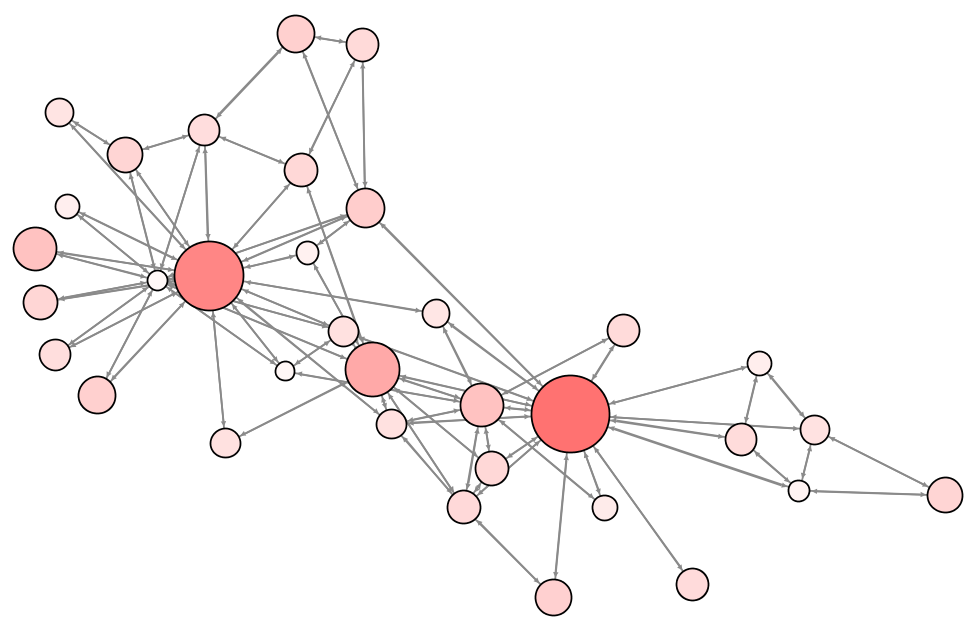}
&
\includegraphics[width=0.32\textwidth]{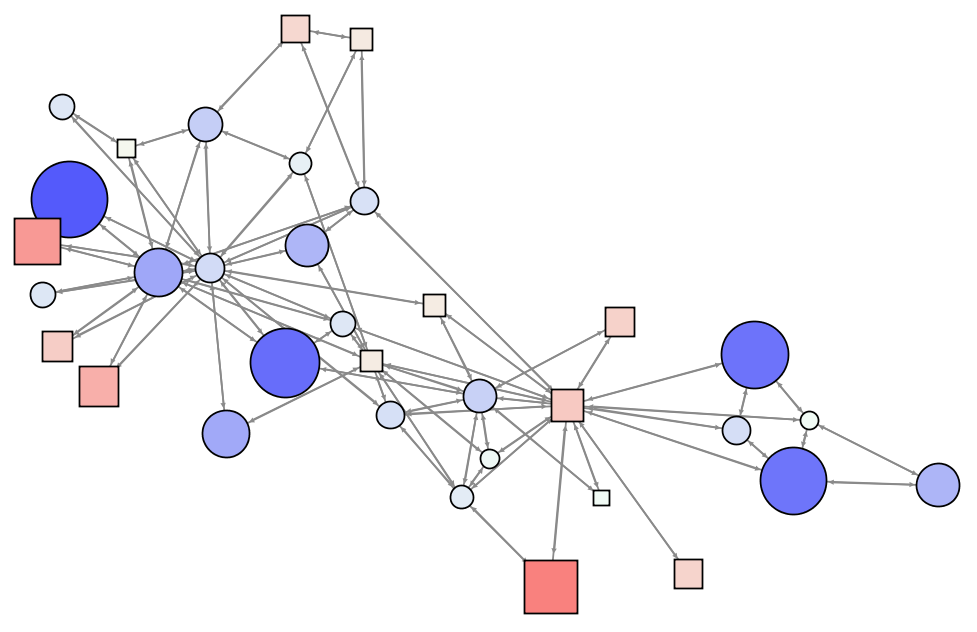}
\\
&&
(f) Final opinion values of nodes signified
\\
(d) Investment made by  good camp 
&
(e) Investment made by  bad camp 
&
$ \! \! \! $ by shapes, sizes, color saturations for $t=10$ $ \! \! \! $
\\ 
on nodes signified by their sizes 
&
on nodes signified by their sizes 
&
  (Circular blue nodes: positive opinions, 
\\
and color saturations for $t=10$
&
        and color saturations for $t=10$
        &
            square red nodes: negative opinions)
        \\
\hline
\end{tabular}
\caption{Simulation results for the Karate club dataset with $k_g=k_b=5$ when the influence function is concave (Setting \ref{sec:concave_unbounded})
}
\label{fig:concave}
\vspace{-3mm}
\end{figure*}

\vspace{-2mm}
\subsection{Simulation Results}

Table \ref{tab:karate_hep} presents the quantitative results of our simulations on Karate club and NetHEPT datasets.
(The results for $w_{ii}^0+w_{ig}+w_{ib} = 0.1$ and $0.9$ are provided in Tables~\ref{tab:karate_W0WgWb} and \ref{tab:hep_W0WgWb} of Appendix~\ref{app:sim}.)

\subsubsection{The fundamental setting}
For both the datasets, the overall opinion value is positive in the fundamental unbounded setting (\ref{sec:fundamental}) when both camps had the same budget (first row of Table \ref{tab:karate_hep}), which implies that the maximum value of $r_i w_{ig}$ in the network was higher than the maximum value of $r_i w_{ib}$ (as clear from Equation (\ref{eqn:steadystate}) and Proposition \ref{prop:fundamental_unbounded}). The results of the good camp doubling its budget can also be seen (second row of Table \ref{tab:karate_hep}). When both camps had the same budget, their overall influences tend to nullify each other to a great extent and so the sum of opinion values is neither exceedingly positive nor exceedingly negative. However, with the good camp doubling its budget, this additional budget could be used to have a large surplus of positive influence in the network. Actually, owing to close competition, even a slight imbalance in the camps' budgets would result in significant skewness in the overall opinion of the network.

The effect of bounded investment (Setting \ref{sec:fundamental_bounded}) can also be seen (third row of Table \ref{tab:karate_hep}); for these particular datasets, the maxmin value decreases implying that the value of $r_i w_{ig}$ is probably concentrated on one node, while that of $r_i w_{ib}$ is well distributed, thus giving the bad camp an advantage in bounded (and hence distributed) investment.
This can be seen from Proposition \ref{prop:fundamental_bounded}.
%
%
Figure \ref{fig:CCC}(a) illustrates the case with bounded investment  per node ($x_i,y_i \leq 1, \forall {i \in N}$) for the Karate club dataset.
Here, the label `c' means that, that node is invested on by both the camps with 1 unit, while `g' or `b' mean that the node is invested on by the good or bad camp, respectively.



\subsubsection{The adversary setting}

Considering the adversary setting of maximizing the competitor's (bad camp's) investment, we could see how much budget the bad camp required to draw the overall opinion in its favor. As expected from the results of the above setting where the overall opinion value turns out to be positive when the camps have the same budget, the investment required in the unbounded case of Setting \ref{sec:minmax_investment_unbounded} ($5.1404$ for Karate club and $136.5231$ for NetHEPT) is more than the budget available in Setting \ref{sec:fundamental} ($5$ for Karate club and $100$ for NetHEPT). In the bounded setting (\ref{sec:minmax_investment_bounded}), for Karate club dataset, we can see that the bad camp could have driven the overall opinion in its favor by expending $4.8936$ instead of  its entire budget of $5$.
For NetHEPT dataset, as expected from the results of the fundamental setting where the overall opinion in the bounded setting (\ref{sec:fundamental_bounded}) was less than that in the unbounded setting (\ref{sec:fundamental}), the investment required by the bad camp in the bounded setting (\ref{sec:minmax_investment_bounded}), which is $102.7266$, is less than that required in the unbounded setting (\ref{sec:minmax_investment_unbounded}), which is $136.5231$.
%



\subsubsection{Concave influence function}
Results under the concave influence function  are presented 	for both  unbounded (\ref{sec:concave_unbounded}) and bounded (\ref{sec:concave_bounded}) cases for $t=2$ and $t=10$, in Table~\ref{tab:karate_hep}.
%
%
Figure \ref{fig:concave} shows the effect of the value of $t$ on the distribution of investment and  final opinion values for the Karate club dataset, in the unbounded case. 
In Figures \ref{fig:concave}(a),(b),(d),(e), the size and color saturation of a node represent the amount of investment on it by the camp mentioned in the corresponding caption.
With careful observation, it can be seen that for $t=2$, the investments are more skewed, while for $t=10$, the investments by the good and bad camps on a node $i$ are close to being proportional to the values of $r_i w_{ig}$ and $r_i w_{ib}$, respectively (as suggested in Proposition \ref{prop:concave_unbounded} and Remark~\ref{rem:skewness}).
In Figures \ref{fig:concave}(c) and \ref{fig:concave}(f), the shape and color of a node  represent its opinion sign (blue circle implies positive, red square implies negative), while
its size and color saturation represent its opinion magnitude.

\begin{figure*}
\hspace{-7mm}
\begin{tabular}{|c||c||c|}
\hline &~ & \vspace{-3mm}\\
\includegraphics[width=0.33\textwidth]{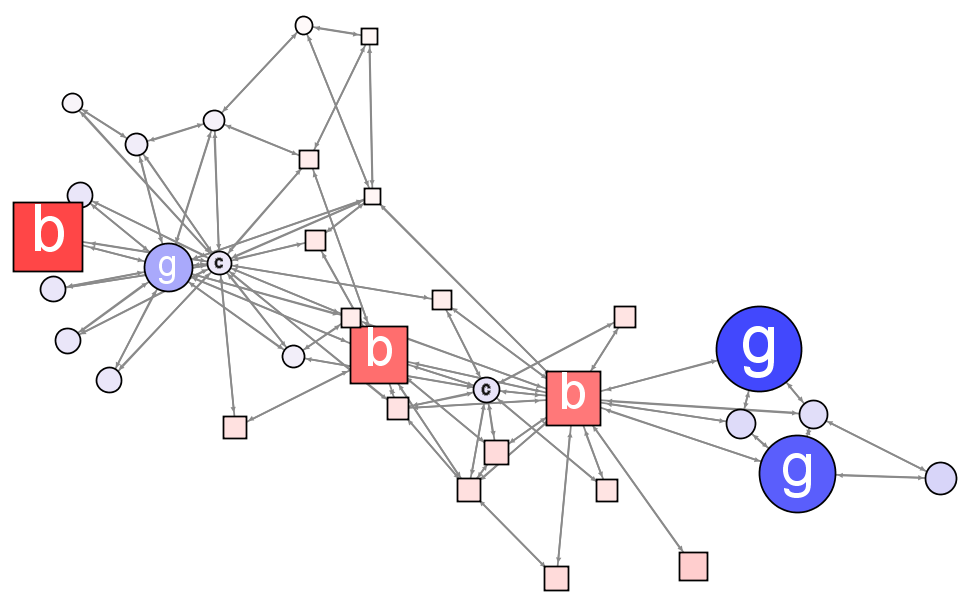}
&
\includegraphics[width=0.33\textwidth]{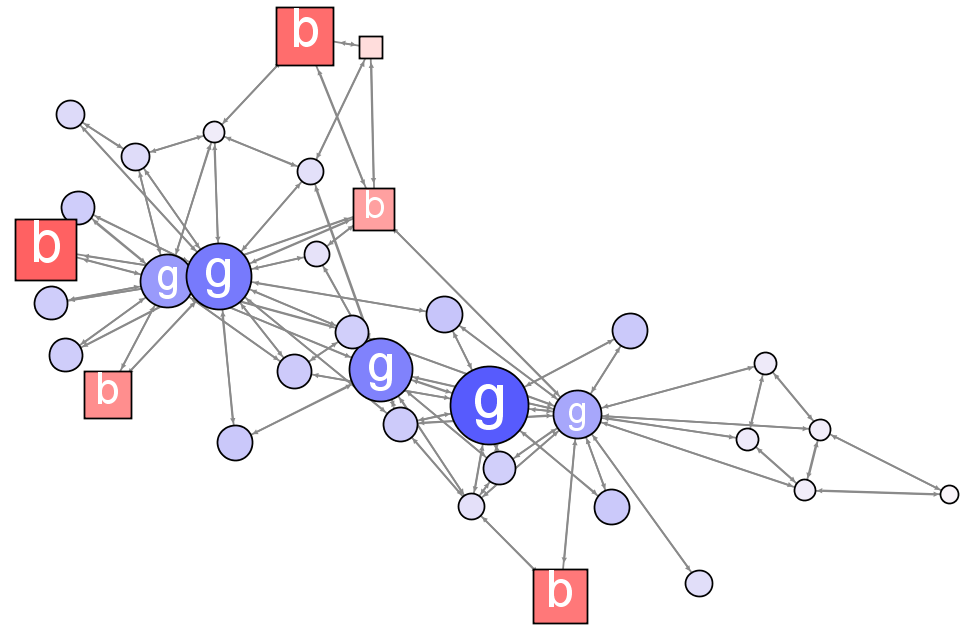}
&
\includegraphics[width=0.33\textwidth]{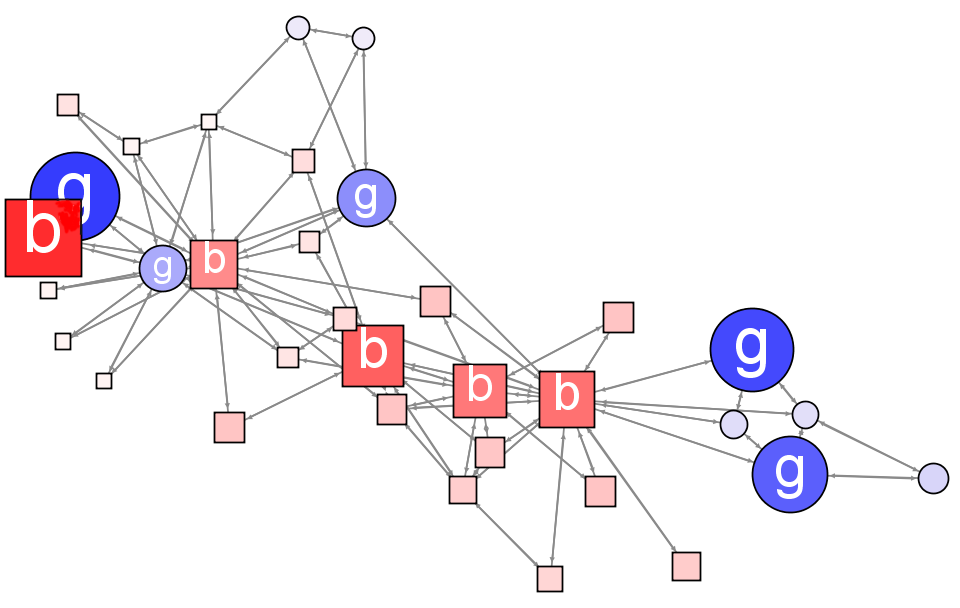}
\\ & & \vspace{-3mm} \\
(a) Bound on investment per node 
&
(b) Common coupled constraints 
&
(c) Common coupled constraints 
\\
 $x_i,y_i\leq 1, \forall i$ 
 &
$x_i+y_i \leq 1, \forall i$ 
&
$x_i+y_i \leq 1, \forall i$ 
\\ 
($\max_{\mathbf{x}} \min_{\mathbf{y}} \sum_i v_i = -0.0538$)
&
($\max_{\mathbf{x}} \min_{\mathbf{y}} \sum_i v_i = 1.5399$)
&
($\min_{\mathbf{y}} \max_{\mathbf{x}} \sum_i v_i = -0.8900$)
\\
\hline 
\end{tabular}
\caption{Results in presence of additional constraints  for the Karate club dataset with $k_g=k_b=5$ ; The nodes are labeled `g/b/c' to signify if invested on by good/bad/both camps respectively. The sign of the opinion value of a node is signified by its shape and color (circle and blue for good, square and red for bad), while the absolute value of its opinion is signified by its size and color saturation.}
\label{fig:CCC}
\vspace{-3mm}
\end{figure*}

In the unbounded case,  for some values of budgets, there exist nodes for which either $x_i$ or $y_i$ or both exceed 1 unit. So for the bounded case, the camps are directed to have different investment strategies than in the unbounded case (as can be understood from Insight~\ref{int:trialnerror}).
The effects can be seen in Table \ref{tab:karate_hep} where the values are different in Settings \ref{sec:concave_unbounded} and \ref{sec:concave_bounded} for the same values of budgets.
In some scenarios such as NetHEPT with $k_g=k_b=100$, however, the investment strategies inadvertently assured $x_i,y_i \leq 1, \forall {i \in N}$ even for the unbounded case; so the investment strategies remain the same in both settings and hence resulted in the same overall opinion value ($-0.8446$ for $t=2$ and $-12.4212$ for $t=10$).
%
A careful analysis of the values would indicate that the constraints $x_i,y_i \leq 1$ are likely to come into picture for some nodes, for lower values of $t$ and higher values of budgets. 
Lower values of $t$ lead to skewed investment and so a higher likelihood of some nodes having investment more than 1 unit in the unbounded case. Similarly, higher values of budgets scale up the investments on the nodes, resulting in a higher likelihood of some nodes having investment more than 1 unit in the unbounded case.
This can also be inferred from our analytically derived investment strategies.


For $w_{ii}^0+w_{ig}+w_{ib}=0.1$ (see Figure~\ref{fig:concave_0.1} of  Appendix~\ref{app:sim}), the values of extra-network parameters ($w_{ii}^0,w_{ig},w_{ib}=0.1$) get scaled down proportionally, while the intra-network ones ($w_{ij}$'s) scale up; that is, the network influence plays a  stronger role than the  camps' recommendations.
Because of the weak impact of the camps' recommendations, that is, the nodes being unwilling to accept opinions that are external to the network, and since all nodes started with a zero initial opinion, the magnitudes of their opinion values are very low.
Furthermore, because of the strong impact of the intra-network influence, the camps invest greatly on the most influential nodes (and not so much on the lesser influential ones), and allow the network to spread its influence.
Owing to the low magnitudes of opinion values and 
a strong network effect which aids in distributing the influence evenly, the opinion values in the network are less skewed.
For $w_{ii}^0+w_{ig}+w_{ib}=0.9$ (see Figure~\ref{fig:concave_0.9} of Appendix~\ref{app:sim}), the above reasoning gets reversed, and so the magnitudes of individual opinion values are high, the camps' investments are more evenly distributed, and the opinion values are highly skewed.



\subsubsection{Common coupled constraints}
%
Figures \ref{fig:CCC} (b-c) illustrate the effect of common coupled constraints $x_i+y_i \leq 1, \forall {i \in N}$, for the Karate club dataset. 
The advantage of playing first is clearly visible from the overall opinion value as well as the distribution of opinion values.
Specifically, the overall opinion value is the highest in Figure \ref{fig:CCC}(b) with a healthy distribution of positive opinion values, followed by the value in Figure \ref{fig:CCC}(a), followed by that in Figure \ref{fig:CCC}(c) which is dominated by negative opinions.


For $w_{ii}^0 + w_{ig} + w_{ib} = 0.9, \forall i \in N$ on NetHEPT dataset, we observed that $\max_{\mathbf{x}} \min_{\mathbf{y} \leq \mathbf{1}-\mathbf{x}} \sum_{i\in N} v_i = \min_{\mathbf{y}} \max_{\mathbf{x} \leq \mathbf{1}-\mathbf{y}} \sum_{i\in N} v_i = 1.5930$ (see Table \ref{tab:hep_W0WgWb} of Appendix~\ref{app:sim}). So the camps did not have the first mover advantage; and it would not matter if the camps played sequentially or simultaneously (Remark~\ref{rem:RAPG}).
%
Hence the camp which plays first, could devise its investment strategy without having to consider the best response investment strategy of the camp which plays second.
Also, as reasoned earlier, the opinion values are less skewed for $w_{ii}^0 + w_{ig} + w_{ib} = 0.1$, while highly skewed for $w_{ii}^0 + w_{ig} + w_{ib} = 0.9$.

\subsubsection{Decision under uncertainty}
For studying the effects of decision under uncertainty as analyzed in Section \ref{sec:robustness}, we consider that the good camp (first mover) is uncertain about  parameters $w_{ig}$ and $w_{ib}$ up to a certain limit. In particular, there is a fractional uncertainty of $\epsilon_{l}$ regarding the values of these parameters, while there is a fractional uncertainty of $\epsilon_{o}$ regarding the values of the sums of these parameters over the entire network. 
$\epsilon_{l}$ can be hence viewed as local uncertainty and $\epsilon_{o}$ as global uncertainty. 
Let $\hat{w}_{ig}$ and $\hat{w}_{ib}$ be the underlying ground truth values for a node $i$ (the actual values destined to be realized).

\vspace{-5.5mm}
\begin{small}
\begin{gather*}
(1-\epsilon_{l})  \hat{w}_{ig} \leq  w_{ig}  \leq (1+\epsilon_{l})  \hat{w}_{ig}
\\
(1-\epsilon_{l})  \hat{w}_{ib} \leq  w_{ib}  \leq (1+\epsilon_{l})  \hat{w}_{ib}
\\
(1-\epsilon_{o}) \sum_{i \in N} \hat{w}_{ig} \leq \sum_{i \in N} w_{ig}  \leq (1+\epsilon_{o}) \sum_{i \in N} \hat{w}_{ig}
\\
(1-\epsilon_{o}) \sum_{i \in N} \hat{w}_{ib} \leq \sum_{i \in N} w_{ib}  \leq (1+\epsilon_{o}) \sum_{i \in N} \hat{w}_{ib}
\end{gather*}
\end{small}
\vspace{-2mm}

It is clear that the latter two constraints would come into picture only if $\epsilon_{o} < \epsilon_{l}$ (this would usually be the case since, though there may be significant relative deviation for the individual parameters, the relative deviation of their sum is usually low owing to significant balancing of positive and negative deviations of the individual parameters). 
For different values of $\epsilon_l$ and $\epsilon_o$, Figure \ref{fig:uncertainty_3d} presents the maxmin values: 
(a) as computed by the good camp (first mover) as its worst case value using our robust optimization approach and
(b) as realized based on the ground truth.
We assume the ground truth values as depicted in Figure \ref{fig:karate_details}.
It can be seen that for a large enough range of values of $\epsilon_{l}$ and $\epsilon_{o}$, though the good camp computes the worst case maxmin value to be very low, the corresponding realized value is the same as when the good camp is certain about  parameter values. The uncertainty factor starts affecting the good camp only for very high  $\epsilon_{l}$ and $\epsilon_{o}$.

 \begin{figure}[t]
 \centering
 \includegraphics[width=0.43\textwidth]{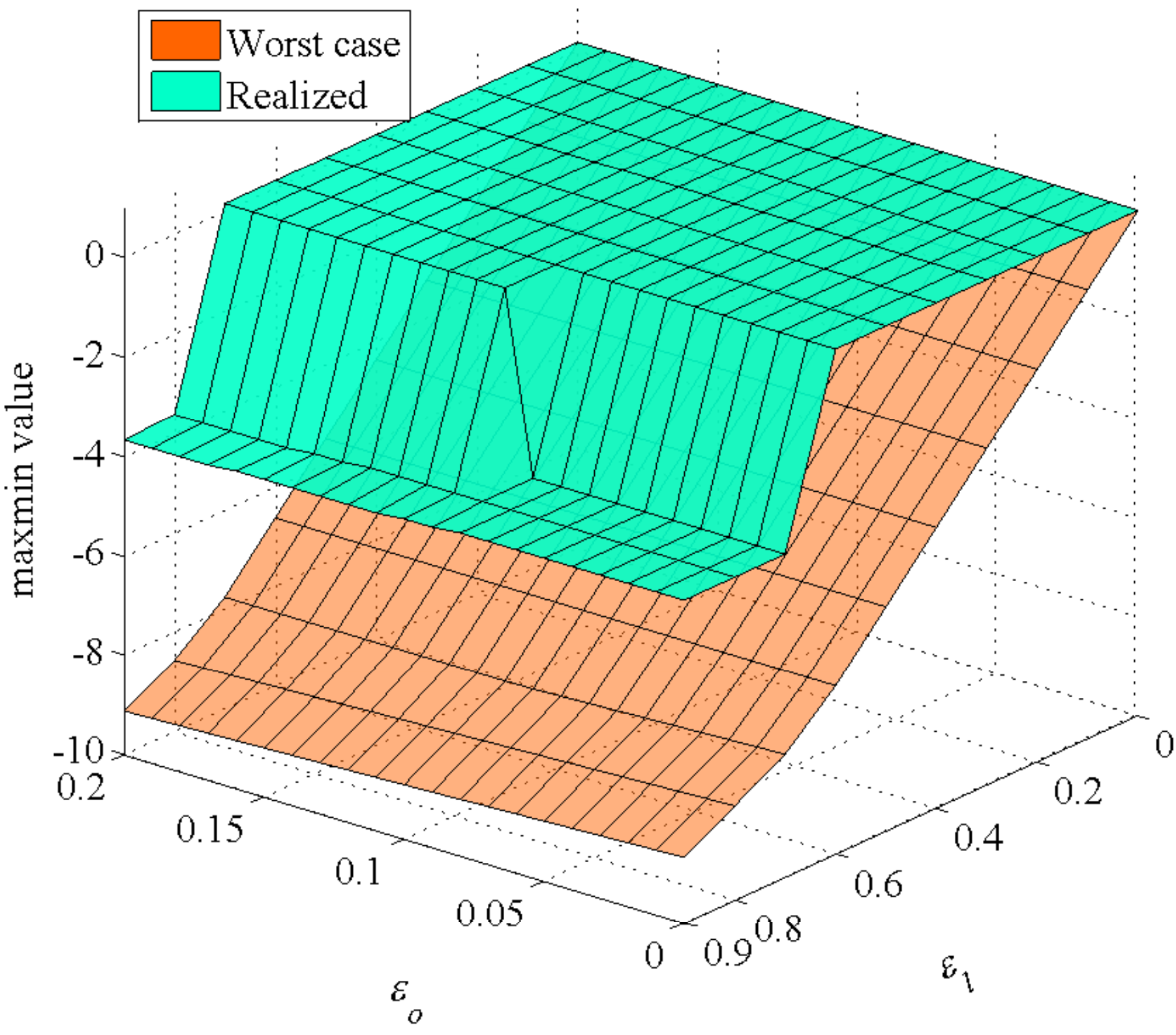}
 \vspace{-3mm}
 \caption{The worst case and realized values of the maxmin value for different values of $\epsilon_{l}$ and $\epsilon_{o}$ for the Karate club dataset with $k_g=k_b=5$}
 \label{fig:uncertainty_3d}
 \vspace{-3mm}
 \end{figure}
 
 \begin{figure}[t]
 \centering
 \includegraphics[width=0.33\textwidth]{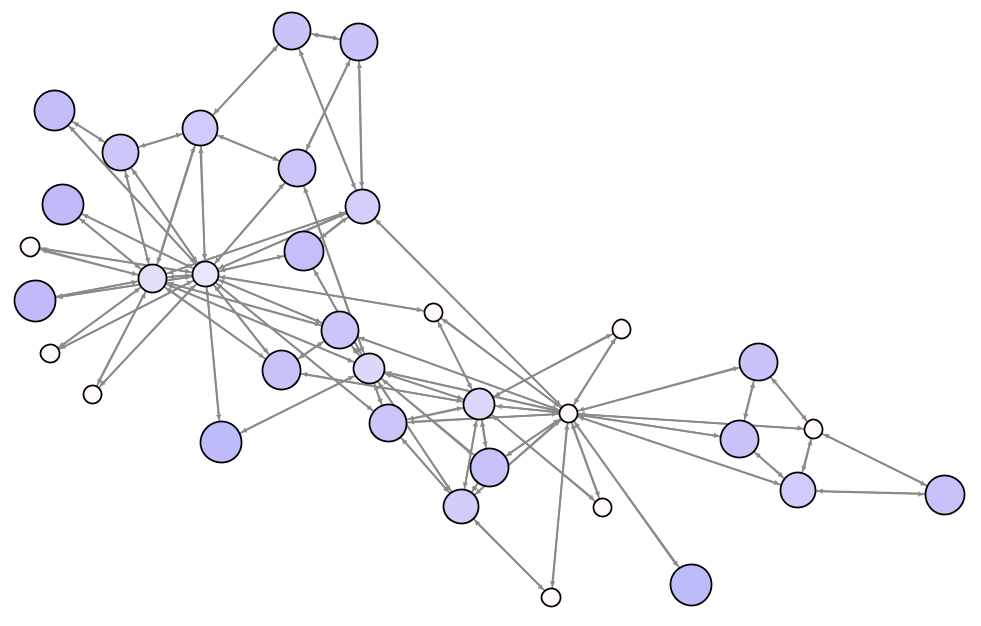}
 \vspace{-2mm}
 \caption{Investment  made by good camp on nodes signified by their sizes and color saturations for the Karate club dataset with $k_g=k_b=5$ (under high uncertainty)}
 \label{fig:uncertainty_graph}
 \vspace{-5mm}
 \end{figure}

As clear from our linear program formulation in Section~\ref{sec:robustness},
the optimal strategy of the good camp under uncertainty is a pessimistic one, that is, ensuring that the good camp performs well even in the worst case.
Hence it is fundamentally different from the optimal strategy under certainty, where there is no requirement of being pessimistic or optimistic
 since all the information required to solve the optimization problem is  known with certainty.
As pointed above, the uncertainty factor plays a role only for very high values of $\epsilon_{l}$ and $\epsilon_{o}$.
Specifically, for values of $\epsilon_l$ upto $0.6$, the good camp's investment is concentrated on one node, like in the certainty setting (Proposition \ref{prop:fundamental_unbounded}).
For $\epsilon_l$ in excess of $0.6$, the distribution of investment showed a very similar nature across different values of $\epsilon_l$ and $\epsilon_o$;
Figure~\ref{fig:uncertainty_graph} illustrates how the good camp distributes its investment over nodes.
Thus,
the investment is distributed under high levels of uncertainty, while it is concentrated on one node under certainty  as well as under low levels of uncertainty.
Furthermore, one could notice a clear correlation between the values of $w_{ig}$ (Figure \ref{fig:karate_details}(a)) and the investment amounts on the corresponding nodes under high uncertainty (Figure~\ref{fig:uncertainty_graph}).

\vspace{-3mm}
\section{Conclusion}

Using a variant of Friedkin-Johnsen model for opinion dynamics, we studied a  zero-sum game framework for optimal investment strategies for two competing camps in a social network. We derived  closed form expressions and efficient algorithms for a number of well-motivated settings.
We showed all the results quantitatively as well as illustratively using simulations on network datasets.
%

Our analysis arrived at a decision parameter analogous to  Katz centrality. We also showed that for some of the popular models of weighing edges,
 a random strategy is indeed optimal. 
We further looked at a setting where the influence of a camp on a node is a concave function of the amount of investment, and derived that a more concave function results in a less skewed investment strategy, which could be perceived as a fairer strategy by the nodes.
%
We studied an adversarial problem where a camp aims to maximize its competitor's investment required to drive the average opinion value of the population in its favor, and saw that the optimal strategies fundamentally remain the same, albeit with different forms of the exact optimal strategies.
We  studied Stackelberg variant of the game, under common coupled constraints stating the bound on combined investment by the two  camps on any node.
We analytically derived the subgame perfect Nash equilibrium strategies of the camps, and 
hence quantified the first-mover advantage.
We derived a linear program for obtaining optimal strategy for a camp to whom the parameters' values are uncertain, while playing against a camp having exact information.
In our simulations, we observed that a camp is likely to get affected only under considerably high levels of uncertainty. Here, the optimal strategy is to have a distributed investment over nodes, as against the concentrated investment in the exact information setting.
%



\vspace{-3mm}
\subsection*{Future Work}

This work has several interesting directions for future work; we mention a few.
The two camps setting can be extended to multiple camps where each camp would attempt to drive the opinion of the population towards its own.
%
With more than two camps, we would need to have the camps hold opinions in a multidimensional plane rather than on the real number line ($+1/\!-\!1$). 
We explain one way in which this could be done.
Let each camp (say $h$)  have a vector associated with its opinion (say $\vec{c}_h$). Let its investment on node $i$ be denoted by $z_{hi}$ and the weightage attributed by  $i$ to the camp's opinion be $w_{ih}$.
Analogous to Equation (\ref{eqn:steadystate}), 
the vector-sum of nodes' opinions could be written as

\vspace{-3mm}
\begin{small}
\begin{align*}
\sum_{i \in N} \vec{v}_i = \sum_{i \in N} r_i  w_{ii}^0 \vec{v}_i^{\,0} + \sum_h \sum_{i \in N} r_i w_{ih} z_{hi} \vec{c}_h
\end{align*}
\end{small}
%
\noindent
 A camp's objective would be
to drive this vector-sum towards the direction of its own opinion vector, that is,
to maximize the inner product between the vector-sum of nodes' opinions and its own opinion vector (that is, $\vec{c}_h \cdot \sum_{i \in N} \vec{v}_i$).
%
%
%
This paper studied the special case where we have two camps: $h=g,b$ with $\vec{c}_g = +1, \vec{c}_b = -1$ and $z_{gi}=x_i, z_{bi}=y_i$.
Since the optimization terms for different camps get decoupled under the settings in Sections 3 and 4, the optimization strategies of all camps would be analogous to those derived in this paper.
However, our set of equations will not be directly adaptable for the CCC and uncertainty settings, and is an interesting direction for future work.


%
%

 One could study other models of opinion dynamics with respect to optimal investment strategies of competing camps and investigate if it is possible to arrive at closed form expressions under them.
 One such type of models could be where the investment parameters get coupled in the optimization problem, and
 the sum of opinion values is no longer a multilinear function in the investment parameters $\{x_i\}_{i\in N},\{y_i\}_{i\in N}$.
 Another possible future direction is to study
the setting of common coupled constraints under more complex constraints, while maintaining  analytical tractability;
our analysis in this paper could act as a conceptual base for solving such a problem.
%
In general, it would be interesting to consider hybrid settings which combine the individually studied settings in this paper, for instance,
 a concave influence function with the cost function accounting for  deviation  from the desired investment.
It would also be interesting to study the tradeoff between the advantages and disadvantages of playing first, since playing first would allow a camp to block certain investments of the competing camp (like in common coupled constraints) but it may also force it to take decision under uncertainty.


%


\bibliographystyle{IEEEtran}
\bibliography{ODSN_references} 

\begin{thebibliography}{10}
\providecommand{\url}[1]{#1}
\csname url@samestyle\endcsname
\providecommand{\newblock}{\relax}
\providecommand{\bibinfo}[2]{#2}
\providecommand{\BIBentrySTDinterwordspacing}{\spaceskip=0pt\relax}
\providecommand{\BIBentryALTinterwordstretchfactor}{4}
\providecommand{\BIBentryALTinterwordspacing}{\spaceskip=\fontdimen2\font plus
\BIBentryALTinterwordstretchfactor\fontdimen3\font minus
  \fontdimen4\font\relax}
\providecommand{\BIBforeignlanguage}[2]{{%
\expandafter\ifx\csname l@#1\endcsname\relax
\typeout{** WARNING: IEEEtran.bst: No hyphenation pattern has been}%
\typeout{** loaded for the language `#1'. Using the pattern for}%
\typeout{** the default language instead.}%
\else
\language=\csname l@#1\endcsname
\fi
#2}}
\providecommand{\BIBdecl}{\relax}
\BIBdecl

\bibitem{gionis2013opinion}
A.~Gionis, E.~Terzi, and P.~Tsaparas, ``Opinion maximization in social
  networks,'' in \emph{Proceedings of the 2013 SIAM International Conference on
  Data Mining}.\hskip 1em plus 0.5em minus 0.4em\relax SIAM, 2013, pp.
  387--395.

\bibitem{grabisch2017strategic}
M.~Grabisch, A.~Mandel, A.~Rusinowska, and E.~Tanimura, ``Strategic influence
  in social networks,'' \emph{Mathematics of Operations Research}, 2017.

\bibitem{easley2010networks}
D.~Easley and J.~Kleinberg, ``Networks, crowds, and markets,'' \emph{Cambridge
  Univ Press}, vol.~6, no.~1, pp. 6--1, 2010.

\bibitem{acemoglu2011opinion}
D.~Acemoglu and A.~Ozdaglar, ``Opinion dynamics and learning in social
  networks,'' \emph{Dynamic Games and Applications}, vol.~1, no.~1, pp. 3--49,
  2011.

\bibitem{charnes1953constrained}
A.~Charnes, ``Constrained games and linear programming,'' \emph{Proceedings of
  the National Academy of Sciences}, vol.~39, no.~7, pp. 639--641, 1953.

\bibitem{dubey2006competing}
P.~Dubey, R.~Garg, and B.~De~Meyer, ``Competing for customers in a social
  network: The quasi-linear case,'' in \emph{WINE}.\hskip 1em plus 0.5em minus
  0.4em\relax Springer, 2006, pp. 162--173.

\bibitem{bimpikis2016competitive}
K.~Bimpikis, A.~Ozdaglar, and E.~Yildiz, ``Competitive targeted advertising
  over networks,'' \emph{Operations Research}, vol.~64, no.~3, pp. 705--720,
  2016.

\bibitem{degroot1974reaching}
M.~H. DeGroot, ``Reaching a consensus,'' \emph{Journal of the American
  Statistical Association}, vol.~69, no. 345, pp. 118--121, 1974.

\bibitem{friedkin1990social}
N.~E. Friedkin and E.~C. Johnsen, ``Social influence and opinions,''
  \emph{Journal of Mathematical Sociology}, vol.~15, no. 3-4, pp. 193--206,
  1990.

\bibitem{friedkin1997social}
------, ``Social positions in influence networks,'' \emph{Social Networks},
  vol.~19, no.~3, pp. 209--222, 1997.

\bibitem{lorenz2007continuous}
J.~Lorenz, ``Continuous opinion dynamics under bounded confidence: A survey,''
  \emph{International Journal of Modern Physics C}, vol.~18, no.~12, pp.
  1819--1838, 2007.

\bibitem{xia2013opinion}
H.~Xia, H.~Wang, and Z.~Xuan, ``Opinion dynamics: A multidisciplinary review
  and perspective on future research,'' in \emph{Multidisciplinary Studies in
  Knowledge and Systems Science}.\hskip 1em plus 0.5em minus 0.4em\relax IGI
  Global, 2013, pp. 311--332.

\bibitem{das2014modeling}
A.~Das, S.~Gollapudi, and K.~Munagala, ``Modeling opinion dynamics in social
  networks,'' in \emph{Proceedings of the 7th ACM International Conference on
  Web Search and Data Mining}.\hskip 1em plus 0.5em minus 0.4em\relax ACM,
  2014, pp. 403--412.

\bibitem{clifford1973model}
P.~Clifford and A.~Sudbury, ``A model for spatial conflict,''
  \emph{Biometrika}, vol.~60, no.~3, pp. 581--588, 1973.

\bibitem{holley1975ergodic}
R.~A. Holley and T.~M. Liggett, ``Ergodic theorems for weakly interacting
  infinite systems and the voter model,'' \emph{The annals of probability}, pp.
  643--663, 1975.

\bibitem{parsegov2017novel}
S.~E. Parsegov, A.~V. Proskurnikov, R.~Tempo, and N.~E. Friedkin, ``Novel
  multidimensional models of opinion dynamics in social networks,'' \emph{IEEE
  Transactions on Automatic Control}, vol.~62, no.~5, pp. 2270--2285, 2017.

\bibitem{ghaderi2013opinion}
J.~Ghaderi and R.~Srikant, ``Opinion dynamics in social networks: A local
  interaction game with stubborn agents,'' in \emph{American Control Conference
  (ACC), 2013}.\hskip 1em plus 0.5em minus 0.4em\relax IEEE, 2013, pp.
  1982--1987.

\bibitem{ben2016robust}
W.~Ben-Ameur, P.~Bianchi, and J.~Jakubowicz, ``Robust distributed consensus
  using total variation,'' \emph{IEEE Transactions on Automatic Control},
  vol.~61, no.~6, pp. 1550--1564, 2016.

\bibitem{jia2015opinion}
P.~Jia, A.~MirTabatabaei, N.~E. Friedkin, and F.~Bullo, ``Opinion dynamics and
  the evolution of social power in influence networks,'' \emph{SIAM review},
  vol.~57, no.~3, pp. 367--397, 2015.

\bibitem{halu2013connect}
A.~Halu, K.~Zhao, A.~Baronchelli, and G.~Bianconi, ``Connect and win: The role
  of social networks in political elections,'' \emph{EPL (Europhysics
  Letters)}, vol. 102, no.~1, p. 16002, 2013.

\bibitem{yildiz2013binary}
E.~Yildiz, A.~Ozdaglar, D.~Acemoglu, A.~Saberi, and A.~Scaglione, ``Binary
  opinion dynamics with stubborn agents,'' \emph{ACM Transactions on Economics
  and Computation}, vol.~1, no.~4, p.~19, 2013.

\bibitem{ballester2006s}
C.~Ballester, A.~Calv{\'o}-Armengol, and Y.~Zenou, ``Who's who in networks.
  wanted: The key player,'' \emph{Econometrica}, vol.~74, no.~5, pp.
  1403--1417, 2006.

\bibitem{sobehy2016elections}
A.~Sobehy, W.~Ben-Ameur, H.~Afifi, and A.~Bradai, ``How to win elections,'' in
  \emph{International Conference on Collaborative Computing: Networking,
  Applications and Worksharing}.\hskip 1em plus 0.5em minus 0.4em\relax
  Springer, 2016.

\bibitem{kempe2003maximizing}
D.~Kempe, J.~Kleinberg, and {\'E}.~Tardos, ``Maximizing the spread of influence
  through a social network,'' in \emph{Proceedings of the ninth ACM SIGKDD
  International Conference on Knowledge Discovery and Data Mining}.\hskip 1em
  plus 0.5em minus 0.4em\relax ACM, 2003, pp. 137--146.

\bibitem{guille2013information}
A.~Guille, H.~Hacid, C.~Favre, and D.~A. Zighed, ``Information diffusion in
  online social networks: A survey,'' \emph{ACM SIGMOD Record}, vol.~42, no.~1,
  pp. 17--28, 2013.

\bibitem{bharathi2007competitive}
S.~Bharathi, D.~Kempe, and M.~Salek, ``Competitive influence maximization in
  social networks,'' in \emph{International Workshop on Web and Internet
  Economics}.\hskip 1em plus 0.5em minus 0.4em\relax Springer, 2007, pp.
  306--311.

\bibitem{goyal2014competitive}
S.~Goyal, H.~Heidari, and M.~Kearns, ``Competitive contagion in networks,''
  \emph{Games and Economic Behavior}, 2014.

\bibitem{etesami2016complexity}
S.~R. Etesami and T.~Ba{\c{s}}ar, ``Complexity of equilibrium in competitive
  diffusion games on social networks,'' \emph{Automatica}, vol.~68, pp.
  100--110, 2016.

\bibitem{myers2012clash}
S.~A. Myers and J.~Leskovec, ``Clash of the contagions: Cooperation and
  competition in information diffusion,'' in \emph{Proceedings of the Twelfth
  International Conference on Data Mining (ICDM)}.\hskip 1em plus 0.5em minus
  0.4em\relax IEEE, 2012, pp. 539--548.

\bibitem{su2016understanding}
Y.~Su, X.~Zhang, L.~Liu, S.~Song, and B.~Fang, ``Understanding information
  interactions in diffusion: an evolutionary game-theoretic perspective,''
  \emph{Frontiers of Computer Science}, vol.~10, no.~3, pp. 518--531, 2016.

\bibitem{rosen1965existence}
J.~B. Rosen, ``Existence and uniqueness of equilibrium points for concave
  n-person games,'' \emph{Econometrica: Journal of the Econometric Society},
  pp. 520--534, 1965.

\bibitem{altman2009constrained}
E.~Altman and E.~Solan, ``Constrained games: The impact of the attitude to
  adversary's constraints,'' \emph{IEEE Transactions on Automatic Control},
  vol.~54, no.~10, pp. 2435--2440, 2009.

\bibitem{ben2009robust}
A.~Ben-Tal, L.~El~Ghaoui, and A.~Nemirovski, \emph{Robust optimization}.\hskip
  1em plus 0.5em minus 0.4em\relax Princeton University Press, 2009.

\bibitem{altafini2013consensus}
C.~Altafini, ``Consensus problems on networks with antagonistic interactions,''
  \emph{IEEE Transactions on Automatic Control}, vol.~58, no.~4, pp. 935--946,
  2013.

\bibitem{proskurnikov2016opinion}
A.~V. Proskurnikov, A.~S. Matveev, and M.~Cao, ``Opinion dynamics in social
  networks with hostile camps: Consensus vs. polarization,'' \emph{IEEE
  Transactions on Automatic Control}, vol.~61, no.~6, pp. 1524--1536, 2016.

\bibitem{hubbard2015vector}
J.~H. Hubbard and B.~B. Hubbard, \emph{Vector calculus, linear algebra, and
  differential forms: a unified approach}.\hskip 1em plus 0.5em minus
  0.4em\relax Matrix Editions, 2015.

\bibitem{katz1953new}
L.~Katz, ``A new status index derived from sociometric analysis,''
  \emph{Psychometrika}, vol.~18, no.~1, pp. 39--43, 1953.

\bibitem{jain1984quantitative}
R.~Jain, D.-M. Chiu, and W.~R. Hawe, \emph{A quantitative measure of fairness
  and discrimination for resource allocation in shared computer system}, 1984,
  vol.~38.

\bibitem{dhamal2018resource}
S.~Dhamal, W.~Ben-Ameur, T.~Chahed, and E.~Altman, ``Resource allocation
  polytope games: Uniqueness of equilibrium, price of stability, and price of
  anarchy,'' in \emph{Proceedings of the 32nd AAAI Conference on Artificial
  Intelligence}.\hskip 1em plus 0.5em minus 0.4em\relax AAAI, 2018, pp.
  997--1006.

\bibitem{chen2009efficient}
W.~Chen, Y.~Wang, and S.~Yang, ``Efficient influence maximization in social
  networks,'' in \emph{Proceedings of the 15th ACM SIGKDD International
  Conference on Knowledge Discovery and Data Mining}.\hskip 1em plus 0.5em
  minus 0.4em\relax ACM, 2009, pp. 199--208.

\bibitem{chen2010scalable}
W.~Chen, C.~Wang, and Y.~Wang, ``Scalable influence maximization for prevalent
  viral marketing in large-scale social networks,'' in \emph{Proceedings of the
  16th ACM SIGKDD International Conference on Knowledge Discovery and Data
  Mining}.\hskip 1em plus 0.5em minus 0.4em\relax ACM, 2010, pp. 1029--1038.

\bibitem{zachary1977information}
W.~W. Zachary, ``An information flow model for conflict and fission in small
  groups,'' \emph{Journal of anthropological research}, vol.~33, no.~4, pp.
  452--473, 1977.

\end{thebibliography}

\newpage
\appendices


\vspace{-2mm}
\section{Proof of Proposition \ref{prop:concave_unbounded}}
\label{app:concave_unbounded}
\begin{proof}
We make  the following one-to-one transformation: $y_i=\mathcal{Y}_i^t$ and $x_i=\mathcal{X}_i^t$.
So the above optimization problem is equivalent to

\vspace{-3mm}
\begin{small}
\begin{align*}
\max_{\substack{\sum_i \mathcal{X}_i^t \leq k_g \\ \mathcal{X}_i \geq 0}} 
 \min_{\substack{\sum_i \mathcal{Y}_i^t \leq k_b \\ \mathcal{Y}_i \geq 0}} \sum_{{i \in N}} v_i
 \\ \text{s.t. }
 \forall {i \in N}: \;
 v_i 
 - \sum_{j \in N} w_{ij}v_j 
 + w_{ib} \mathcal{Y}_i
 = w_{ig} \mathcal{X}_i
 + w_{ii}^0 v_i^0
 & 
\end{align*}
\end{small}
\vspace{-3mm}

\noindent
Similar to the derivation of Equation (\ref{eqn:steadystate}) in  main text, we get

\vspace{-3mm}
\begin{small}
\begin{align}
\sum_{i \in N} v_i = \sum_{i \in N} r_i ( w_{ig}\mathcal{X}_i + w_{ii}^0 v_i^0) - \sum_{i \in N} r_i w_{ib}\mathcal{Y}_i
\end{align}
\vspace{-2mm}
\end{small}

Owing to the mutual independence between $\bar{\mathcal{X}}$ and $\bar{\mathcal{Y}}$ (where $\bar{\mathcal{X}}$ and $\bar{\mathcal{Y}}$ are the vectors with components $\mathcal{X}_i$ and $\mathcal{Y}_i$ respectively), we can write

\vspace{-3mm}
\begin{small}
\begin{align*}
&\max_{\substack{\sum_i \mathcal{X}_i^t \leq k_g \\ \mathcal{X}_i \geq 0}} 
 \min_{\substack{\sum_i \mathcal{Y}_i^t \leq k_b \\ \mathcal{Y}_i \geq 0}}  \sum_{{i \in N}} v_i
 \\&= \sum_{i \in N} r_i  w_{ii}^0 v_i^0
 + \max_{\substack{\sum_i \mathcal{X}_i^t \leq k_g \\ \mathcal{X}_i \geq 0}} \sum_{i \in N} r_i  w_{ig}\mathcal{X}_i 
 - \max_{\substack{\sum_{i} \mathcal{Y}_i^t \leq k_b \\ \mathcal{Y}_i \geq 0}} \sum_{i \in N} r_i  w_{ib}\mathcal{Y}_i 
\end{align*}
\vspace{-1mm}
\end{small}

We now solve for the first optimization term with respect to $\bar{\mathcal{X}}$, which can be written as the following convex optimization problem

~\vspace{-3mm}
\begin{small}
\begin{align*}
 \min \;\;\;  -\sum_{{i \in N}} r_i w_{ig} \mathcal{X}_i 
 &
\\ \text{s.t. }
\sum_{i \in N} \mathcal{X}_i^t - k_g \leq 0
& \;\;\; \leftarrow \gamma
\\
\hspace{-3mm}
\forall {i \in N}: \;
-\mathcal{X}_i \leq 0
& \;\;\; \leftarrow \beta_i
\end{align*}
\vspace{-4mm}
\end{small}

The modified objective function with the Lagrangian multipliers is

\vspace{-3mm}
\begin{small}
\begin{align*}
g = -\sum_{{i \in N}} r_i w_{ig} \mathcal{X}_i + \gamma \left( \sum_{i \in N} \mathcal{X}_i^t - k_g \right) -\sum_{i \in N} \beta_i \mathcal{X}_i
\end{align*}
\vspace{-1mm}
\end{small}

The KKT conditions with the modified objective function $g$ give the following additional constraints

\vspace{-2mm}
\begin{small}
\begin{align}
\forall {i \in N}: \;
\frac{\partial g}{\partial \mathcal{X}_i} &= -r_i w_{ig} + t \gamma \mathcal{X}_i^{t-1} -\beta_i = 0
\label{eqn:KKTx}
\\
\gamma &\geq 0
\label{eqn:KKTgamma}
\\
\forall {i \in N}: \; \beta_i &\geq 0
\label{eqn:KKTbeta}
\\
\gamma \left(\sum_i \mathcal{X}_i^t - k_g \right) &= 0
\label{eqn:KKTgammaslack}
\\
\forall {i \in N}: \; \beta_i \mathcal{X}_i &= 0
\label{eqn:KKTbetaslack}
\end{align}
\end{small}
\vspace{-3mm}

\noindent
Constraint (\ref{eqn:KKTx}) gives

\vspace{-3mm}
\begin{small}
\begin{align}
\mathcal{X}_i = \left( \frac{\beta_i + r_i w_{ig}}{t \gamma} \right) ^\frac{1}{t-1}
\label{eqn:xnbeta}
\end{align}
\vspace{-2mm}
\end{small}

Note that if $\gamma=0$, Constraint (\ref{eqn:KKTx}) gives $\beta_i + r_i w_{ig}=0, \forall {i \in N}$. Since $\beta_i \geq 0$, this can hold only when $r_i w_{ig} \leq 0, \forall {i \in N}$. For nodes with $r_i w_{ig} = 0$, the value of $r_i w_{ig} \mathcal{X}_i$ stays 0 for any $\mathcal{X}_i$ and so $\mathcal{X}_i=0$ is an optimal solution. For nodes with $r_i w_{ig} < 0$, we must have $\beta_i>0$ and hence $\mathcal{X}_i=0$ (from (\ref{eqn:KKTbetaslack})).

Henceforth we assume $\gamma>0$ and so Equation (\ref{eqn:xnbeta}) is valid.
This, with (\ref{eqn:KKTgammaslack}) and (\ref{eqn:KKTgamma}), gives

\vspace{-3mm}
\begin{small}
\begin{align*}
\sum_{i \in N} \mathcal{X}_i^t = k_g
\end{align*}
\vspace{-2mm}
\end{small}

\noindent
From (\ref{eqn:xnbeta}), we get

\vspace{-3mm}
\begin{small}
\begin{gather*}
\sum_{i \in N} \left( \frac{\beta_i + r_i w_{ig}}{t \gamma} \right)^\frac{t}{t-1} = k_g
\\ \Longleftrightarrow 
\left( \frac{1}{t \gamma} \right)^\frac{1}{t-1}  = \frac{k_g^{\frac{1}{t}}}{\left( \sum_{i \in N} \left( {\beta_i + r_i w_{ig}} \right)^\frac{t}{t-1} \right)^{\frac{1}{t}}}
\end{gather*}
\vspace{-2mm}
\end{small}

\noindent
This, with (\ref{eqn:xnbeta}) gives

\vspace{-3mm}
\begin{small}
\begin{align}
\mathcal{X}_i^* = k_g^{\frac{1}{t}} \left( \frac{(\beta_i + r_i w_{ig})^\frac{1}{t-1}}{\left( \sum_{j \in N} \left( {\beta_j + r_j w_{jg}} \right)^\frac{t}{t-1} \right)^{\frac{1}{t}}} \right)
\label{eqn:xopt}
\end{align}
\vspace{-2mm}
\end{small}

Now, the following cases are possible depending on the sign of $r_i w_{ig}$.

\begin{description}
\item[Case 1:] ~ 
If $r_i w_{ig} < 0$, we must have $\beta_i > 0$ (from the constraint $\mathcal{X}_i^* \geq 0$).
This with (\ref{eqn:KKTbetaslack}) gives
$\mathcal{X}_i^* = 0$.

\item[Case 2:] ~ 
If $r_i w_{ig} = 0$, we must have $\mathcal{X}_i^* = \beta_i = 0$ (from (\ref{eqn:KKTbetaslack}) and (\ref{eqn:xnbeta})).
\end{description}
Furthermore in both Case 1 and Case 2, from (\ref{eqn:xnbeta}), we have $\beta_i + r_i w_{ig} = 0$.
So the $i$-terms for which $r_i w_{ig} \leq 0$ vanish from the denominator of (\ref{eqn:xopt}) and it transforms into

\vspace{-3mm}
\begin{small}
\begin{align}
\mathcal{X}_i^* = k_g^{\frac{1}{t}} \left( \frac{(\beta_i + r_i w_{ig})^\frac{1}{t-1}}{\left( \sum_{j:r_j w_{jg}>0} \left( {\beta_j + r_j w_{jg}} \right)^\frac{t}{t-1} \right)^{\frac{1}{t}}} \right)
, \text{if } r_i w_{ig} \geq 0
\label{eqn:xoptnew}
\\
\text{and }
\mathcal{X}_i^* = 0
, \text{if } r_i w_{ig} < 0
\nonumber
\end{align}
\vspace{-3mm}
\end{small}

\begin{description}
\item[Case 3:] ~ 
If $r_i w_{ig} > 0$, we  get $\mathcal{X}_i^* > 0$ (from (\ref{eqn:KKTbeta})). This with (\ref{eqn:KKTbetaslack}) gives
$\beta_i = 0$, and hence from (\ref{eqn:xoptnew}), we get 

\vspace{-3mm}
\begin{small}
\begin{align}
\mathcal{X}_i^* = k_g^{\frac{1}{t}} \left( \frac{(r_i w_{ig})^\frac{1}{t-1}}{\left( \sum_{j:r_j w_{jg}>0} \left( {r_j w_{jg}} \right)^\frac{t}{t-1} \right)^{\frac{1}{t}}} \right)
\nonumber
\end{align}
\vspace{-3mm}
\end{small}

\end{description}

The above cases can be concisely written as

\vspace{-3mm}
\begin{small}
\begin{align}
\mathcal{X}_i^* &=  k_g^{\frac{1}{t}} \left( \frac{(r_i w_{ig})^\frac{1}{t-1}}{\left( \sum_{j:r_j w_{jg}>0} \left( {r_j w_{jg}} \right)^\frac{t}{t-1} \right)^{\frac{1}{t}}} \right) \cdot \mathbb{I}_{r_i w_{ig} > 0}
\nonumber
\\[.5em]
\Longleftrightarrow
x_i^* &=  k_g \left( \frac{(r_i w_{ig})^{\frac{t}{t-1}}}{ \sum_{j:r_j w_{jg}>0} (r_j w_{jg})^{\frac{t}{t-1}}  } \right) \cdot \mathbb{I}_{r_i w_{ig} > 0}
\label{eqn:xoptfinal}
\end{align}
\vspace{-1mm}
\end{small}

where $\mathbb{I}_{r_i w_{ig} > 0} = 1$ if $r_i w_{ig} > 0$, and $0$ otherwise.

Similarly, it can be shown that

\vspace{-3mm}
\begin{small}
\begin{align}
y_i^* &=  k_b \left( \frac{(r_i w_{ib})^{\frac{t}{t-1}}}{ \sum_{j:r_j w_{jb}>0} (r_j w_{jb})^{\frac{t}{t-1}}  } \right) \cdot \mathbb{I}_{r_i w_{ib} > 0}
\label{eqn:yoptfinal}
\end{align}
\vspace{-1mm}
\end{small}

From the optimization problem, it can be directly seen that, if $\forall {i \in N}, r_i w_{ig} \leq 0$, then $x_i^*=0, \forall {i \in N}$. Similarly, if $\forall {i \in N}, r_i w_{ib} \leq 0$, then $y_i^*=0, \forall {i \in N}$.
\end{proof}

\vspace{-2mm}
\section{Proof of Proposition \ref{prop:concave_bounded}}
\label{app:concave_bounded}

\begin{proof}
This proof goes as an extension to the proof of Proposition \ref{prop:concave_unbounded}, with the additional constraint

\vspace{-3mm}
\begin{small}
\begin{align}
\forall {i \in N}: \; \mathcal{X}_i \leq 1 \;\;\; \leftarrow \xi_i
\end{align}
\vspace{-3mm}
\end{small}

\noindent
Equation (\ref{eqn:KKTx}) changes to

\vspace{-3mm}
\begin{small}
\begin{align}
\forall {i \in N}: \;
 -r_i w_{ig} + t \gamma \mathcal{X}_i^{t-1} -\beta_i+\xi_i = 0
\label{eqn:KKTxnew}
\end{align}
\vspace{-3mm}
\end{small}

\noindent
and Equation (\ref{eqn:xnbeta}) changes to

\vspace{-3mm}
\begin{small}
\begin{align}
\mathcal{X}_i = \left( \frac{\beta_i -\xi_i+ r_i w_{ig}}{t \gamma} \right) ^\frac{1}{t-1}
\end{align}
\vspace{-3mm}
\end{small}

\noindent
The complementary slackness conditions for the additional constraints are

\vspace{-3mm}
\begin{small}
\begin{align}
\forall {i \in N}: \; \xi_i (\mathcal{X}_i - 1) = 0
\label{eqn:KKTxislack}
\end{align}
\vspace{-3mm}
\end{small}

Further note that the constraints $\mathcal{X}_i \geq 0$ and $\mathcal{X}_i \leq 1$ cannot both be tight, and so at least one of $\beta_i$ or $\xi_i$ must be 0. That is,

\vspace{-3mm}
\begin{small}
\begin{align}
\forall {i \in N}: \; \beta_i \xi_i = 0
\label{eqn:betaxi_conflict}
\end{align}
\vspace{-3mm}
\end{small}

When $\gamma=0$, 
Equation (\ref{eqn:KKTxnew}) transforms into $\beta_i-\xi_i+r_iw_{ig}=0, \forall {i \in N}$. 
Now if $r_i w_{ig}=0$, we must have $\beta_i-\xi_i=0$ and hence $\beta_i=\xi_i=0$ (from Equation (\ref{eqn:betaxi_conflict})).
It is further clear from the objective function that if $r_i w_{ig}=0$, it is optimal to have $\mathcal{X}_i=0$ and hence $x_i^* = 0$.
If $r_i w_{ig}<0$, we must have $\beta_i>0$ (since $\beta_i-\xi_i+r_iw_{ig}=0$) and so $\mathcal{X}_i=x_i^*=0$ (from Equation (\ref{eqn:KKTbetaslack})).
If $r_i w_{ig}>0$, we must have $\xi_i>0$ and so $\mathcal{X}_i=x_i^*=1$ (from Equation (\ref{eqn:KKTxislack})).
That is, when $\gamma=0$, 
we invest an amount of 1 on all nodes $i$ with positive values of $r_i w_{ig}$ and 0 on all other nodes.

\vspace{-2mm}
\begin{small}
\begin{align}
\gamma = 0 \implies x_i^* = 1\cdot\mathbb{I}_{r_i w_{ig} > 0}
\label{eqn:investon_allpositives}
\end{align}
\vspace{-2mm}
\end{small}

For $\gamma>0$, we have

\vspace{-2mm}
\begin{small}
\begin{align*}
&\;\sum_{i \in N} \mathcal{X}_i^t = k_g
\\
\Longleftrightarrow &\;
\sum_{i \in N} \left( \frac{\beta_i -\xi_i+ r_i w_{ig}}{t \gamma} \right) ^\frac{t}{t-1} = k_g
\end{align*}
\vspace{-1mm}
\end{small}

When $\beta_i > 0$, we have $\xi_i=0$ (from (\ref{eqn:betaxi_conflict})) and $\mathcal{X}_i=0$ (from (\ref{eqn:KKTbetaslack})). From the above equation, this corresponds to $r_i w_{ig} < 0$. Further note that for $r_i w_{ig}=0$, $\mathcal{X}_i=0$ is an optimal solution. So we have

\vspace{-2mm}
\begin{small}
\begin{align*}
\sum_{i:r_i w_{ig}>0} \left( \frac{-\xi_i+ r_i w_{ig}}{t \gamma} \right) ^\frac{t}{t-1} = k_g
\end{align*}
\vspace{-1mm}
\end{small}

We can have $\mathcal{X}_i=1$ or if not, we should have $\xi_i = 0$. 
So for any $i$ such that $r_i w_{ig} > 0$, we have $\mathcal{X}_i = \min \left\{ \left( \frac{r_i w_{ig}}{t \gamma} \right) ^\frac{1}{t-1} , 1 \right\}$.
Let $J_\gamma = \{i: 0<{r_i w_{ig}} \leq t \gamma\}$.
So the above equation results in

\vspace{-2mm}
\begin{small}
\begin{align}
\label{eqn:equals_kg_app}
\sum_{{i\in J_\gamma}} \left( \frac{r_i w_{ig}}{t \gamma} \right) ^\frac{t}{t-1} + \sum_{\substack{i\notin J_\gamma\\ i:r_i w_{ig}>0}} 1 = k_g
\\
\Longleftrightarrow
\left( \frac{1}{t\gamma} \right) ^\frac{t}{t-1} = \frac{k_g - \sum_{i\notin J_\gamma,i:r_i w_{ig}>0} 1}{\sum_{i\in J_\gamma} (r_i w_{ig})^\frac{t}{t-1}}
\nonumber
\end{align}
\vspace{-1mm}
\end{small}


Note that if $i \in J_\gamma$, then any $j$ for which $r_j w_{jg} < r_i w_{ig}$ belongs to $J_\gamma$.
So $J_\gamma$ and hence $\gamma$ can be determined by adding nodes to $J_\gamma$, one at a time in increasing order of $r_i w_{ig}$, subject to $r_i w_{ig}>0$.
Let $\hat{\gamma}$ be the value of $\gamma$ so obtained.
%
%
It can be seen that as $\gamma$ decreases, the left hand side of (\ref{eqn:equals_kg_app}) increases. Since the right hand side is a constant, we would obtain a unique $\hat{\gamma}$ satisfying the equality. Furthermore, for $\gamma>0$, this budget constraint is tight and so we are ensured the existence of $\hat{\gamma}$.
Once $\hat{\gamma}$ and hence $J_{\hat{\gamma}}$ are obtained, an optimal solution for the good camp can be expressed as follows (recall that $x_i=\mathcal{X}_i^t$):

If $r_i w_{ig} \leq 0$, $x_i^*=0$.

If $r_i w_{ig} > t\hat{\gamma}$, $x_i^*=1$.

If $0 < r_i w_{ig} \leq t\hat{\gamma}$, 

\vspace{-2mm}
\begin{small}
\begin{align*}
x_i^* = \Bigg( k_g - \sum_{\substack{i\notin J_{\hat{\gamma}}\\i:r_i w_{ig}>0}} 1 \Bigg) \left( \frac{(r_i w_{ig})^\frac{t}{t-1}}{\sum_{{i\in J_{\hat{\gamma}}}} (r_i w_{ig})^\frac{t}{t-1}} \right)  
\end{align*}
\vspace{-1mm}
\end{small}

However, if Equation (\ref{eqn:equals_kg_app}) is not satisfied for any $\gamma$, that is, when 

\vspace{-3mm}
\begin{small}
\begin{align*}
\sum_{{i\in J_\gamma}} \left( \frac{r_i w_{ig}}{t \gamma} \right) ^\frac{t}{t-1} + \sum_{\substack{i\notin J_\gamma\\ i:r_i w_{ig}>0}} 1 < k_g
\end{align*}
\vspace{-1mm}
\end{small}

\noindent
even for the lowest value of $\gamma>0$, we have $J_{\gamma}=\{\}$ and hence $\sum_{i:r_i w_{ig}>0} 1 < k_g$.
Here, the number of nodes with $r_i w_{ig}>0$ is less than $k_g$, meaning that the budget constraint is not tight and so $\gamma=0$ (from Equation (\ref{eqn:KKTgammaslack})). The investment is thus as per Equation (\ref{eqn:investon_allpositives}).
\end{proof}

\begin{table*} 
\caption{Results for Karate club dataset for different values of $w_{ii}^0+w_{ig}+w_{ib}$}
\label{tab:karate_W0WgWb}
\centering
\begin{tabular}{|c|c|c|c|c|c|c|c|c|}
\hline 
\multicolumn{2}{|c|}{{Setting}} &  \multirow{3}{*}{Section} & \multirow{3}{*}{$k_g$} & \multirow{3}{*}{$k_b$} & \multirow{3}{*}{Value of} & \multicolumn{3}{c|}{\T \B $w_{ii}^0+w_{ig}+w_{ib}$} 
\\ \cline{1-2}\cline{7-9} \T \B
Aspect & Case & &  & & & $0.1$  & $0.5$ & $0.9$
\\ \hline \T \B
\multirow{4}{*}{Fundamental} & \multirow{3}{*}{Unbounded} & \multirow{3}{*}{3.1.1} & 5 & 5 & $\max_{\mathbf{x}} \min_{\mathbf{y}} \sum_{i\in N} v_i$ & $0.1141$ & $0.1564$ 
 & $0.1723$
\\ \cline{4-9} \T \B
&  & & 10 & 5 & $\max_{\mathbf{x}} \min_{\mathbf{y}} \sum_{i\in N} v_i$ & $7.6575$ & $5.8811$ 
 & $3.5673$
\\ \cline{2-9} \T \B
& Bounded & 3.1.2 & 5 & 5 & $\max_{\mathbf{x}} \min_{\mathbf{y}} \sum_{i\in N} v_i$ & $-0.0678$ & $-0.0538$ 
 & $0.3309$
\\ \hline \T \B
\multirow{3}{*}{Adversary} &
Unbounded & 5.1 & 5 & - & $\max_{\mathbf{x}} \min_{\mathbf{y}} \sum_{i\in N} y_i$ & $5.0768$ & $5.1404$ 
 & $5.2674$
\\ \cline{2-9} \T \B
& Bounded & 5.2 & 5 & - & $\max_{\mathbf{x}} \min_{\mathbf{y}} \sum_{i\in N} y_i$ & $4.8542$ & $4.8936$ 
 & $5.6844$
\\ \hline \T \B
\multirow{6}{*}{Concave ($t=2$)} &
\multirow{3}{*}{Unbounded} &
\multirow{3}{*}{4.1} & \multirow{1}{*}{5} & \multirow{1}{*}{5} &  $\max_{\mathbf{x}} \min_{\mathbf{y}} \sum_{i\in N} v_i$ & $0.3545$ & $0.4581$ 
 &  $0.5435$
\\ \cline{4-9} \T \B
&  &  & 20 & 20 &  $\max_{\mathbf{x}} \min_{\mathbf{y}} \sum_{i\in N} v_i$ & $0.7089$ & $0.9163$ 
&  $1.0870$
\\ \cline{2-9} \T \B
& \multirow{3}{*}{Bounded} & \multirow{3}{*}{4.2} & \multirow{1}{*}{5} & \multirow{1}{*}{5} & $\max_{\mathbf{x}} \min_{\mathbf{y}} \sum_{i\in N} v_i$ & $0.3931$ & $0.4612$ 
& $0.5435$
\\ \cline{4-9} \T \B
&  &  & 20 & 20 &  $\max_{\mathbf{x}} \min_{\mathbf{y}} \sum_{i\in N} v_i$ & $1.4950$ & $1.2653$ 
 &  $1.0714$
\\ \hline \T \B
\multirow{6}{*}{Concave ($t=10$)} &
\multirow{3}{*}{Unbounded} &
\multirow{3}{*}{4.1} & \multirow{1}{*}{5} & \multirow{1}{*}{5} &  $\max_{\mathbf{x}} \min_{\mathbf{y}} \sum_{i\in N} v_i$ & $1.2512$ & $1.1180$ 
 & $1.0120$
\\ \cline{4-9} \T \B
&  &  & 20 & 20 &  $\max_{\mathbf{x}} \min_{\mathbf{y}} \sum_{i\in N} v_i$ & $1.4373$ & $1.2842$  
 &  $1.1625$
\\ \cline{2-9} \T \B
& \multirow{3}{*}{Bounded} & \multirow{3}{*}{4.2} & \multirow{1}{*}{5} & \multirow{1}{*}{5} & $\max_{\mathbf{x}} \min_{\mathbf{y}} \sum_{i\in N} v_i$ & $1.2512$ & $1.1180 $ 
 & $1.0120$
\\ \cline{4-9} \T \B
&  &  & 20 & 20 &  $\max_{\mathbf{x}} \min_{\mathbf{y}} \sum_{i\in N} v_i$ & $1.5027$ & $1.3104$ 
 &  $1.1633$
 \\ \hline \T \B
 \multirow{3}{*}{CCC} & \multirow{3}{*}{Bounded} &
\multirow{3}{*}{6} & \multirow{3}{*}{5} & \multirow{3}{*}{5} & $\max_{\mathbf{x}} \min_{\mathbf{y}} \sum_{i\in N} v_i$ & $2.8105$ & $1.5399$ 
& $0.3951$
\\  \T \B
 &  &  & & & $\min_{\mathbf{y}} \max_{\mathbf{x}} \sum_{i\in N} v_i$ & $-1.8655$ & $-0.8900$ 
 & $0.2581$
\\ \hline
\end{tabular}
\end{table*}

\section{Proof of Proposition \ref{prop:investment_unbounded}}
\label{app:investment_unbounded}

\begin{proof}

\begin{small}
\begin{gather*}
\vspace{-3mm}
\max_{\substack{\sum_i x_i \leq k_g \\ x_i \geq 0}} 
 \min_{y_i \geq 0} \sum_{{i \in N}} y_i
\\
\text{s.t. }
\sum_{i \in N} v_i \leq 0
\\
\hspace{-3mm}
\sum_{i \in N}
v_i 
=
 \sum_{i \in N} r_i ( w_{ig}x_i + w_{ii}^0 v_i^0) - \sum_{i \in N} r_i w_{ib}y_i
\end{gather*}
\vspace{-2mm}
\end{small}

The inner term of this optimization problem is

\vspace{-3mm}
\begin{small}
\begin{align*}
 \min \sum_{{i \in N}} y_i
 &
\\ \text{s.t. }
  \sum_{i \in N} r_i w_{ib}y_i \geq \sum_{i \in N} r_i ( w_{ig}x_i + w_{ii}^0 v_i^0)
  & \;\;\; \leftarrow \pi
\\
\forall {i \in N}: \;
y_i \geq 0
\end{align*}
\vspace{-3mm}
\end{small}

It is clear that if $\sum_{i \in N} r_i ( w_{ig}x_i + w_{ii}^0 v_i^0) \leq 0$, it is optimal for the bad camp to have $y_i = 0, \forall {i \in N}$, that is, $\sum_{i \in N} y_i = 0$.
So we need to only consider the case where  $\sum_{i \in N} r_i ( w_{ig}x_i + w_{ii}^0 v_i^0) > 0$.
In this case, for the constraint $\sum_{i \in N} v_i \leq 0$ to be satisfied, it is necessary that there exists a node $j$ with positive value of $r_j w_{jb}$.

The dual problem of the above optimization problem can be written as

\vspace{-2mm}
\begin{small}
\begin{align}
\nonumber
\max \pi \sum_{i \in N} r_i ( w_{ig} x_i + w_{ii}^0 v_i^0 ) 
&
\\ \text{s.t. }
\forall {i \in N}: \;
\pi r_i w_{ib} 
\leq 1
& \;\;\; \leftarrow y_i
\nonumber
\\
\pi \geq 0
\nonumber
\end{align}
\vspace{-3mm}
\end{small}

From the dual constraints, we have
$
\pi \leq  \frac{1}{\max_{j \in N} r_j w_{jb} } 
$
(recall that $\max_{j \in N} r_j w_{jb}>0$ in this considered case).
Since we aim to maximize the dual objective function, we have

\vspace{-1mm}
\begin{small}
\begin{align*}
\pi =  \frac{1}{\max_{j \in N} w_{jb} r_j} 
\end{align*}
\end{small}
\vspace{-1mm}

\noindent
The optimization problem now becomes

\vspace{-3mm}
\begin{small}
\begin{align*}
 &
 \max_{\substack{\sum x_i \leq k_g \\ x_i \geq 0}} 
  \pi  \sum_{i \in N} r_i ( w_{ig} x_i + w_{ii}^0 v_i^0 ) 
  \\
   &=
   \pi \left( \max_{\substack{\sum x_i \leq k_g \\ x_i \geq 0}} \sum_{i \in N} r_i w_{ig} x_i +  \sum_{i \in N} r_i w_{ii}^0 v_i^0 \right)
   \\
   &= \frac{1}{\max_{j \in N} w_{jb} r_j} \left( k_g \max \left\{ \max_{i \in N} r_i w_{ig} , 0 \right\} +  \sum_{i \in N} r_i w_{ii}^0 v_i^0\right) 
\end{align*}
\end{small}
\vspace{-1mm}

\noindent
The above expression thus gives the total investment made by the bad camp in the case where $(k_g \max\{ \max_{i \in N} r_i w_{ig} , 0 \} +  \sum_{i \in N} r_i w_{ii}^0 v_i^0) > 0$. 

Accounting for the case that it optimal for the bad camp to have $y_i=0, \forall {i \in N}$ when $(k_g \max\{ \max_{i \in N} r_i w_{ig} , 0 \} +  \sum_{i \in N} r_i w_{ii}^0 v_i^0) \leq 0$, the total investment made by the bad camp is

\vspace{-3mm}
\begin{small}
\begin{align*}
\max \left\{ \frac{1}{\max_{j \in N} w_{jb} r_j}  \left( k_g \max\{ \max_{i \in N} r_i w_{ig} , 0 \} +  \sum_i r_i w_{ii}^0 v_i^0\right), 0 \right\}
\end{align*}
\end{small}
\end{proof}

\section{Further Simulation Results}
\label{app:sim}

The following pages present simulation results not included in the main text.
\begin{itemize}
\item
The quantitative results for Karate club and NetHEPT datasets for different values of $w_{ii}^0+w_{ig}+w_{ib}$ are presented in Tables \ref{tab:karate_W0WgWb} and \ref{tab:hep_W0WgWb}, respectively.
\item
The illustrative results for Karate club dataset under additional constraints such as bounded investment per node and common coupled constraints, when $w_{ii}^0+w_{ig}+w_{ib} = 0.1$ and $0.9$ are presented in Figures \ref{fig:CCC_0.1} and \ref{fig:CCC_0.9}, respectively.
\item
The illustrative results for Karate club dataset under concave influence function when $w_{ii}^0+w_{ig}+w_{ib} = 0.1$ and $0.9$ are presented in Figures \ref{fig:concave_0.1} and \ref{fig:concave_0.9}, respectively.
\item
The illustrations for progression of opinion values in Karate club dataset under different settings are presented in Figures \ref{fig:karate_prog_concave_2}-\ref{fig:karate_prog_minmax}.
\end{itemize}

\begin{table*} 
\caption{Results for NetHEPT dataset for different values of $w_{ii}^0+w_{ig}+w_{ib}$}
\label{tab:hep_W0WgWb}
\centering
\begin{tabular}{|c|c|c|c|c|c|c|c|c|}
\hline 
\multicolumn{2}{|c|}{{Setting}} &  \multirow{3}{*}{Section} & \multirow{3}{*}{$k_g$} & \multirow{3}{*}{$k_b$} & \multirow{3}{*}{Value of} & \multicolumn{3}{c|}{\T \B $w_{ii}^0+w_{ig}+w_{ib}$} 
\\ \cline{1-2}\cline{7-9} \T \B
Aspect & Case & &  & & & $0.1$  & $0.5$ & $0.9$
\\ \hline \T \B
\multirow{4}{*}{Fundamental} & \multirow{3}{*}{Unbounded} & \multirow{3}{*}{3.1.1} & 100 & 100 & $\max_{\mathbf{x}} \min_{\mathbf{y}} \sum_{i\in N} v_i$ & $145.9562$ & $73.2539$ 
 & $13.5988$
\\ \cline{4-9} \T \B
&  & & 200 & 100 & $\max_{\mathbf{x}} \min_{\mathbf{y}} \sum_{i\in N} v_i$ & $626.5203$ & $347.0770$ 
 & $131.4902$
\\ \cline{2-9} \T \B
& Bounded & 3.1.2 & 100 & 100 & $\max_{\mathbf{x}} \min_{\mathbf{y}} \sum_{i\in N} v_i$ & $7.6865$ & $2.8513$ 
 & $1.5930$
\\ \hline \T \B
\multirow{3}{*}{Adversary} &
Unbounded & 5.1 & 100 & - & $\max_{\mathbf{x}} \min_{\mathbf{y}} \sum_{i\in N} y_i$ & $143.6201$ & $136.5231$ 
 & $113.0391$
\\ \cline{2-9} \T \B
& Bounded & 5.2 & 100 & - & $\max_{\mathbf{x}} \min_{\mathbf{y}} \sum_{i\in N} y_i$ & $105.1987$ & $102.7266$ 
 & $101.9037$
\\ \hline \T \B
\multirow{6}{*}{Concave ($t=2$)} &
\multirow{3}{*}{Unbounded} &
\multirow{3}{*}{4.1} & \multirow{1}{*}{100} & \multirow{1}{*}{100} &  $\max_{\mathbf{x}} \min_{\mathbf{y}} \sum_{i\in N} v_i$ & $2.8513$ & $-0.8446$ 
 &  $-1.7849$
\\ \cline{4-9} \T \B
&  &  & 400 & 400 &  $\max_{\mathbf{x}} \min_{\mathbf{y}} \sum_{i\in N} v_i$ & $5.7026$ & $-1.6892$ 
&  $-3.5698$
\\ \cline{2-9} \T \B
& \multirow{3}{*}{Bounded} & \multirow{3}{*}{4.2} & \multirow{1}{*}{100} & \multirow{1}{*}{100} & $\max_{\mathbf{x}} \min_{\mathbf{y}} \sum_{i\in N} v_i$ & $2.8513$ & $-0.8446$ 
& $-1.7849$
\\ \cline{4-9} \T \B
&  &  & 400 & 400 &  $\max_{\mathbf{x}} \min_{\mathbf{y}} \sum_{i\in N} v_i$ & $2.3493$ & $-1.7117$ 
 &  $-3.5698$
\\ \hline \T \B
\multirow{6}{*}{Concave ($t=10$)} &
\multirow{3}{*}{Unbounded} &
\multirow{3}{*}{4.1} & \multirow{1}{*}{100} & \multirow{1}{*}{100} &  $\max_{\mathbf{x}} \min_{\mathbf{y}} \sum_{i\in N} v_i$ & $-7.3610$ & $-12.4212$ 
 & $-13.9045$
\\ \cline{4-9} \T \B
&  &  & 400 & 400 &  $\max_{\mathbf{x}} \min_{\mathbf{y}} \sum_{i\in N} v_i$ & $-8.4556$ & $-14.2682$  
 &  $-15.9721$
\\ \cline{2-9} \T \B
& \multirow{3}{*}{Bounded} & \multirow{3}{*}{4.2} & \multirow{1}{*}{100} & \multirow{1}{*}{100} & $\max_{\mathbf{x}} \min_{\mathbf{y}} \sum_{i\in N} v_i$ & $-7.3610$ & $-12.4212$ 
 & $-13.9045$
\\ \cline{4-9} \T \B
&  &  & 400 & 400 &  $\max_{\mathbf{x}} \min_{\mathbf{y}} \sum_{i\in N} v_i$ & $-8.4556$ & $-14.2682$ 
 &  $-15.9721$
 \\ \hline \T \B
 \multirow{3}{*}{CCC} & \multirow{3}{*}{Bounded} &
\multirow{3}{*}{6} & \multirow{3}{*}{100} & \multirow{3}{*}{100} & $\max_{\mathbf{x}} \min_{\mathbf{y}} \sum_{i\in N} v_i$ & $28.5077$ & $7.4843$ 
& $1.5930$
\\  \T \B
 &  &  & & & $\min_{\mathbf{y}} \max_{\mathbf{x}} \sum_{i\in N} v_i$ & $-19.2724$ & $-3.6795$ 
 & $1.5930$
\\ \hline
\end{tabular}
\end{table*}

\begin{figure*}
\hspace{-7mm}
\begin{tabular}{|c||c||c|}
\hline &~ & \vspace{-3mm} \\
\includegraphics[width=0.33\textwidth]{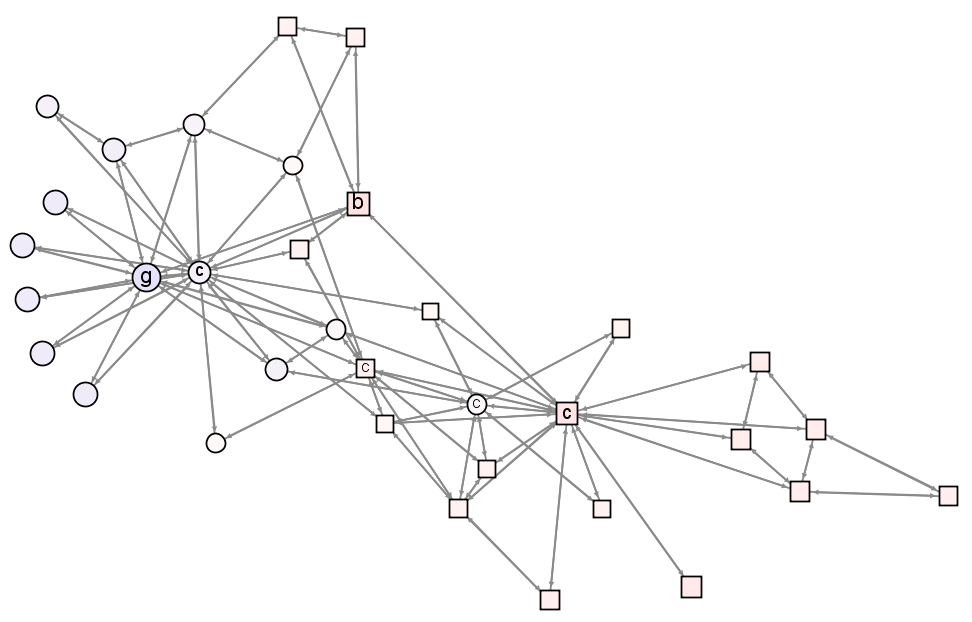}
&
\includegraphics[width=0.33\textwidth]{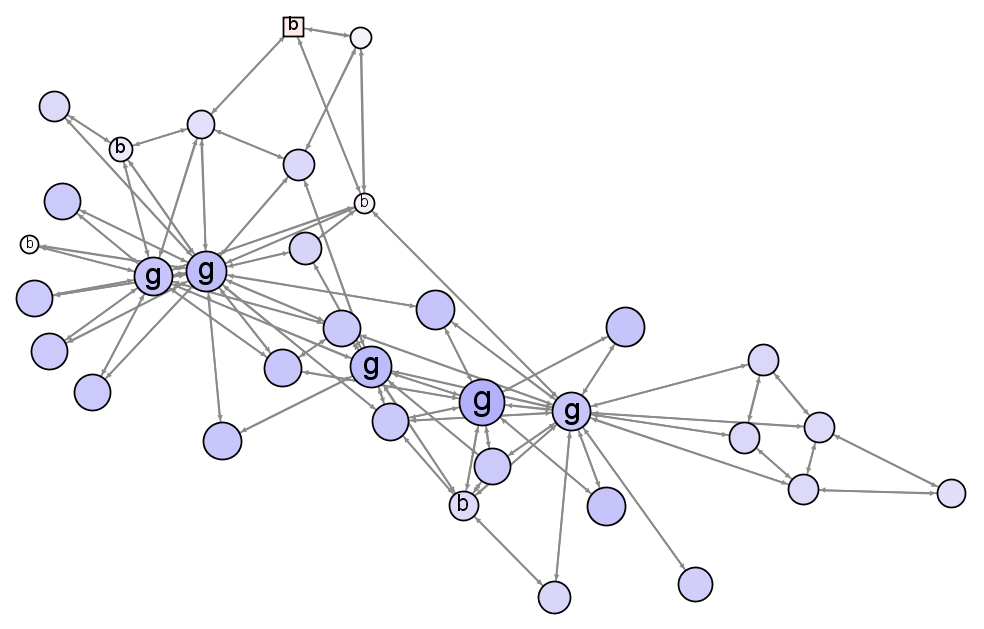}
&
\includegraphics[width=0.33\textwidth]{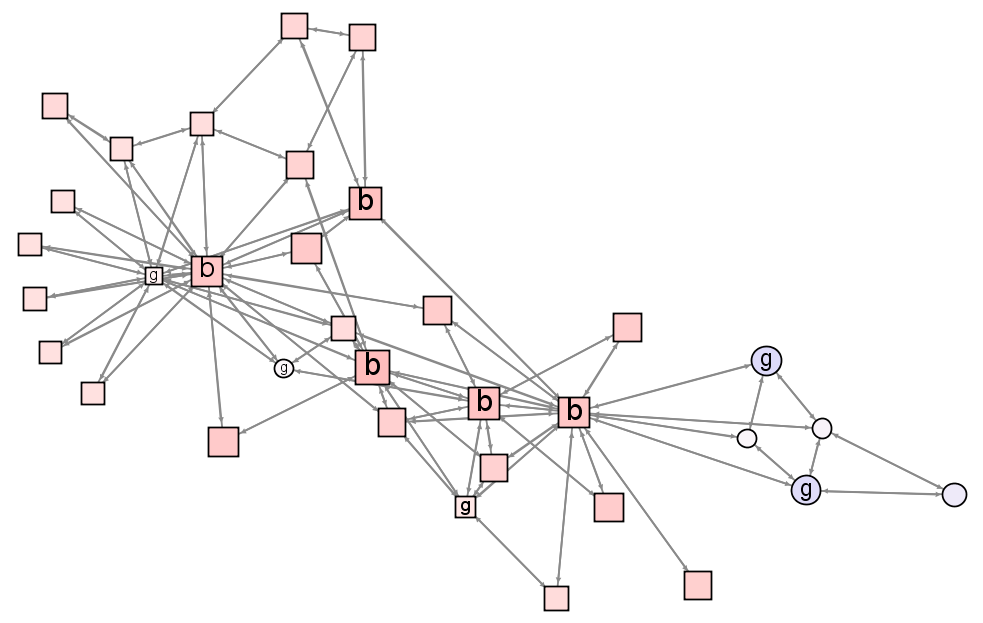}
\\ & & \vspace{-3mm} \\
(a) Bound on investment per node 
&
(b) Common coupled constraints 
&
(c) Common coupled constraints 
\\
 $x_i,y_i\leq 1, \forall i$ 
 &
$x_i+y_i \leq 1, \forall i$ 
&
$x_i+y_i \leq 1, \forall i$ 
\\ 
($\max_{\mathbf{x}} \min_{\mathbf{y}} \sum_i v_i = -0.0678$)
&
($\max_{\mathbf{x}} \min_{\mathbf{y}} \sum_i v_i = 2.8105$)
&
($\min_{\mathbf{y}} \max_{\mathbf{x}} \sum_i v_i = -1.8655$)
\\
\hline 
\end{tabular}
\caption{Results in presence of additional constraints  for the Karate club dataset with $k_g=k_b=5$ ; The nodes are labeled `g/b/c' to signify if invested on by good/bad/both camps respectively. The sign of the opinion value of a node is signified by its shape and color (circle and blue for good, square and red for bad), while the absolute value of its opinion is signified by its size and color saturation. ($w_{ii}^0+w_{ig}+w_{ib}=0.1$)}
\label{fig:CCC_0.1}
\end{figure*}

\begin{figure*}
\hspace{-7mm}
\begin{tabular}{|c||c||c|}
\hline &~ & \vspace{-3mm} \\
\includegraphics[width=0.33\textwidth]{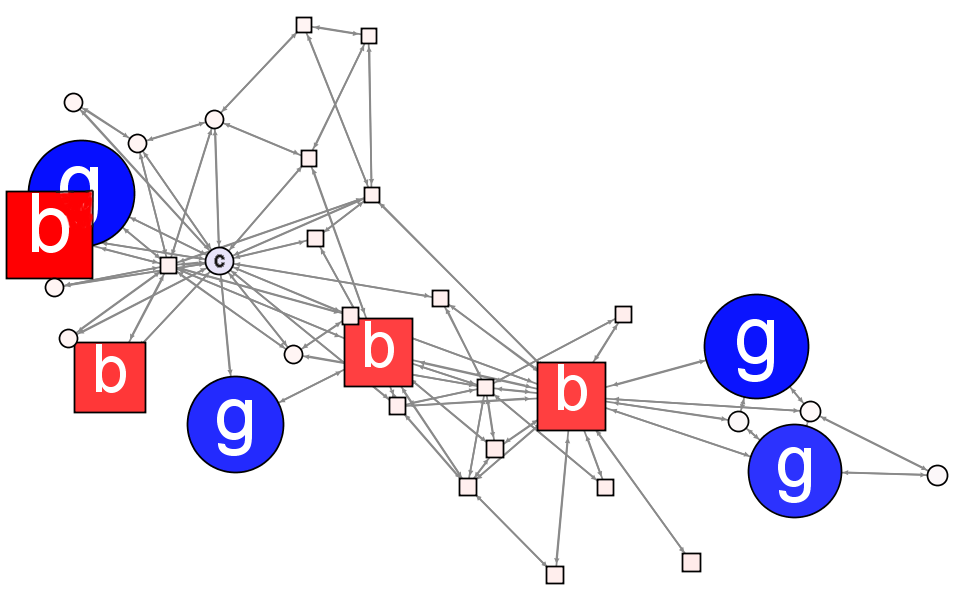}
&
\includegraphics[width=0.33\textwidth]{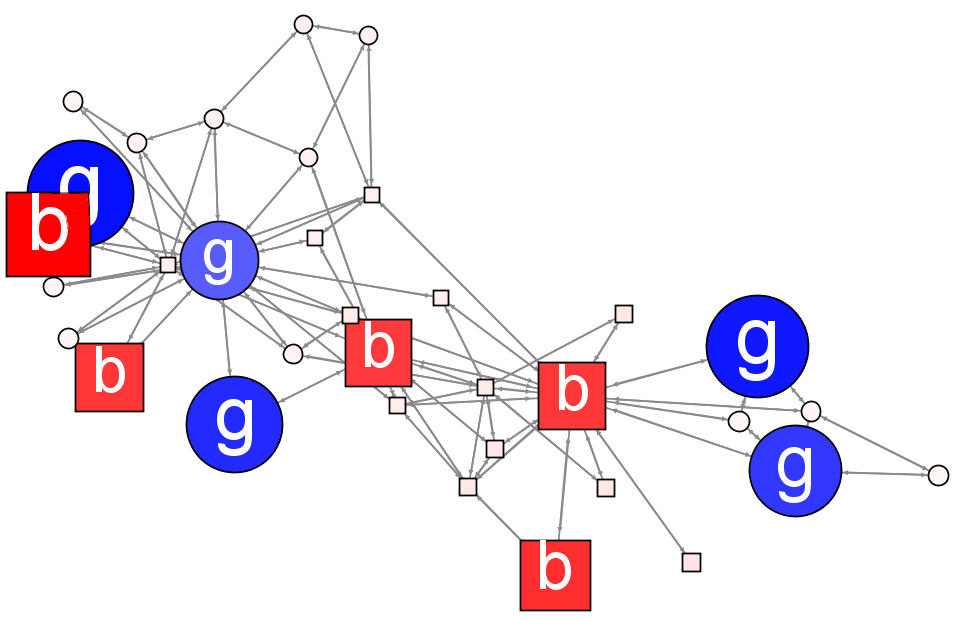}
&
\includegraphics[width=0.33\textwidth]{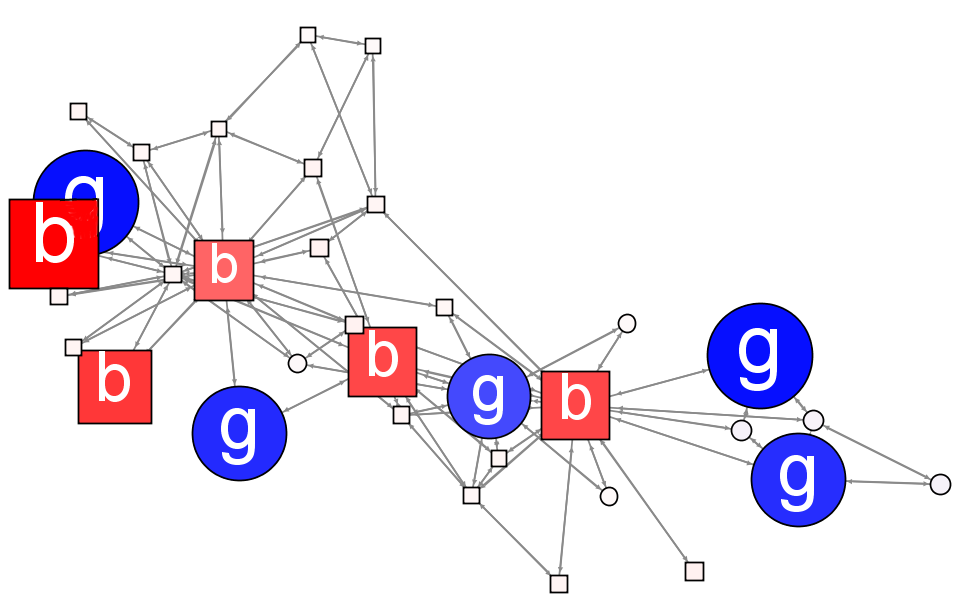}
\\ & & \vspace{-3mm} \\
(a) Bound on investment per node 
&
(b) Common coupled constraints 
&
(c) Common coupled constraints 
\\
 $x_i,y_i\leq 1, \forall i$ 
 &
$x_i+y_i \leq 1, \forall i$ 
&
$x_i+y_i \leq 1, \forall i$ 
\\ 
($\max_{\mathbf{x}} \min_{\mathbf{y}} \sum_i v_i = 0.3309$)
&
($\max_{\mathbf{x}} \min_{\mathbf{y}} \sum_i v_i = 0.3951$)
&
($\min_{\mathbf{y}} \max_{\mathbf{x}} \sum_i v_i = 0.2581$)
\\
\hline 
\end{tabular}
\caption{Results in presence of additional constraints  for the Karate club dataset with $k_g=k_b=5$ ; The nodes are labeled `g/b/c' to signify if invested on by good/bad/both camps respectively. The sign of the opinion value of a node is signified by its shape and color (circle and blue for good, square and red for bad), while the absolute value of its opinion is signified by its size and color saturation. ($w_{ii}^0+w_{ig}+w_{ib}=0.9$)}
\label{fig:CCC_0.9}
\end{figure*}

\begin{figure*}
\hspace{-7mm}
\begin{tabular}{|c||c||c|}
\hline &~ &~ \vspace{-4mm} \\
\includegraphics[width=0.32\textwidth]{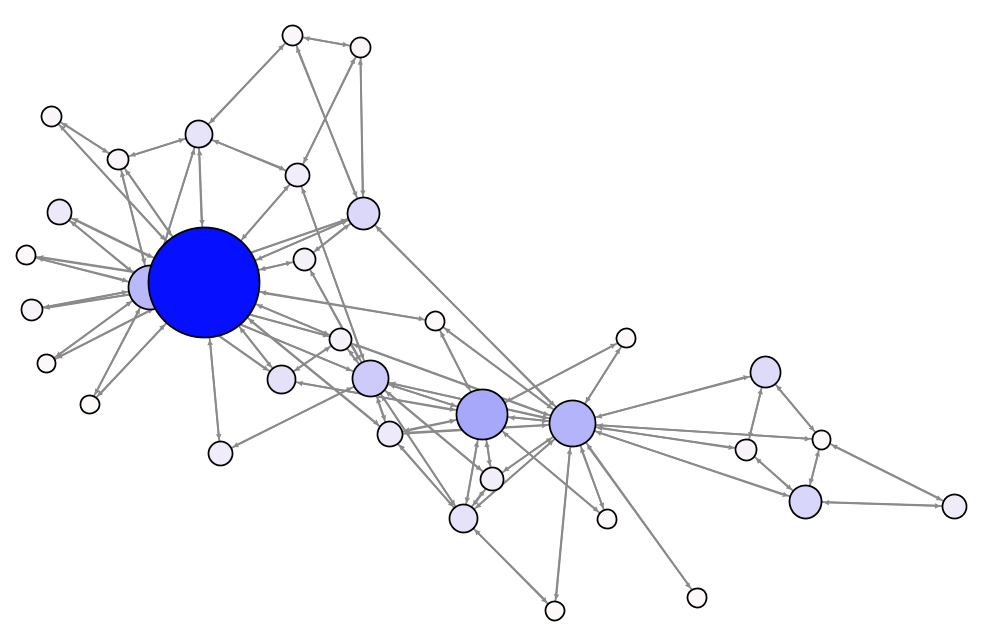} 
&
\includegraphics[width=0.32\textwidth]{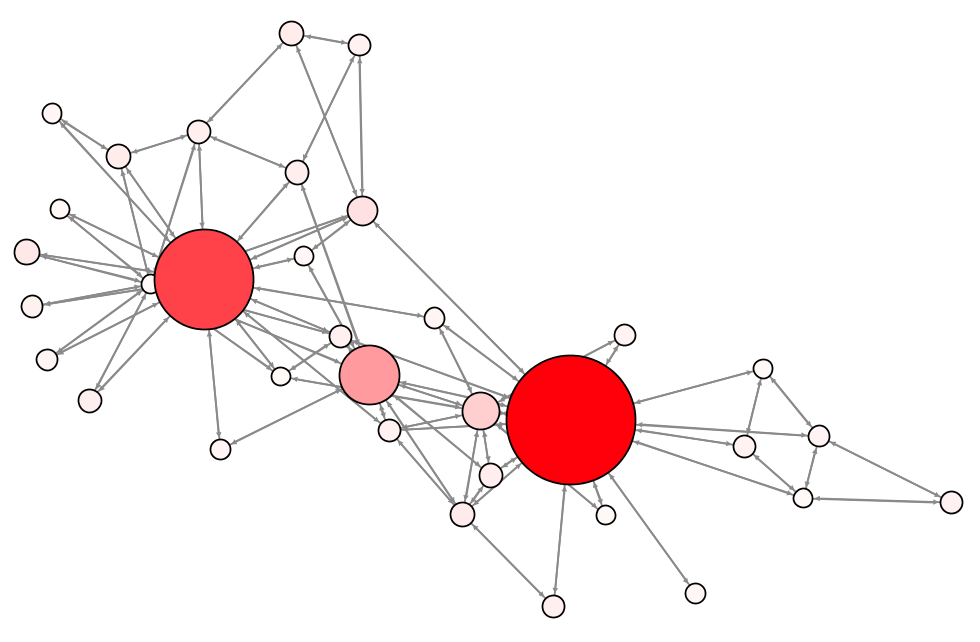}
&
\includegraphics[width=0.32\textwidth]{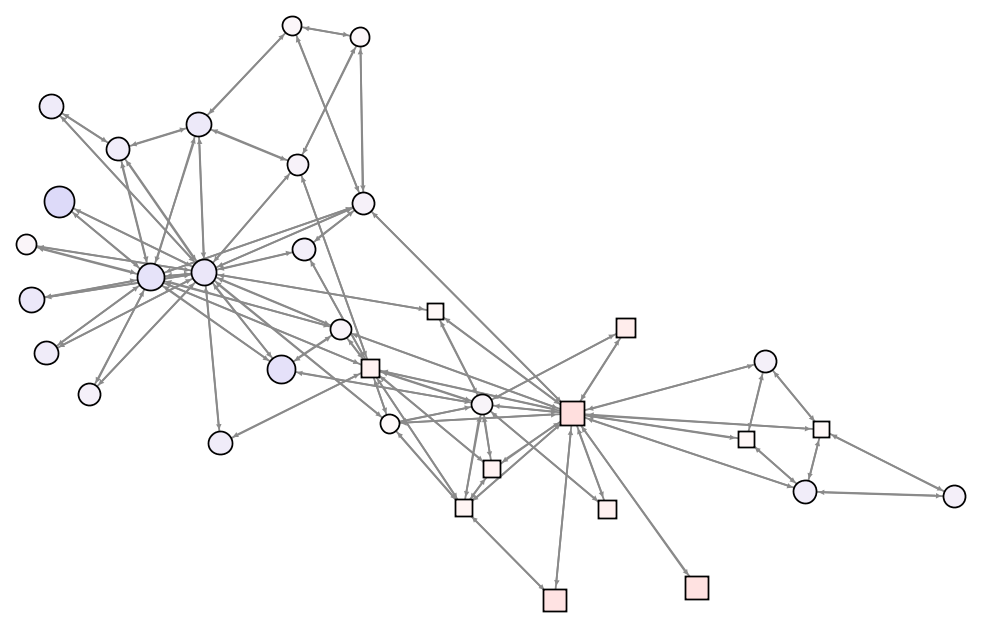} 
\\
&
&
(c) Final opinion values of nodes signified 
\\
(a) Investment made by  good camp 
&
(b) Investment made by  bad camp 
&
$ \! \! \! $ by  shapes, sizes,  color saturations for $t=2$ $ \! \! \! $
\\ 
on nodes signified by their sizes
&
on nodes signified by their sizes 
&
  (Circular blue nodes: positive opinions, 
\\
 and color saturations for $t=2$
 &
  and color saturations for $t=2$
  &
    square red nodes: negative opinions)
  \\
\hline
\hline &~ & \vspace{-4mm} \\
\includegraphics[width=0.32\textwidth]{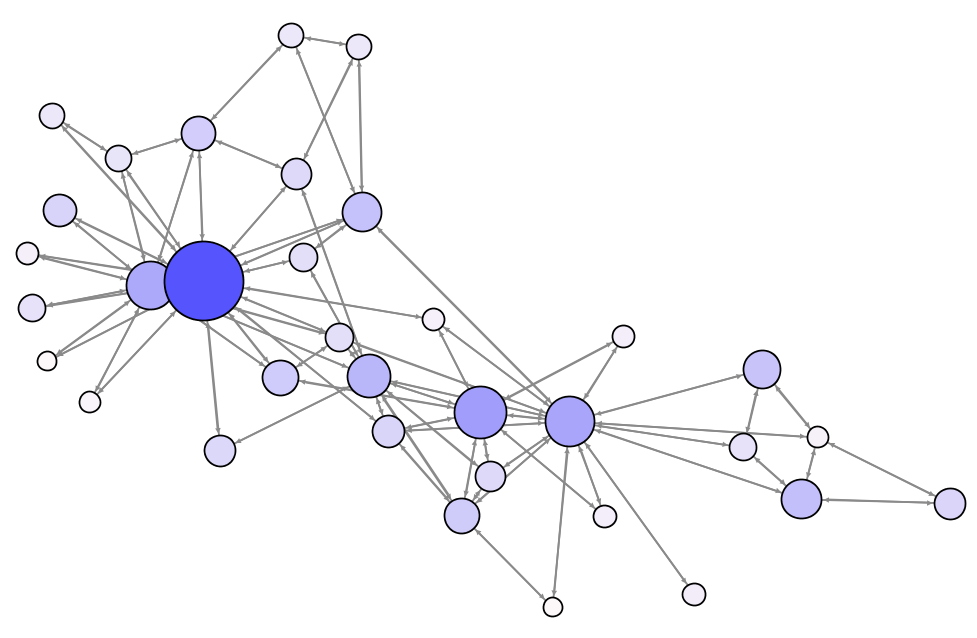} 
&
\includegraphics[width=0.32\textwidth]{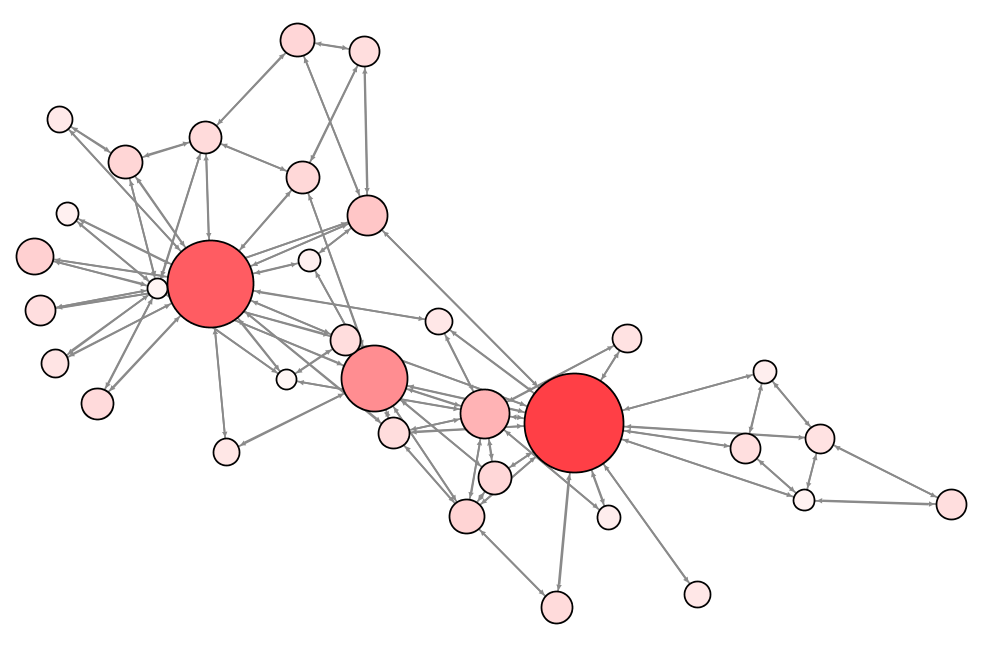}
&
\includegraphics[width=0.32\textwidth]{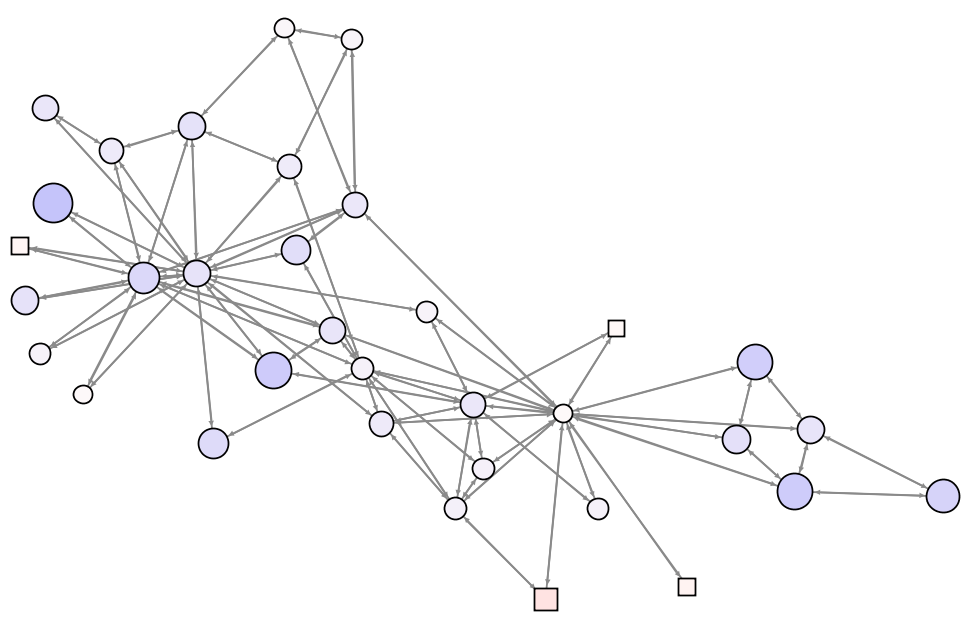}
\\
&&
(f) Final opinion values of nodes signified
\\
(d) Investment made by  good camp 
&
(e) Investment made by  bad camp 
&
$ \! \! \! $ by shapes, sizes, color saturations for $t=10$ $ \! \! \! $
\\ 
on nodes signified by their sizes 
&
on nodes signified by their sizes 
&
  (Circular blue nodes: positive opinions, 
\\
and color saturations for $t=10$
&
        and color saturations for $t=10$
        &
            square red nodes: negative opinions)
        \\
\hline
\end{tabular}
\caption{Simulation results for the Karate club dataset with $k_g=k_b=5$ when the influence function is concave 
($w_{ii}^0+w_{ig}+w_{ib}=0.1$)
}
\label{fig:concave_0.1}
\vspace{-1mm}
\end{figure*}

\begin{figure*}
\hspace{-7mm}
\begin{tabular}{|c||c||c|}
\hline &~ &~ \vspace{-4mm} \\
\includegraphics[width=0.32\textwidth]{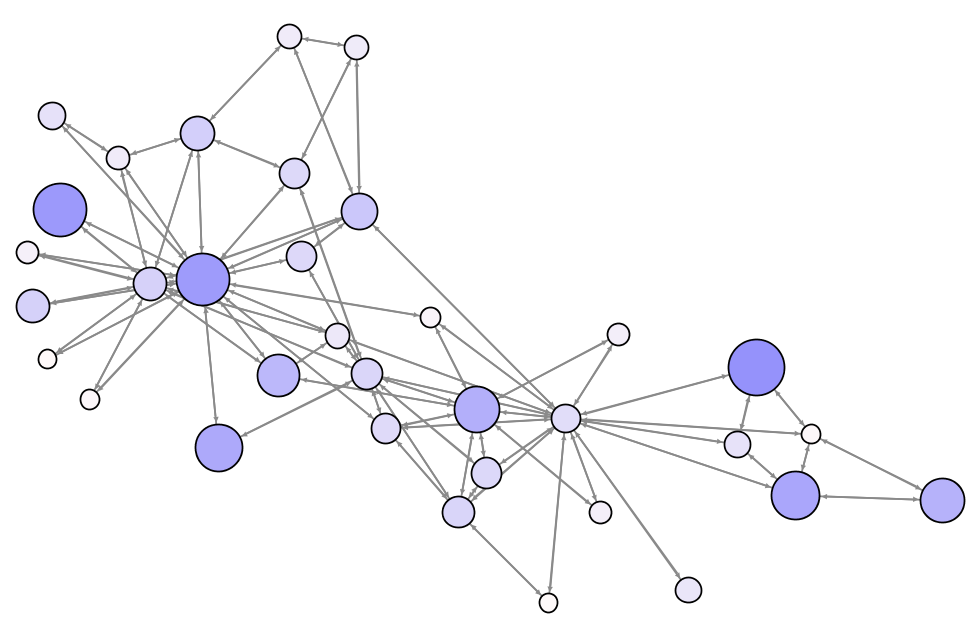} 
&
\includegraphics[width=0.32\textwidth]{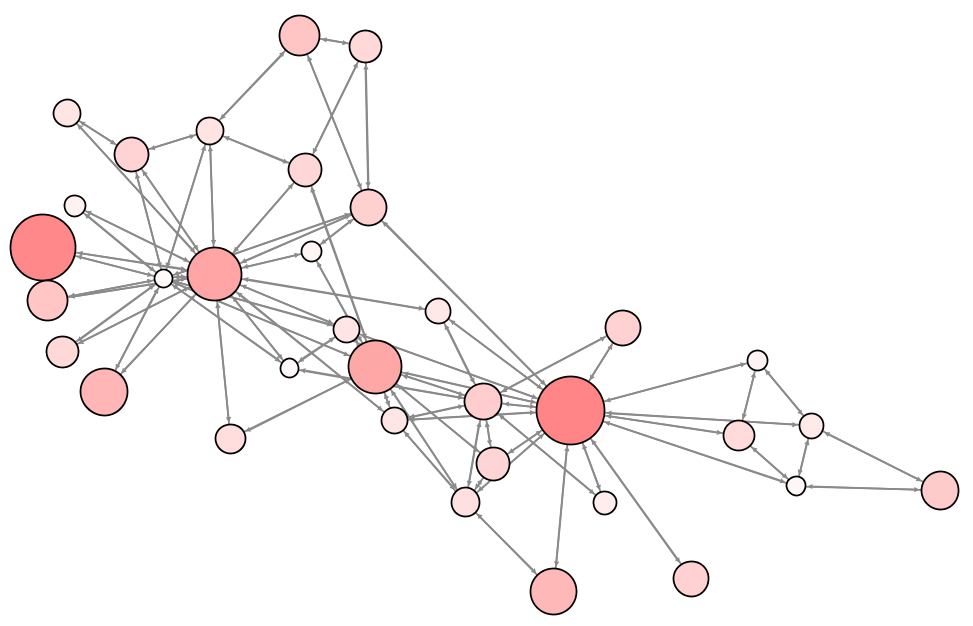}
&
\includegraphics[width=0.32\textwidth]{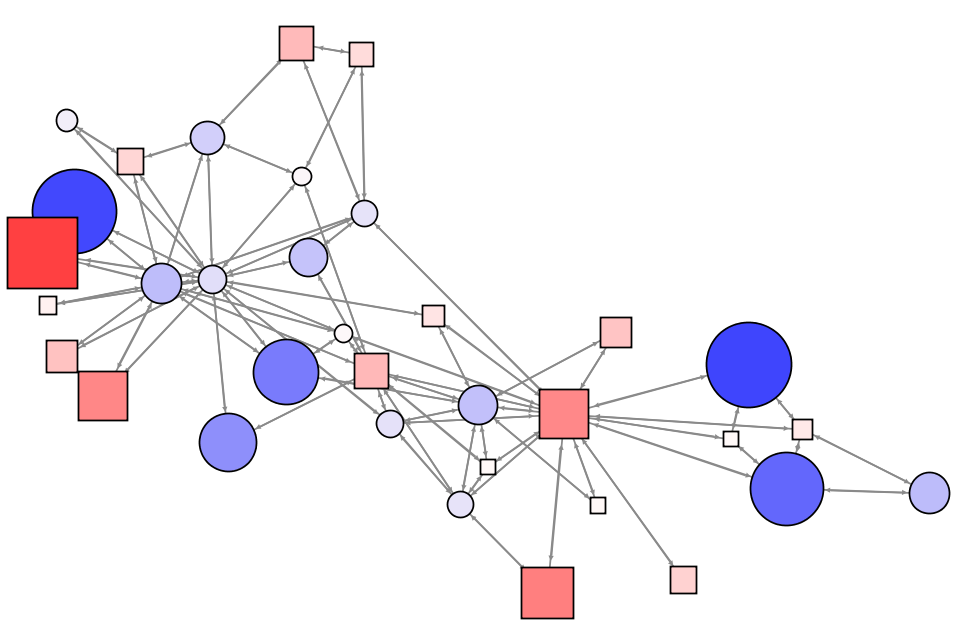} 
\\
&
&
(c) Final opinion values of nodes signified 
\\
(a) Investment made by  good camp 
&
(b) Investment made by  bad camp 
&
$ \! \! \! $ by  shapes, sizes,  color saturations for $t=2$ $ \! \! \! $
\\ 
on nodes signified by their sizes
&
on nodes signified by their sizes 
&
  (Circular blue nodes: positive opinions, 
\\
 and color saturations for $t=2$
 &
  and color saturations for $t=2$
  &
    square red nodes: negative opinions)
  \\
\hline
\hline &~ & \vspace{-4mm} \\
\includegraphics[width=0.32\textwidth]{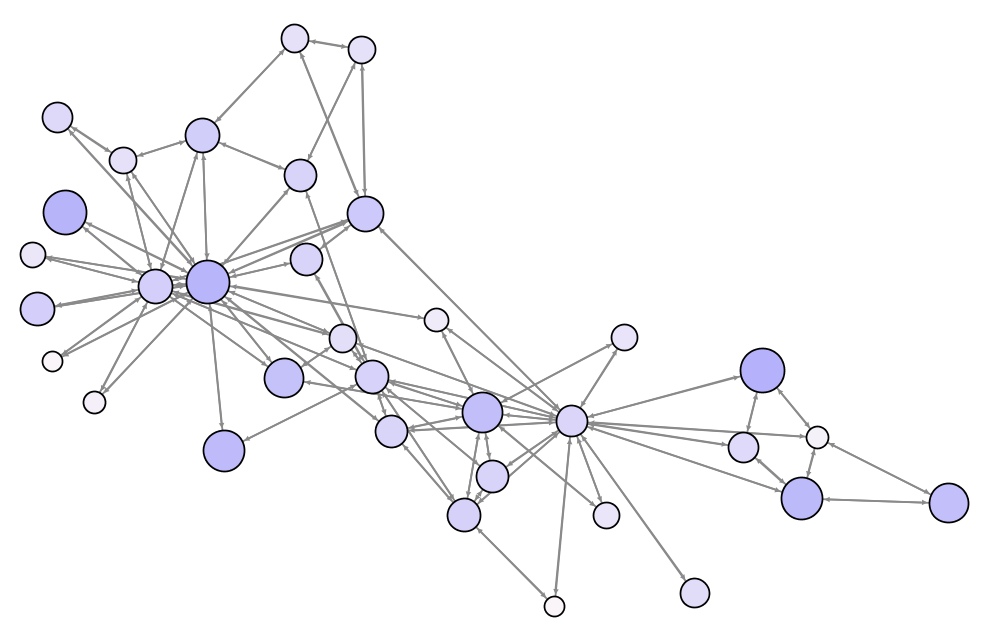} 
&
\includegraphics[width=0.32\textwidth]{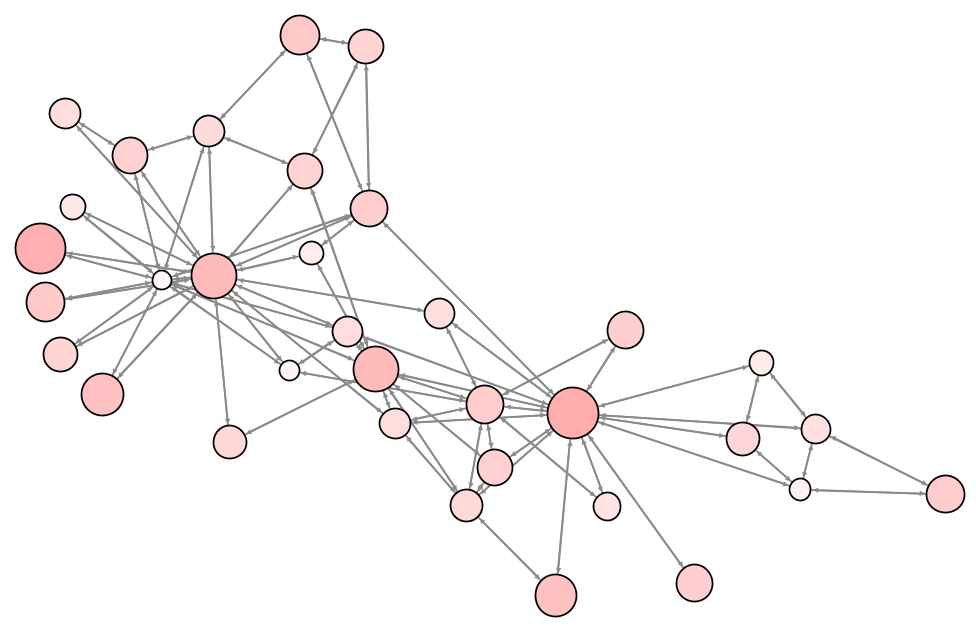}
&
\includegraphics[width=0.32\textwidth]{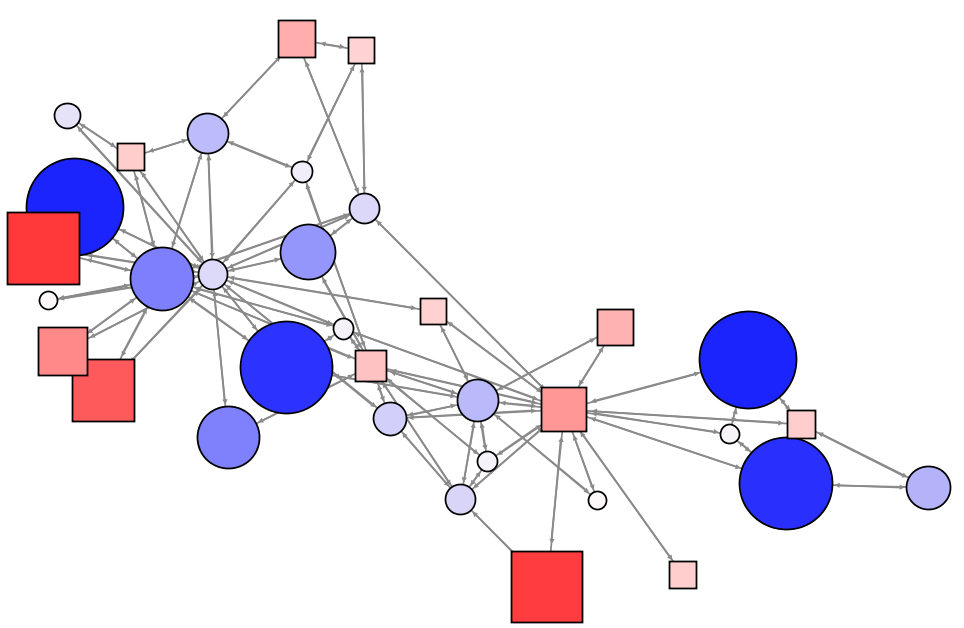}
\\
&&
(f) Final opinion values of nodes signified
\\
(d) Investment made by  good camp 
&
(e) Investment made by  bad camp 
&
$ \! \! \! $ by shapes, sizes, color saturations for $t=10$ $ \! \! \! $
\\ 
on nodes signified by their sizes 
&
on nodes signified by their sizes 
&
  (Circular blue nodes: positive opinions, 
\\
and color saturations for $t=10$
&
        and color saturations for $t=10$
        &
            square red nodes: negative opinions)
        \\
\hline
\end{tabular}
\caption{Simulation results for the Karate club dataset with $k_g=k_b=5$ when the influence function is concave 
($w_{ii}^0+w_{ig}+w_{ib}=0.9$)
}
\label{fig:concave_0.9}
\end{figure*}

\begin{figure*} 
\hspace{-7mm}
\begin{tabular}{|c||c||c|}
\hline 
 & & \\
\includegraphics[width=0.33\textwidth]{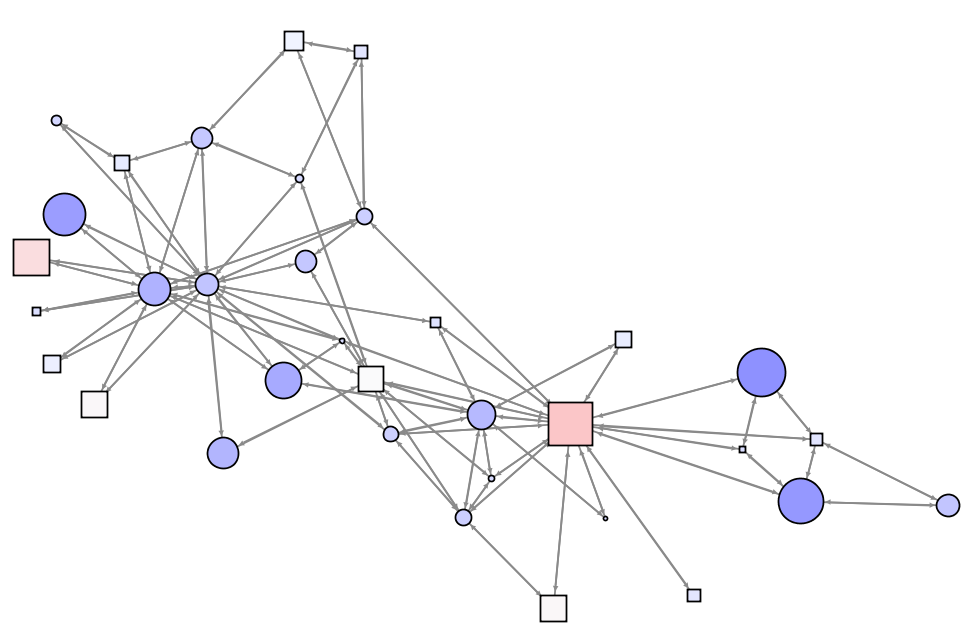} 
&
\includegraphics[width=0.33\textwidth]{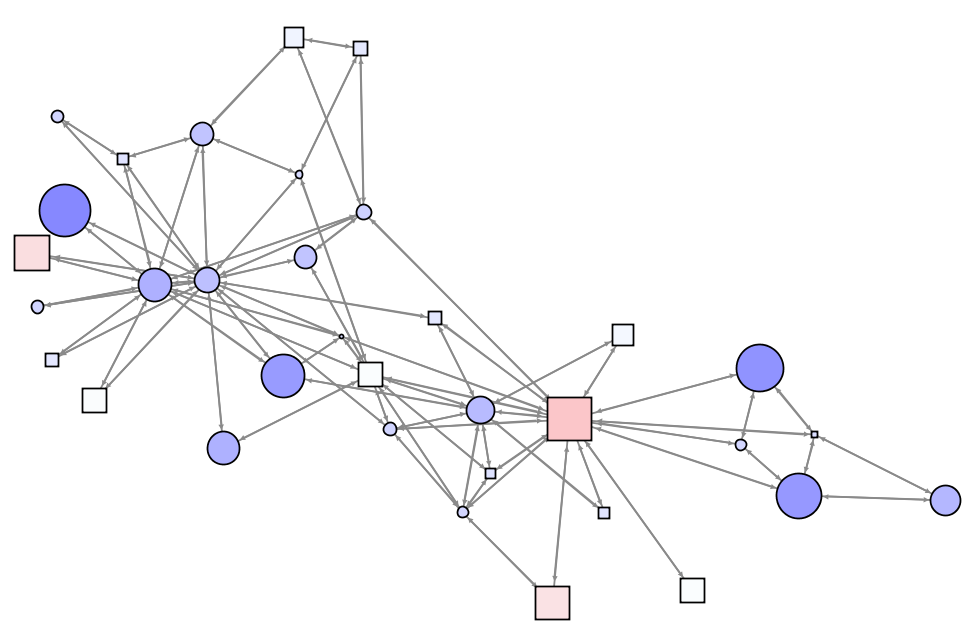}
&
\includegraphics[width=0.33\textwidth]{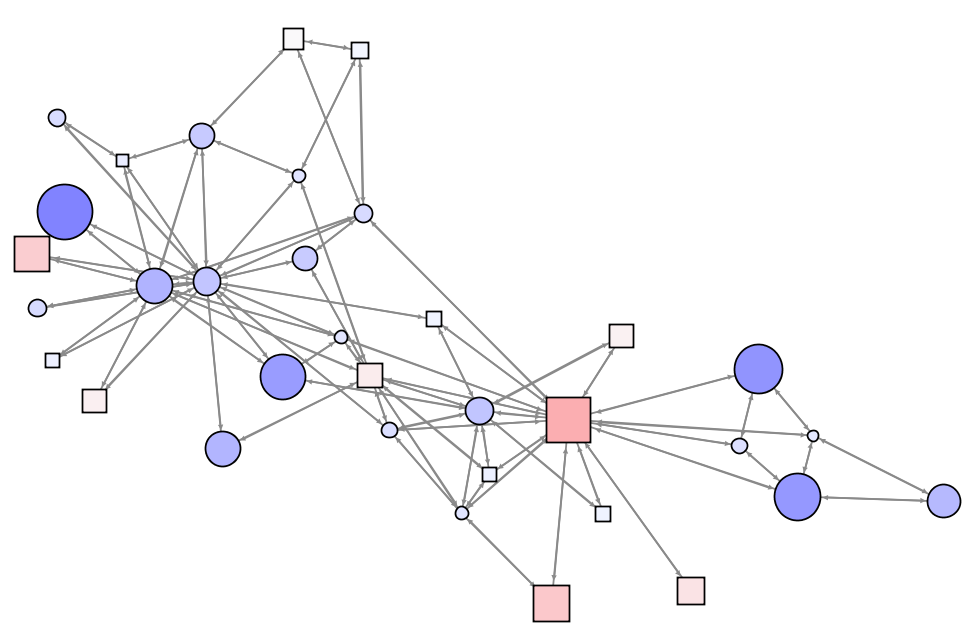} 
\\ & & \\ 
(a)
$\tau=1$ ($\max_{\mathbf{x}} \! \min_{\mathbf{y}} \! \sum_{i} v_i^{\langle \tau \rangle} \!=\! 0.3037$) 
&
(b)
$\tau=2$ ($\max_{\mathbf{x}} \! \min_{\mathbf{y}} \! \sum_{i} v_i^{\langle \tau \rangle} \!=\! 0.3681$)
&
(d)
$\tau=4$ ($\max_{\mathbf{x}} \! \min_{\mathbf{y}} \! \sum_{i} v_i^{\langle \tau \rangle} \!=\! 0.4353$)
\\
\hline
\end{tabular}
\caption{Progression of opinion values for the Karate club dataset with $k_g=k_b=5$ when the influence function is concave ($t=2$)
}
\label{fig:karate_prog_concave_2}
\vspace{-3mm}
\end{figure*}

\begin{figure*} 
\hspace{-7mm}
\begin{tabular}{|c||c||c|}
\hline 
 & & \\
\includegraphics[width=0.33\textwidth]{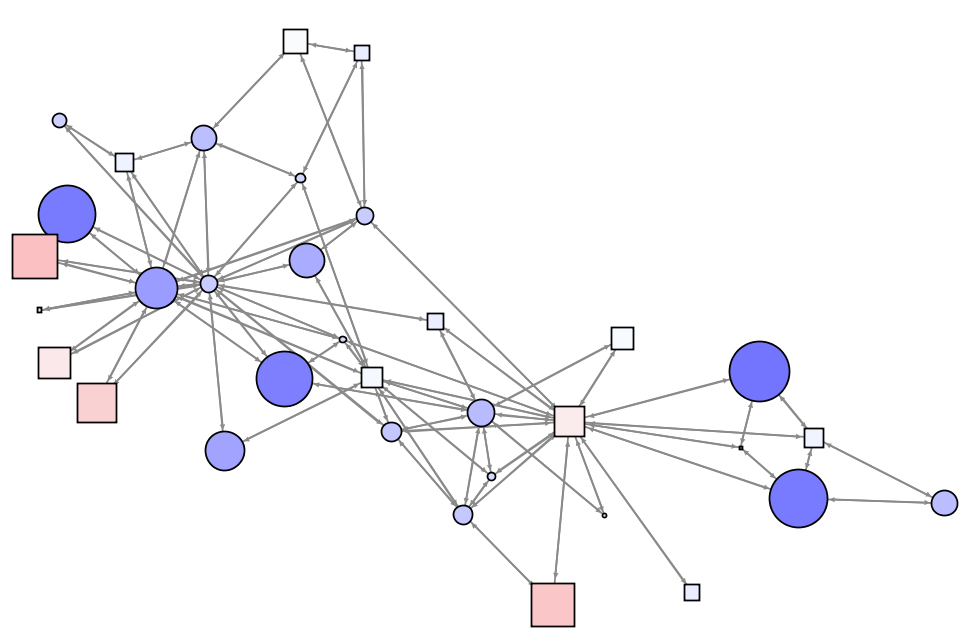} 
&
\includegraphics[width=0.33\textwidth]{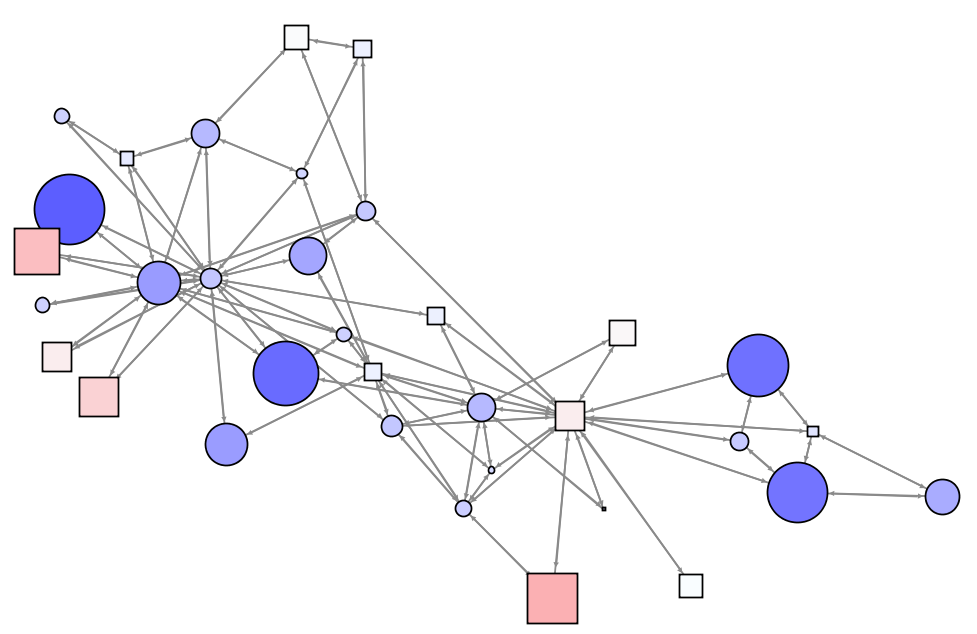}
&
\includegraphics[width=0.33\textwidth]{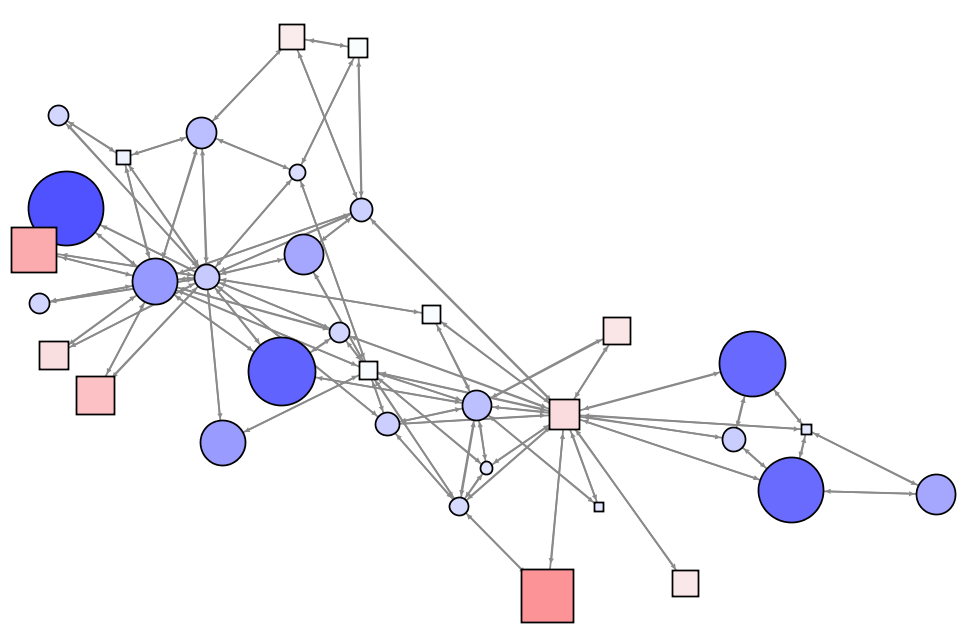} 
\\ & & \\ 
(a)
$\tau=1$ ($\max_{\mathbf{x}} \! \min_{\mathbf{y}} \! \sum_{i} v_i^{\langle \tau \rangle} \!=\! 0.4956$)
&
(b)
$\tau=2$ ($\max_{\mathbf{x}} \! \min_{\mathbf{y}} \! \sum_{i} v_i^{\langle \tau \rangle} \!=\! 0.7986$)
&
(d)
$\tau=4$ ($\max_{\mathbf{x}} \! \min_{\mathbf{y}} \! \sum_{i} v_i^{\langle \tau \rangle} \!=\! 1.0362$)
\\
\hline
\end{tabular}
\caption{Progression of opinion values for the Karate club dataset with $k_g=k_b=5$ when the influence function is concave ($t=10$)
}
\label{fig:karate_prog_concave_10}
\vspace{-3mm}
\end{figure*}

\begin{figure*} 
\hspace{-7mm}
\begin{tabular}{|c||c||c|}
\hline 
 & & \\
\includegraphics[width=0.33\textwidth]{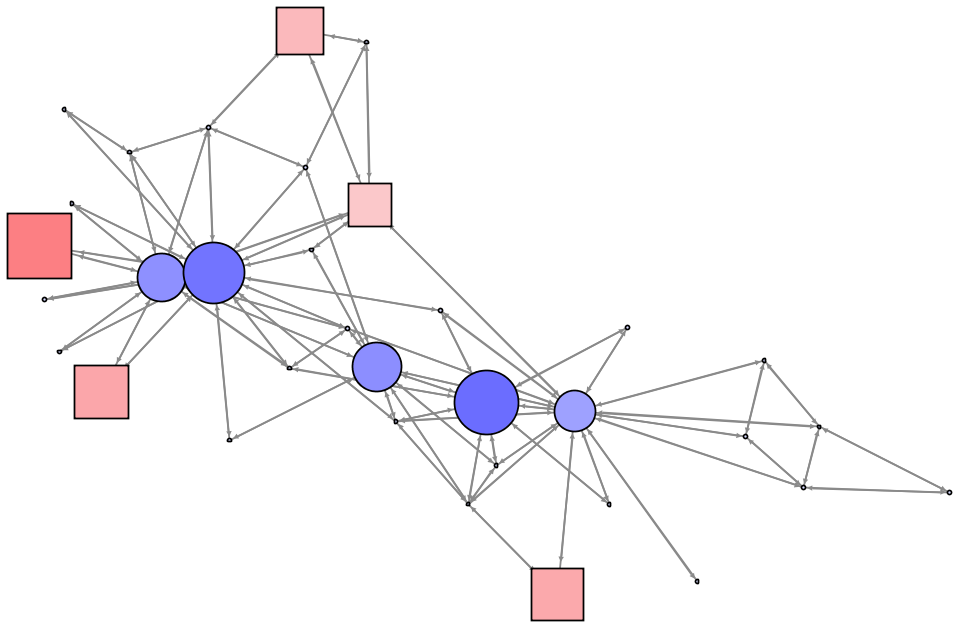} 
&
\includegraphics[width=0.33\textwidth]{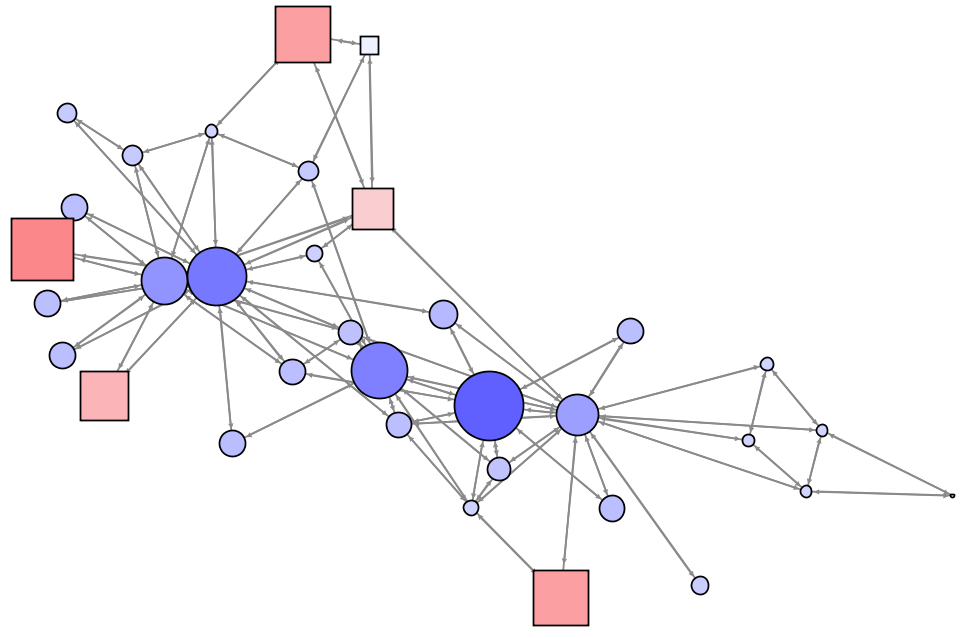}
&
\includegraphics[width=0.33\textwidth]{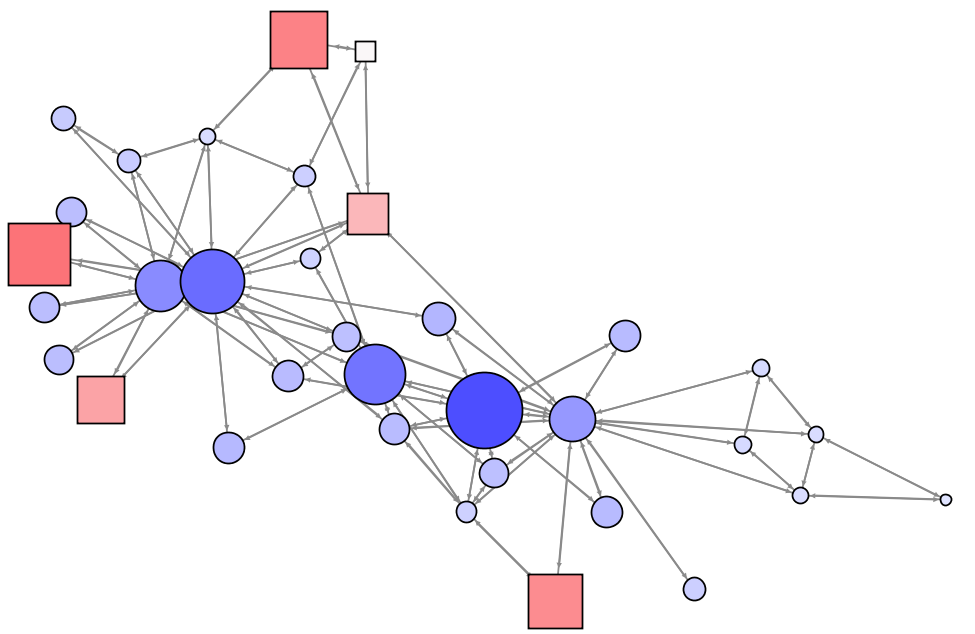} 
\\ & & \\ 
(a)
$\tau=1$ ($\max_{\mathbf{x}} \! \min_{\mathbf{y}} \! \sum_{i} v_i^{\langle \tau \rangle} \!=\! -0.2375$)
&
(b)
$\tau=2$ ($\max_{\mathbf{x}} \! \min_{\mathbf{y}} \! \sum_{i} v_i^{\langle \tau \rangle} \!=\! 0.8117$)
&
(d)
$\tau=4$ ($\max_{\mathbf{x}} \! \min_{\mathbf{y}} \! \sum_{i} v_i^{\langle \tau \rangle} \!=\! 1.3420$)
\\
\hline
\end{tabular}
\caption{Progression of opinion values for the Karate club dataset with $k_g=k_b=5$ under common coupled constraints (maxmin value)
}
\label{fig:karate_prog_maxmin}
\vspace{-3mm}
\end{figure*}

\begin{figure*} 
\hspace{-7mm}
\begin{tabular}{|c||c||c|}
\hline 
 & & \\
\includegraphics[width=0.33\textwidth]{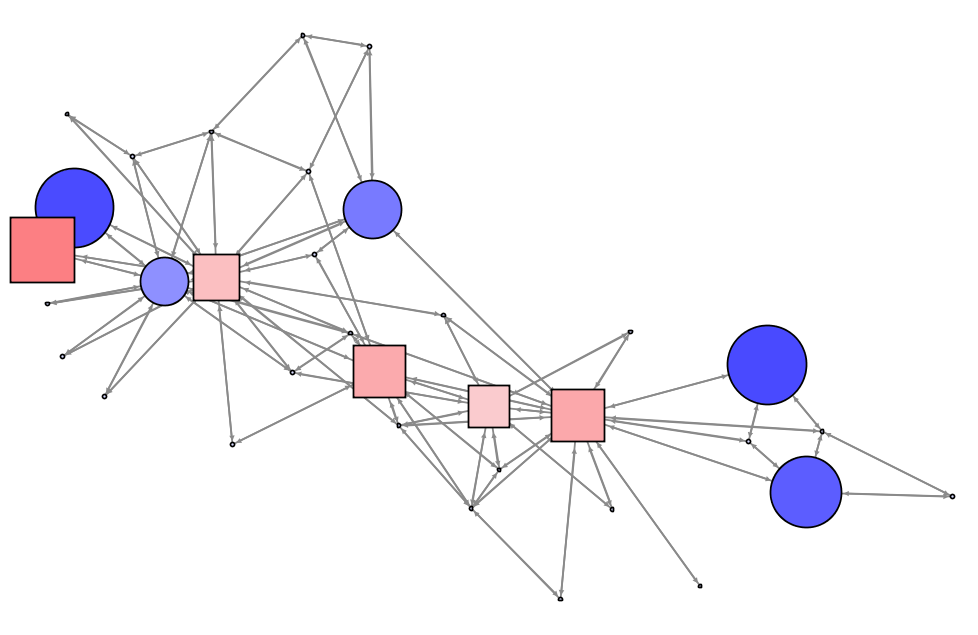} 
&
\includegraphics[width=0.33\textwidth]{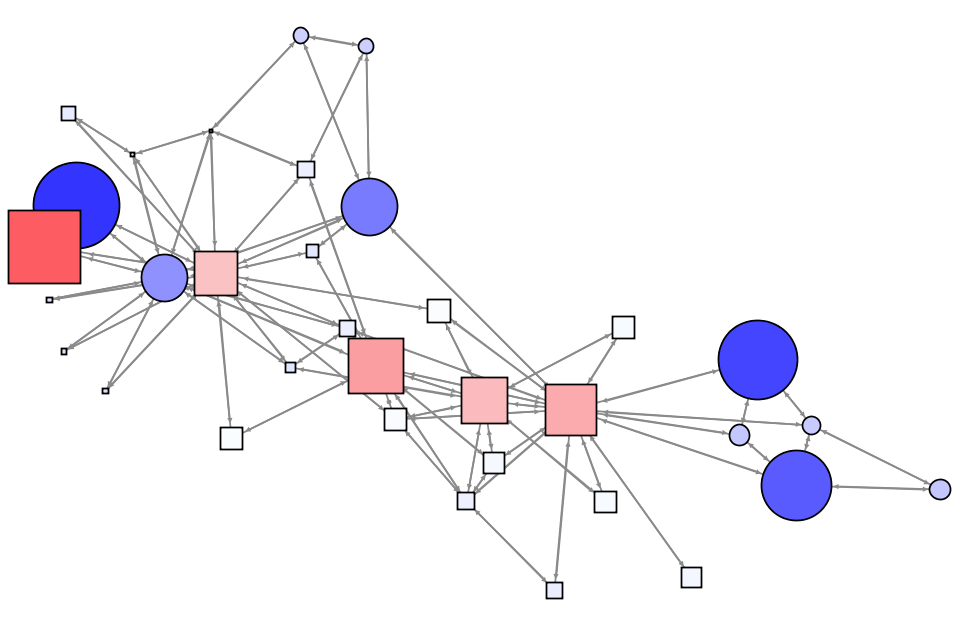}
&
\includegraphics[width=0.33\textwidth]{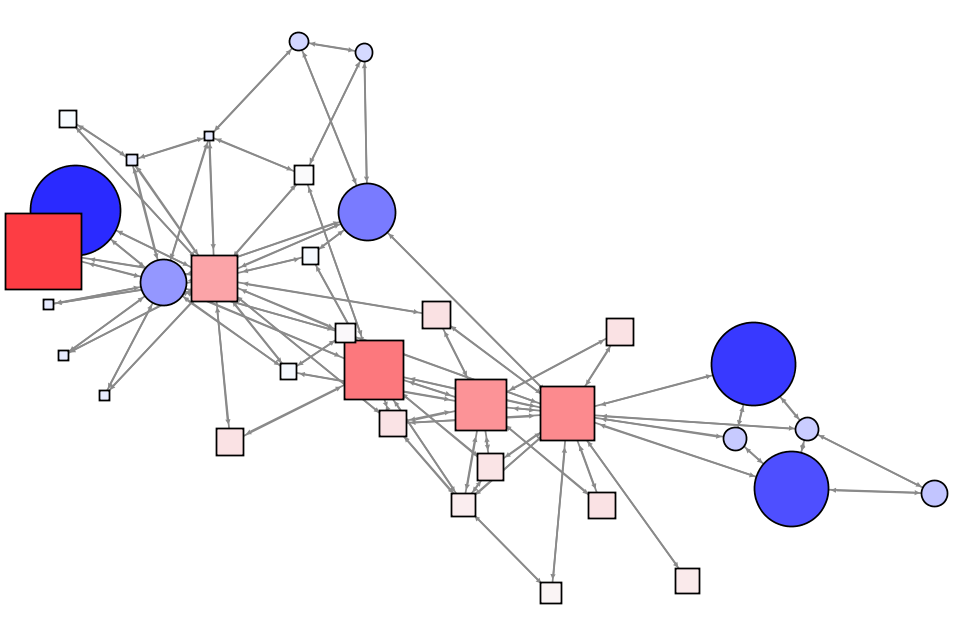} 
\\ & & \\ 
(a)
$\tau=1$ ($\min_{\mathbf{y}} \! \max_{\mathbf{x}} \! \sum_{i} v_i^{\langle \tau \rangle} \!=\! 0.1775$)
&
(b)
$\tau=2$ ($\min_{\mathbf{y}} \! \max_{\mathbf{x}} \! \sum_{i} v_i^{\langle \tau \rangle} \!=\! -0.4655$)
&
(d)
$\tau=4$ ($\min_{\mathbf{y}} \! \max_{\mathbf{x}} \! \sum_{i} v_i^{\langle \tau \rangle} \!=\! -0.7709$)
\\
\hline
\end{tabular}
\caption{Progression of opinion values for the Karate club dataset with $k_g=k_b=5$ under common coupled constraints (minmax value)
}
\label{fig:karate_prog_minmax}
\vspace{-3mm}
\end{figure*}

\end{document}